\documentclass[preprint,11pt]{imsart}
\usepackage{amsmath,amssymb,amsthm}
\usepackage{mathrsfs}
\usepackage[OT1]{fontenc}
\usepackage[utf8]{inputenc}
\usepackage[english]{babel}
\usepackage{arydshln}
\usepackage{indentfirst}
\usepackage{multirow}
\usepackage{float}
\usepackage{subfig}
\usepackage{bm}
\usepackage[dvips]{graphicx}
\usepackage{multicol}

\allowdisplaybreaks[4] 

\usepackage[nice]{nicefrac}
\usepackage{booktabs}

\usepackage[round,authoryear]{natbib}
\usepackage{a4wide}

\startlocaldefs
\numberwithin{equation}{section} 
\theoremstyle{plain}
\newtheorem{theorem}{Theorem}[section]
\newtheorem{lemma}{Lemma}[section]
\newtheorem{corollary}{Corollary}[section]
\newtheorem{proposition}{Proposition}[section]

\newtheorem{remark}{Remark}[section]
\endlocaldefs

\newcommand{\calX}{\mathcal{X}}
\newcommand{\calA}{\mathcal{A}}
\newcommand{\calB}{\mathcal{B}}
\newcommand{\calC}{\mathcal{C}}
\newcommand{\calN}{\mathcal{N}}
\newcommand{\calL}{\mathcal{L}}

\usepackage[mathscr]{euscript}

\newcommand{\scrU}{\mathscr{U}}

\newcommand{\scrC}{\mathscr{C}}

\newcommand{\bbZ}{\mathbb{Z}}

\newcommand{\bA}{\mathbf{A}}
\newcommand{\bB}{\mathbf{B}}
\newcommand{\bC}{\mathbf{C}}
\newcommand{\bD}{\mathbf{D}}
\newcommand{\bE}{\mathbf{E}}
\newcommand{\bF}{\mathbf{F}}

\newcommand{\bH}{\mathbf{H}}
\newcommand{\bI}{\mathbf{I}}

\newcommand{\bM}{\mathbf{M}}

\newcommand{\bS}{\mathbf{S}}
\newcommand{\bT}{\mathbf{T}}

\newcommand{\bX}{\mathbf{X}}
\newcommand{\bY}{\mathbf{Y}}
\newcommand{\bZ}{\mathbf{Z}}

\newcommand{\be}{\mathbf{e}}

\newcommand{\bh}{\mathbf{h}}

\newcommand{\bq}{\mathbf{q}}
\newcommand{\br}{\mathbf{r}}

\newcommand{\bv}{\mathbf{v}}

\newcommand{\bx}{\mathbf{x}}
\newcommand{\by}{\mathbf{y}}

\newcommand{\convas}{\overset{a.s.}{\longrightarrow}}

\newcommand{\convp}{\overset{p}{\longrightarrow}}
\newcommand{\convd}{\overset{d}{\longrightarrow}}

\newcommand{\bSigma}{\boldsymbol{\Sigma}}

\usepackage{dsfont}

\newcommand{\indicator}{\mathds{1}}
\newcommand\Cov{\mathrm{Cov}}
\newcommand\tr{\mathsf{tr}}
\newcommand\Expe{\mathbb{E}}
\newcommand\Prob{\mathbb{P}}
\newcommand{\Var}{\mathrm{Var}}

\newcommand\lb{\left(}
\newcommand\rb{\right)}

\newcommand{\dif}{\mathop{}\!\mathrm{d}}

\usepackage[colorlinks,citecolor=blue,urlcolor=blue,linkcolor=red]{hyperref}

\begin{document}

\begin{frontmatter}
  \title{Asymptotic normality for eigenvalue statistics of a general sample covariance matrix
    when $p/n \to \infty$ and applications}
    \runtitle{CLT for LSS}

\begin{aug}
\author{\fnms{Jiaxin} \snm{Qiu}\ead[label=e1]{qiujx@connect.hku.hk}},
\author{\fnms{Zeng} \snm{Li}\ead[label=e2]{liz9@sustech.edu.cn}}
\and
\author{\fnms{Jian-feng} \snm{Yao}\ead[label=e3]{jeffyao@hku.hk}}
\address{Department of Statistics and Actuarial Science,\\
The University of Hong Kong,\\}
\address{Department of Statistics and Data Science,\\
Southern University of Science and Technology,\\}
\end{aug}

\begin{abstract}
	The asymptotic normality for a large family of eigenvalue
    statistics of a  general sample covariance matrix  is derived   
    under the  ultra-high dimensional setting, that is, when  the dimension to
    sample size ratio $p/n \to \infty$.
    Based on this CLT result, we first adapt the 
    covariance matrix test problem to the new ultra-high dimensional context.
    Then as a second application, we develop a new test
    for the separable covariance structure of a matrix-valued  white
    noise.
    Simulation experiments are conducted
    for the investigation of  finite-sample properties of the general
    asymptotic normality of eigenvalue statistics, as well as the second
    test for separable covariance structure of matrix-valued white noise.
\end{abstract}
\end{frontmatter}

\section{Introduction}\label{sec:intro}

Let  $\by\in\mathbb{R}^p$ be a population of the form $\by=\bSigma_p^{\nicefrac{1}{2}}\bx$ where $\bSigma_p$ is a $p\times p$ positive definite  matrix, $\bx\in\mathbb{R}^p$ a $p$-dimensional random vector with independent and identically distributed (i.i.d.) components with zero mean and unit variance. Given an i.i.d. sample  $\left\{\by_j=\bSigma_p^{\nicefrac{1}{2}}\bx_j,~1\leq j\leq n\right\}$ of $\by$, the sample covariance matrix is $\bS_n=\frac{1}{n}\sum_{j=1}^n\by_j\by_j'=\frac{1}{n}\bSigma_p^{\nicefrac{1}{2}}\bX\bX'\bSigma_p^{\nicefrac{1}{2}}$, where $\bX=(\bx_1, \bx_2, \ldots, \bx_n)$. We consider the ultra-high dimensional setting where $n\to \infty, p=p(n)\to\infty$ such that $p/n\to \infty$. The $p\times p$ matrix $\bS_n$ has only a small number of non-zero eigenvalues which are the same as those of its $n\times n$ companion matrix $\underline{\bS}_n=\frac{1}{n}\bX'\bSigma_p\bX$. The limiting distribution of these non-zero eigenvalues is known (see \citet{bai1988convergence}, \citet{wang2014limiting}). Precisely, consider the re-normalized sample covariance matrix
\begin{equation}\label{eq:A_def}
\bA_n=\cfrac{1}{\sqrt{npb_p}}\lb  \bX'\bSigma_p \bX-pa_p  \bI_n\rb,
\end{equation}
where $\bI_n$ is the identity matrix of order $n$, $a_p = \frac{1}{p}\tr (\bSigma_p)$, $b_p = \frac{1}{p}\tr (\bSigma_p^2)$. Denote the eigenvalues of  $\bA_n$ as $\lambda_1,\cdots,\lambda_n$.   According to \citet{wang2014limiting}, under the condition that $\sup_p \|\bSigma_p\|<\infty$, the eigenvalue distribution of $\bA_n$, i.e. $F^{\bA_n}=\frac{1}{n}\sum_{i=1}^n\delta_{\lambda_i}$ converges to the celebrated semi-circle law. In this paper, we focus on the so-called linear spectral statistics (LSS) of $\bA_n$, i.e.  $\frac{1}{n}\sum_{i=1}^n f(\lambda_i)$ where $f(\cdot)$ is any smooth function we are interested in.
The main contribution of this paper is to establish the central limit theorem (CLT) for LSS of $\bA_n$ under the ultra-high dimensional setting. The study of fluctuations of LSS for different types of random matrix models has received extensive attention in the past decades, see monographs \citet{BSbook,couillet2011random,yao2015sample}. It plays a very important role in high dimensional data analysis because many well-established statistics can be represented as LSS of sample covariance or correlation matrix. In facing the curse of dimensionality, most asymptotic results are discussed under the \emph{Marchenko-Pastur} asymptotic regime, where $p/n\rightarrow c\in (0,\infty)$. However, this doesn't fit the case of ultra-high dimension when $p\gg n$. Hence in this paper we re-examined the asympototic behavior of LSS of $\bA_n$ when $n\to \infty, p=p(n)\to\infty$ such that $p/n\to \infty$.

A special version of $\bA_n$ for the case where $\bSigma_p=\bI_p$ has already been studied in the literature. The matrix becomes
\begin{equation}\label{eq:A_def_iden}
	\bA_n^{\mathsf{iden}}=\cfrac{1}{\sqrt{np}}\lb  \bX' \bX- p\bI_n\rb.
\end{equation}
\citet{bai1988convergence} is the first  to study this matrix. They proved that the ultra-high  dimensional limiting eigenvalue distribution of $\bA_n^{\mathsf{iden}}$ is the semi-circle law.
\citet{chen2012convergence} studied the behavior of the largest eigenvalue of $\bA_n^{\mathsf{iden}}$.
\citet{chen2015clt} and \citet{bao2015asymptotic} independently established the CLT for LSS of $\bA_n^{\mathsf{iden}}$, the limiting variance function of which coincides with that of a Wigner matrix given in \citet{BaiYao2005}. From these results, we know that some spectral properties of $\bA_n^{\mathsf{iden}}$ are similar to those of a $n\times n$ Wigner matrix. 
Indeed, the general matrix $\bA_n$ also has some spectral properties similar to those of a Wigner matrix. In particular, the eigenvalue distribution of $\bA_n$, $F^{\bA_n}$, also converges to the semi-circle law. However, the second order fluctuations for LSS of $\bA_n$ are quite different and worth further investigation.

In this paper, we establish the CLT for LSS of $\bA_n$. The general strategy of the proof follows that of \citet{BaiYao2005} for the CLT for LSS of a large Wigner matrix. However, the calculations are more involved here as the matrix $\bA_n$ is a quadratic function of the independent entries $(X_{ij})$ while a Wigner matrix is a linear function of its entries. Similar to \citet{chen2015clt}, a key step is to establish the CLT for some smooth integral  of the Stieltjes transform $M_n(z)$ of $\bA_n$, see Proposition~\ref{prop:Mn_CLT}.
To derive the limiting mean and covariance functions, we divide $M_n(z)$ into two parts: a non-random part and a random part. Our approaches to handle these two parts are technically different. For the random part, we follow a method in \citet{chen2015clt} which depends heavily on an explicit expression for $\tr\bigl(\bM_k^{(1)}\bigr)/(npb_p)$ (see Section~\ref{sec:conv_Mn1_proof} for more details). This explicit expression does not exist in our matrix model, so we need to provide a first-order approximation for it, which is given in Lemma~\ref{lem:Mk_limit}. For the non-random part, we
utilize the generalized Stein's equation to find the asymptotic expansion of the expectation of Stieltjes transform, which provides some new enlightenment for conventional procedures. 

To demonstrate the potential of our newly established CLT, we further studied two hypothesis testing problems about population covariance matrices. First, we examine the identity hypothesis $H_0:~\bSigma_p=\bI_p$ under the ultra-high dimensional setting and compare it with cases of relatively low dimensions. Next, we consider the hypothesis that a matrix-valued noise has a separable covariance matrix. 
For a sequence of i.i.d. $p_1 \times p_2 $ matrices $\{\bE_t\}_{1\leqslant t\leqslant T}$,  we adopt a Frobenius-norm-type statistic to test whether the covariance matrix of $\mathsf{vec}(\bE_t)$ is separable, i.e. $\Cov(\mathsf{vec}(\bE_t))=\bSigma_1\otimes \bSigma_2$, where $\bSigma_1$ and $\bSigma_2$ are two given  $p_1\times p_1$ and $p_2\times p_2$ nonnegative definite matrices. Here $p_1, p_2$ and $T$ are of comparable magnitude. Our test statistic can be represented as a  LSS of the sample covariance matrix with dimension $p_1p_2$ much larger than the sample size $T$. Therefore our CLT can then be employed to derive the null distribution and perform power analysis of the test. Good numerical performance lends full support to the correctness of our CLT results.

The paper is organized as follows. Section~\ref{sec:pre_notation} provides preliminary knowledge of some technical tools. Section~\ref{sec:main_result} establishes our main CLT for LSS of $\bA_n$. Section~\ref{sec:application_hypothesis_testing} contains two hypothesis testing applications. Section~\ref{sec:simulation} reports numerical studies. Technical proofs and lemmas are relegated to Section~\ref{sec:proof_CLT_1} and Appendices.

Throughout the paper, we reserve boldfaced symbols for vectors and matrices. For any matrix $\bA$, we let $A_{ij}$, $\lambda_j^{\bA}$, $\bA'$, $\tr(\bA)$ and $\|\bA\|$ represent, respectively, its $(i, j)$-th element, its $j$-th largest eigenvalue, its transpose, its trace, and its spectral norm (i.e., the largest singular value of $\bA$). $\indicator_{\{\cdot\}}$ stands for the indicator function. For the random variable $X_{11}$, we denote the $a$-th moment of $X_{11}$ by $\nu_{a}$ and the $a$-th cumulant of $X_{11}$ by $\kappa_a$. 
We use $K$ to denote constants which may vary from line to line. For simplicity,  we sometimes omit the variable $z$ when representing some matrices and functions (e.g. Stieltjes transforms) of $z$, provided that it does not lead to confusion.

\section{Preliminaries}\label{sec:pre_notation}
In this section, we introduce some useful preliminary results. 
For any $n\times n$ Hermitian matrix $\bB_n$, its empirical spectral distribution (ESD) is defined by
\[
	F^{\bB_n}(x) = \frac{1}{n}\sum_{i=1}^n \indicator_{\{\lambda_i^{\bB_n} \leqslant x\}}.
\]
 If $F^{\bB_n}(x)$ converges to a non-random limit $F(x)$ as $n\to\infty$, we call $F(x)$ the limiting spectral distribution (LSD) of $\bB_n$. 

As for the LSD of $\bA_n$ defined in \eqref{eq:A_def}, 
\cite{wang2014limiting} derived the LSD of re-normalized sample covariance matrices with  more generalized form
\begin{equation}\label{eq:Cp_wang2014limiting}
	\bC_n = \sqrt{\frac{p}{n}} \biggl( \frac{1}{p} \bT_{n}^{1/2}\bX_n^*\bSigma_{p}\bX_n\bT_{n}^{1/2} - \frac{1}{p} \tr(\bSigma_{p})\bT_{n} \biggr),
\end{equation}
where $\bX_n$ and $\bSigma_p$ are the same as those in \eqref{eq:A_def}. $\bT_{n}$ is a $n\times n$ nonnegative definite Hermitian matrix, whose ESD, $F^{\bT_{n}}$, converges weakly to $H$, a nonrandom distribution function on $\mathbb{R}^{+}$ which does not degenerate to zero.
The LSD of $\bC_n$ is described in terms of its Stieltjes transform. The Stieltjes transform of any cumulative distribution function $G$ is defined by
\[
	m_G(z) = \int \frac{1}{\lambda-z} \dif G(\lambda),\qquad z\in\mathbb{C}^{+}:=\{u+iv,\, u\in\mathbb{R}, v>0\}.
\]
\cite{wang2014limiting} proved that, when $p\wedge n\to\infty$ and $p/n\to\infty$, $F^{\bC_n}$ almost surely converges  to a nonrandom distribution, whose Stieltjes transform $m_{\bC}(z)$ satisfies the following system of equations:
\begin{equation}\label{eq:wang2014limiting_LSD}
	\begin{cases}
		m_{\bC}(z)= -\int\frac{\dif H(x)}{z+x\theta g(z)},\\[0.5em]
		g(z)= -\int\frac{x \dif H(x)}{z+x\theta g(z)},
	\end{cases}
\end{equation}
for any $z\in\mathbb{C}^{+}$, where $\theta=\lim\limits_{p\to\infty}(1/p)\tr(\bSigma_{p}^2)$.


Note that $\bA_n$ is a special case of $\bC_n$ with $\bT_n=\bI_n$. By \eqref{eq:wang2014limiting_LSD} we can easily show that the Stieltjes transform $m_{\bA}(z)$ of LSD of $\bA_n$ satisfies 
\begin{equation}\label{eq:SCLaw_equation}
	m_{\bA}(z)=-\frac{1}{z+m_{\bA}(z)},
\end{equation}
which is exactly the Stieltjes transform of the semi-circle law with density function given by
\begin{equation}\label{eq:semi-circle-law}
	F'(x) =\frac{1}{2\pi}\sqrt{4-x^2}\,\indicator_{\{ |x|\leqslant 2\}}.
\end{equation}
Hereafter we use $m(z)$ to represent $m_{\bA}(z)$ for ease of presentation.

\section{Main Results}\label{sec:main_result}

Let $\mathscr{U}$ denote any open region on the complex plane including $[-2,2]$ and $\mathscr{M}$ be the set of analytic functions defined on $\mathscr{U}$.
For any $f\in \mathscr{M}$, we consider a LSS of $\bA_n$ of the form:
\[
	\int f(x) \dif F^{\bA_n}(x) = \frac{1}{n} \sum_{i=1}^n f(\lambda_i^{\bA_n}).
\]
Since $F^{\bA_n}$ converges to $F$ almost surely, we have
\[
	\int f(x) \dif F^{\bA_n}(x) \to	\int f(x) \dif F(x).
\]
A question naturally arises: how fast does  
$\int f(x) \dif\left\{F^{\bA_n}(x) -F(x)\right\} $ converge to zero?

To answer this question, we consider a re-normalized functional:
\begin{equation}\label{eq:LSS}
G_n(f)= n\int_{-\infty}^{+\infty}f(x) \mathrm{d} \left\{ F^{\bA_n}(x)-F(x)\right\}-\frac{n}{2\pi i}\oint_{|m|=\rho}f(-m-m^{-1})\calX_n(m)\frac{1-m^2}{m^2}\dif m,
\end{equation}
where 
\begin{equation}\label{eq:calX}
	\begin{split}
		\mathcal{X}_n(m) &= \frac{-\calB+\sqrt{\calB^2-4\calA\calC}}{2\calA},\qquad \calA = m-\sqrt{\frac{n}{p}}\frac{c_p}{b_p\sqrt{b_p}}\bigl(1+m^2\bigr),	\\[0.5em]
		 \calC &= \frac{m^3}{n}\biggl\{ \frac{1}{1-m^2}+\frac{(\nu_4-3)\tilde{b}_p}{b_p} \biggr\}  - \sqrt{\frac{n}{p}}\frac{c_p}{b_p\sqrt{b_p}} m^4 +\frac{n}{p}\biggl(-\frac{c_p^2}{b_p^3}  +\frac{d_p}{b_p^2}\biggr)m^5,\\[0.5em]
		 \calB &=m^2-1-\sqrt{\frac{n}{p}}\frac{c_p}{b_p\sqrt{b_p}}m(1+2m^2),~a_p = \frac{1}{p}\tr (\bSigma_p)\\[0.5em]
	 b_p &= \frac{1}{p}\tr (\bSigma_p^2),~\tilde{b}_p = \frac{1}{p}\sum_{i=1}^p \sigma_{ii}^2,~c_p  =\frac{1}{p}\tr (\bSigma_p^3), ~ d_p  = \frac{1}{p}\tr (\bSigma_p^4),
	\end{split}
\end{equation}
here $\rho < 1$ and $\sqrt{\calB^2-4\calA\calC}$ is a complex number whose imaginary part has the same sign as that of $\calB$.
In what follows, we have established the asymptotic normality of $G_n(f)$ and the main result is formulated in the  theorem below.

\begin{theorem}\label{thm:CLT}
  Suppose that
  \begin{enumerate}
    \item[(A)]\label{itm:asum1}  ${\bf X}=(X_{ij})_{p\times n}$ where $\{X_{ij},~1\leqslant i \leqslant p,~ 1\leqslant j\leqslant n\}$ are i.i.d. real random variables with $\Expe X_{ij}=0$, $\Expe X_{ij}^2=1$, $\Expe X_{ij}^4=\nu_4$ and $\Expe |X_{ij}|^{6+\varepsilon_0}<\infty$ for some small positive $\varepsilon_0$;
    \item[(B)]\label{itm:asum2} $\{\bSigma_p,~p\geq 1\}$ is a sequence of non-negative definite matrices, bounded in spectral norm, such that the following limits exist:
	\begin{itemize}
		\item $\gamma=\lim_{p\to \infty} \frac{1}{p}\tr (\bSigma_p)$,
		\item $\theta=\lim_{p\to \infty} \frac{1}{p}\tr (\bSigma_p^2)$,
		\item $\omega=\lim_{p\to \infty} \frac{1}{p}\sum_{i=1}^p \sigma_{ii}^2$;
	\end{itemize}
    \item[(C1)] \label{item:np0} $p\wedge n\to\infty$ and $n^2/p=O(1)$.
  \end{enumerate}
  Then, for any $f_1,\cdots,f_k\in \mathscr{M}$, the finite dimensional random vector $\lb G_n(f_1),\cdots,G_n(f_k)\rb$ converges weakly to a Gaussian vector $\lb Y(f_1),\cdots, Y(f_k)\rb$ with mean function $\Expe Y(f) = 0$
   and covariance function
  \begin{align}
    \Cov\lb Y(f_1), Y(f_2)\rb&=\frac{\omega}{\theta}(\nu_4-3)\Psi_1(f_1)\Psi_1(f_2)+2\sum_{k=1}^{\infty}k\Psi_k(f_1)\Psi_k(f_2)\label{eq:cov_1}\\
    &=\frac{1}{4\pi^2}\int_{-2}^2\int_{-2}^2f_1'(x)f_2'(y)H(x,y)\dif x \dif y,\label{eq:cov_2}
  \end{align}
  where
  \begin{equation}\label{eq:Phi_k}
	  \Psi_k(f)=\cfrac{1}{2\pi}\int_{-\pi}^{\pi}f(2\cos\theta)e^{ik\theta}\dif \theta=\cfrac{1}{2\pi}\int_{-\pi}^{\pi}f(2\cos\theta)\cos k\theta \dif \theta,
  \end{equation}
  \begin{equation}\label{eq:Hxy}
  	H(x,y)=\frac{\omega}{\theta}(\nu_4-3)\sqrt{4-x^2}\sqrt{4-y^2}+2\log\lb \cfrac{4-xy+\sqrt{(4-x^2)(4-y^2)}}{4-xy-\sqrt{(4-x^2)(4-y^2)}}\rb.
  \end{equation}
\end{theorem}
The proofs of Theorem \ref{thm:CLT} is postponed to Section \ref{sec:proof_CLT_1}.

\begin{remark}
	Note that we require $p\geqslant Kn^2$ asymptotically in Assumption (C1), while the situation of $n\ll p \ll n^2$ remains unknown.
\end{remark}

\begin{remark}
	If $\bSigma_p = \bI_p$, we have $a_p=b_p=\tilde{b}_p=c_p=d_p=1$ and $\gamma=\theta=\omega=1$. Our Theorem~\ref{thm:CLT} reduces to the CLT derived in \citet{chen2015clt}.
\end{remark}

Applying Theorem~\ref{thm:CLT} to three polynomial functions, we obtain the following corollary. 

\begin{corollary}\label{coro:CLT-x-x2-x3}
	With the same notations and assumptions given in Theorem \ref{thm:CLT}, consider three analytic functions $f_1(x)=x, f_2(x)=x^2, f_3(x)=x^3$, we have
	\begin{align*}
		G_n(f_1)  &= \tr(\bA_n) \convd \calN\Bigl(0, \frac{\omega}{\theta}(\nu_4-3)+2\Bigr);\\[0.5em]
		G_n(f_2)  &= \tr(\bA_n^2) - n - \Bigl\{\frac{\tilde{b}_p}{b_p}(\nu_4-3)+1\Bigr\} \convd \calN(0, 4);\\[0.5em]
		G_n(f_3)  &= \tr(\bA_n^3) - \frac{c_p}{b_p\sqrt{b_p}}\sqrt{\frac{n}{p}}\Bigl\{ n+1+\frac{\tilde{b}_p}{b_p}(\nu_4-3) \Bigr\} \convd \calN\Bigl(0, \frac{9\omega}{\theta}(\nu_4-3)+24\Bigr).
	\end{align*}
\end{corollary}

The calculations in these applications are elementary, thus omitted. Note that the mean correction terms for $G_n(f_1)$, $G_n(f_2)$, and $G_n(f_3)$ are $0$, $\frac{\tilde{b}_p}{b_p}(\nu_4-3)+1$, and $\frac{c_p}{b_p\sqrt{b_p}}\sqrt{\frac{n}{p}}\bigl\{ n+1+\frac{\tilde{b}_p}{b_p}(\nu_4-3) \bigr\}$, respectively.

\subsection{Case of $p\geqslant Kn^3$}

When $p\geqslant Kn^3$, the mean correction term in \eqref{eq:LSS} can be further simplified, i.e.
\begin{align}
	&\;-\frac{n}{2\pi i}\oint_{|m|=\rho}f(-m-m^{-1})\calX_n(m)\frac{1-m^2}{m^2}\dif m \nonumber\\
	= &\; -\biggl[\frac{1}{4}\bigl( f(2)+f(-2)\bigr)-\frac{1}{2}\Psi_0(f)+\frac{\tilde{b}_p}{b_p}(\nu_4-3)\Psi_2(f)\biggr] - \sqrt{\frac{n^3}{p}}\frac{c_p}{b_p\sqrt{b_p}}\Psi_3(f) + o(1).\label{eq:mean-correction-simplify}
\end{align}
For any function $f\in\mathscr{M}$, we define a new normalization of the LSS:
\begin{equation}\label{eq:LSS_n3}
	Q_{n}(f)= n\int_{-\infty}^{+\infty}f(x) \mathrm{d} \left\{ F^{\bA_n}(x)-F(x)\right\}-\sqrt{\frac{n^3}{p}}\frac{c_p}{b_p\sqrt{b_p}}\Psi_3(f).
\end{equation}
Note that the last term in \eqref{eq:LSS_n3} makes no contribution if the function $f$ is even ($\Psi_3(f)=0$) or $n^3/p = o(1)$. Substituting \eqref{eq:mean-correction-simplify} into Theorem \eqref{thm:CLT}, we obtain the following CLT for $Q_n(f)$.

\begin{corollary}\label{coro:CLT_pn3}
	Under assumptions $(A)$, $(B)$ in Theorem~\eqref{thm:CLT} and
	\begin{enumerate}
	    \item[(C2)] $p\wedge n\to\infty$ and $n^3/p = O(1)$.
  	\end{enumerate}
 for any $f_1,\cdots,f_k\in \mathscr{M}$, the finite dimensional random vector $\lb Q_n(f_1),\cdots,Q_n(f_k)\rb$ converges weakly to a Gaussian vector $\lb Y(f_1),\cdots, Y(f_k)\rb$ with mean function
  \begin{equation}\label{eq:CLT-mean}
  	\Expe Y(f_k)=\frac{1}{4}\bigl( f_k(2)+f_k(-2)\bigr)-\frac{1}{2}\Psi_0(f_k)+\cfrac{\omega}{\theta}(\nu_4-3)\Psi_2(f_k)
  \end{equation}
   and covariance function given in \eqref{eq:cov_1}.
\end{corollary}

\begin{remark}
As a special case of Theorem~\ref{thm:CLT}, Corollary \ref{coro:CLT_pn3} is used in \citet{li2016testing} to derive the asymptotic power of two sphericity tests, John's invariant test and Quasi-likelihood ratio test (QLRT), when the dimension $p$ is much larger than sample size $n$.  
Specifically, let $\bX=(\bx_1,\ldots,\bx_n)$ be a $p\times n$ data matrix with $n$ i.i.d. $p$-dimensional random vectors $\{\bx_i\}_{1\leqslant i \leqslant n}$ with covariance matrix $\bSigma=\Var(\bx_i)$. The goal is to test
\[
	H_0: \bSigma = \sigma^2\bI_p,\qquad \text{vs.} \qquad  H_1: \bSigma \neq \sigma^2\bI_p,
\]
where $\sigma^2$ is an unknown positive constant. John's test statistic is defined by 
\[
	U = \frac{1}{p}\tr \biggl[\biggl(\frac{\bS_n}{\tr(\bS_n)/p}-\bI_p\biggr)^2\biggr] = \frac{p^{-1}\sum_{i=1}^p (l_i-\bar{l})^2}{\bar{l}^2},
\]
where $\{l_i\}_{1\leqslant i \leqslant p}$ are eigenvalues of $p$-dimensional sample covariance matrix $\bS_n=\frac{1}{n}\sum_{i=1}^n \bx_i\bx_i'=\frac{1}{n}\bX\bX'$ and $\bar{l}=\frac{1}{p}\sum_{i=1}^p l_i$.
The QLRT statistic is defined by 
\[
	\mathcal{L}_n = \frac{p}{n}\log \frac{(n^{-1}\sum_{i=1}^n \tilde{l}_i )^n}{\prod_{i=1}^n \tilde{l}_i},
\]
where $\{\tilde{l}_i\}_{1\leqslant i \leqslant n}$ are the eigenvalues of the $n\times n$ matrix $\frac{1}{p}\bX'\bX$. The main idea is that both $U$ and $\mathcal{L}_n$ can be expressed as functions of eigenvalues of $\bA_n$ in \eqref{eq:A_def}. Thus, asymptotic distributions of John's statistic and QLRT statistic can be derived either using Theorem~\ref{thm:CLT} or Corollary~\ref{coro:CLT_pn3}. \citet{li2016testing} used Corollary \ref{coro:CLT_pn3} to derive the limiting distributions of $U$ and $\mathcal{L}_n$ under the alternative hypothesis. Their power functions are proven to converge to $1$ under the assumption $n^3/p=O(1)$. More details can be found in \citet{li2016testing}. 
\end{remark}

\section{Applications to Hypothesis Testing about Large Covariance Matrices}\label{sec:application_hypothesis_testing}

	\subsection{The Identity Hypothesis ``$\Sigma_p = \bI_p$''}\label{sec:identity_test}

	Let $\bY=(\by_1,\ldots,\by_n)$ be a $p\times n$ data matrix with $n$ i.i.d. $p$-dimensional random vectors $\{\by_i=\bSigma_p^{1/2}\bx_i\}_{1\leqslant i \leqslant n}$ with covariance matrix $\bSigma_p=\Var(\by_i)$ and $\bx_i$ has $p$ i.i.d. components $\{X_{ij},~1\leq j\leq p\}$ satisfying $\Expe X_{ij}=0$, $\Expe X_{ij}^2=1$, $\Expe X_{ij}^4=\nu_4$. We explore the identity testing problem
	\begin{equation}\label{eq:identity_test}
		H_0: \bSigma_p = \bI_p,\qquad \text{vs.} \qquad H_1: \bSigma_p \ne \bI_p,
	\end{equation}
	under two different asymptotic regimes: high-dimensional regime, ``$p\wedge n\to\infty,~p/n\to c \in(0,\infty)$'' and ultra-high dimensional regime, ``$p\wedge n\to\infty,~p/n\to\infty$''.
	We will consider two well-known test statistics and discuss their limiting distributions under both regimes.
	
	For the identity testing problem \eqref{eq:identity_test}, \citet{Nagao1973On} proposed a statistic based on the Frobenius norm:
	\[
		V = \frac{1}{p} \tr\bigl[ (\bS_n-\bI)^2 \bigr],
	\]
	where $\bS_n=\tfrac{1}{n}\bY\bY'$ is the sample covariance matrix.
	Nagao's test based on $V$ performs well when $n$ tends to infinity while $p$ remains fixed. However, \citet{LedoitWolf2002} showed that Nagao's test has poor properties when $p$ is large.
 	They made some modifications as
 	\begin{equation}\label{eq:W_def}
		W = \frac{1}{p}	\tr\Bigl[(\bS_n-\bI_p)^2\Bigr] - \frac{p}{n}\biggl[\frac{1}{p}\tr(\bS_n)\biggr]^2 +\frac{p}{n}.
	\end{equation}
	When $p\wedge n\to\infty, p/n=c_n\to c\in(0, \infty)$, under normality assumption, \citet{LedoitWolf2002} proved that the limiting distribution of $W$ under $H_0$ is
	\[
		nW-p-1 \convd \calN(0,4).
	\]
	\citet{wang2013sphericity} further removed the normality assumption and show that under $H_0$, when $p\wedge n\to\infty, p/n=c_n\to c\in(0, \infty)$,
	\begin{equation}\label{eq:W_null_dist_high}
		nW-p-(\nu_4-2) \convd \calN(0,4).
	\end{equation}
	
	Now we derive the limiting distribution of $W$ under both $H_0$ and $H_1$ when $p/n\rightarrow \infty$. We will show that the test based on $W$ is consistent under the ultra-high dimensional
	setting. The main results of the test based on $W$ is as follows. 
	
	\begin{theorem}\label{thm:W_limit_dist_H0}
		Assume that $\bY=(\by_1,\ldots,\by_n)$ is a $p\times n$ data matrix with $n$ i.i.d. $p$-dimensional random vectors $\{\by_i=\bSigma_p^{1/2}\bx_i\}_{1\leqslant i \leqslant n}$ with covariance matrix $\bSigma_p=\Var(\by_i)$ and $\bx_i$ has $p$ i.i.d. components $\{X_{ij},~1\leq j\leq p\}$ satisfying $\Expe X_{ij}=0$, $\Expe X_{ij}^2=1$, $\Expe X_{ij}^4=\nu_4$ and $\Expe |X_{ij}|^{6+\varepsilon_0}<\infty$ for some small positive $\varepsilon_0$. $W$ is defined as \eqref{eq:W_def}. Then under $H_0$, when $p\wedge n\to\infty$ and $n^2/p=O(1)$,
		\begin{equation}\label{eq:W_null_dist_ultra}
			nW-p-(\nu_4-2) \convd \calN(0,4).
		\end{equation}
	\end{theorem}
	Note that the asymptotic distribution \eqref{eq:W_null_dist_ultra} coincides with \eqref{eq:W_null_dist_high}, which means $W$ has the same limiting null distribution in both high dimensional and ultra-high dimensional setting. Therefore $W$ can be used to test \eqref{eq:identity_test} under the ultra-high dimensional setting. For nominal level $\alpha$, the corresponding rejection rule is 
	\begin{equation}\label{eq:W_reject_rule}
		\frac{1}{2}\Bigl\{ nW-p-(\nu_4-2) \Bigr\} \geqslant z_{\alpha},
	\end{equation}
	where $z_{\alpha}$ is the $\alpha$ upper quantile of standard normal distribution.

As for the case of $H_1$ when $\bSigma_p\neq \bI_p$, we have	

\begin{theorem}\label{thm:W_limit_dist_H1}
Under the same assumptions as in Theorem~\ref{thm:W_limit_dist_H0},  further assume that  $\{\bSigma_p,~p\geq 1\}$ is a sequence of non-negative definite matrices, bounded in spectral norm such that the following limits exist:
		\begin{gather*}
			\gamma=\lim_{p\to \infty} \frac{1}{p}\tr(\bSigma_p),
			\qquad \theta=\lim_{p\to \infty} \frac{1}{p}\tr(\bSigma_p^2),
			\qquad \omega=\lim_{p\to \infty} \frac{1}{p}\sum_{i=1}^p (\bSigma_p)_{ii}^2,
		\end{gather*}
		then when $p\wedge n\to\infty$ and $n^2/p=O(1)$,
		\[
			nW-p-\theta\Bigl[\frac{\omega}{\theta}(\nu_4-3)+1\Bigr] +n(2\gamma-1-\theta) \convd \calN(0,4\theta^2).
		\]
	\end{theorem}
Note that Theorem \ref{thm:W_limit_dist_H1} reveals the limiting null distribution of $W$. Let $\bSigma_p=\bI_p$, then
$\gamma=\theta=\omega=1$, Theorem \ref{thm:W_limit_dist_H1} reduces to Theorem \ref{thm:W_limit_dist_H0}, which states the limiting null
distribution of $W$. With Theorem \ref{thm:W_limit_dist_H0} and \ref{thm:W_limit_dist_H1}, asymptotic power of $W$ can be derived.
\begin{proposition}\label{prop:power_theo_W}
	With the same assumptions as in Theorem~\ref{thm:W_limit_dist_H1}, when $p \wedge n\to\infty$ and $n^2/p=O(1)$, the testing power of $W$ for \eqref{eq:identity_test}
	\[
	\beta(H_1) \rightarrow 1-\Phi\biggl( \frac{1}{2\theta} \Bigl\{ 2z_{\alpha} - \omega(\nu_4-3) -\theta +n(2\gamma-1-\theta) + (\nu_4-2) \Bigr\} \biggr).
	\]
	
	If $\gamma=\theta=1$, then $\beta(H_1) \to 1-\Phi\bigl(z_{\alpha}-\tfrac{\omega-1}{2}(\nu_4-3)\bigr)$; otherwise, $\beta(H_1)\to 1$.
\end{proposition}


	The second test statistic of \eqref{eq:identity_test} we consider is the likelihood ratio test (LRT) statistic studied in \citet{Bai2009Correction}.  \citet{Bai2009Correction} assumed that $\nu_4=3$. The LRT statistic is defined as 
	\begin{equation}\label{eq:Bai_identity_test_L0}
		\calL_0 = \tr(\bS_n) - \log |\bS_n| - p.
	\end{equation}
	\citet{Bai2009Correction} derived the limiting null distribution of $\calL_0$ when $p\wedge n\to\infty,~p/n\to c\in (0,1)$. However, this LRT statistic is degenerate and not applicable when $p> n$ because $|\bS_n|=0$.  Thus for $p>n$ we introduce a quasi-LRT test statistic 
	\[
	\calL = \tr(\widehat{\bS}_n) - \log |\widehat{\bS}_n| - n,
	\]
	where $\widehat{\bS}_n = \frac{1}{p}\bY'\bY$.	When $p\wedge n\to\infty$, $p/n=c_n\to c \in(1, \infty)$, the limiting null distribution of $\calL$ is 
	\begin{equation}\label{eq:Bai_identity_test_stat_normalized}
		\calL^*:=\frac{\calL - n F_1(c_n) - \mu_1}{\sigma_1} \convd \calN(0,1),	
	\end{equation}
	where 
	\[
	F_1(c_n) = 1 - (1-c_n) \log\Bigl(1-\frac{1}{c_n}\Bigr),
	~ \mu_1 = -\frac{1}{2}\log\Bigl(1-\frac{1}{c_n}\Bigr),
	~\sigma_1^2 = -2\log\Bigl(1-\frac{1}{c_n}\Bigr)-\frac{2}{c_n}.
	\]

	Now we will show that this asymptotic distribution \eqref{eq:Bai_identity_test_stat_normalized} still holds in the ultra-high dimensional setting.
	Note that
	\[
	\sigma_1  =  \sqrt{-2\log\Bigl(1-\frac{1}{c_n}\Bigr)-\frac{2}{c_n}} = \sqrt{\frac{1}{c_n^2}+\frac{2}{3c_n^3}+o\biggl(\frac{1}{c_n^3}\biggr)}  = \frac{1}{c_n} +\frac{1}{3c_n^2} + o\biggl(\frac{1}{c_n^2}\biggr),
	\]
	which implies that 
	\begin{equation}\label{eq:frac_sigma_1}
		\frac{1}{\sigma_1} = c_n-\frac{1}{3} + o(1).
	\end{equation}
	Firstly, we consider the random part of $\calL^*$. Let $\widehat{\lambda}_1\geqslant \cdots \geqslant \widehat{\lambda}_n$ be the eigenvalues of $\widehat{\bS}_n$ and $\widetilde{\lambda}_1\geqslant \cdots \geqslant \widetilde{\lambda}_n$ be the eigenvalues of $\widetilde{\bS}_n=\sqrt{\tfrac{n}{p}}(\tfrac{1}{n}\bX'\bX-\tfrac{p}{n}\bI_n)$. By using the basic identity $\widehat{\lambda}_i=\tfrac{\widetilde{\lambda}_i}{\sqrt{c_n}}+1$, we have
	\begin{align}
		\calL & = \sum_{i=1}^n \widehat{\lambda}_i  - n - \sum_{i=1}^n\log(\widehat{\lambda}_i) = \sum_{i=1}^n \frac{\widetilde{\lambda}_i}{\sqrt{c_n}}  - \sum_{i=1}^n\log \Bigl( 1+ \frac{\widetilde{\lambda}_i}{\sqrt{c_n}} \Bigr) \nonumber\\[0.5em]
		& = \sum_{i=1}^n \frac{\widetilde{\lambda}_i}{\sqrt{c_n}}  - \sum_{i=1}^n\Bigl(\frac{\widetilde{\lambda}_i}{\sqrt{c_n}}-\frac{1}{2}\frac{\widetilde{\lambda}_i^2}{c_n} +\frac{1}{3}\frac{\widetilde{\lambda}_i^3}{c_n\sqrt{c_n}}-\frac{1}{4}\frac{\widetilde{\lambda}_i^4}{c_n^2} + o\Bigl(\frac{1}{c_n^2}\Bigr)\Bigr)\nonumber\\[0.5em]
		& = \frac{1}{2c_n}\tr(\widetilde{\bS}_n^2) - \frac{1}{3c_n\sqrt{c_n}}\tr(\widetilde{\bS}_n^3) +\frac{1}{4c_n^2}\tr(\widetilde{\bS}_n^4) + o\Bigl(\frac{n}{c_n^2}\Bigr).\label{eq:calL_expan}
	\end{align}
	Takeing $\nu_4=3$ (the assumption in \citet{Bai2009Correction}) and $\bSigma_p=\bI_p$ in Corollary~\ref{coro:CLT-x-x2-x3}, we have, under $H_0$,
	\begin{equation}\label{eq:tr_S2_S3_CLT}
		\tr(\widetilde{\bS}_n^2)-n-1\convd \calN(0,4),\qquad
		\tr(\widetilde{\bS}_n^3) - \frac{n+1}{\sqrt{c_n}} \convd \calN(0,24).
	\end{equation}
	\begin{equation}\label{eq:tr_S4_CLT}
		\tr(\widetilde{\bS}_n^4) - 2n - \Bigl(\frac{n}{c_n}+\frac{1}{c_n}+5\Bigr)\convd \calN(0,72).	
	\end{equation}
	Combining \eqref{eq:frac_sigma_1} $\sim$ \eqref{eq:tr_S4_CLT} gives us that 
	\begin{equation}\label{eq:random_calL}
		\frac{\calL}{\sigma_1} = \frac{1}{2}\tr(\widetilde{\bS}_n^2)+o\Bigl(\frac{n^2}{p}\Bigr).
	\end{equation}
	Secondly, we consider the determinist part of $\calL^*$. Note that
	\begin{align*}
		nF_1(c_n)+\mu & = n - \Bigl[n(1-c_n)+\frac{1}{2}\Bigr]\log\Bigl(1-\frac{1}{c_n}\Bigr)\\
		& = n - \Bigl[n(1-c_n)+\frac{1}{2}\Bigr]\cdot \Bigl[-\frac{1}{c_n}-\frac{1}{2c_n^2}-\frac{1}{3c_n^3}+o\Bigl(\frac{1}{c_n^3}\Bigr)\Bigr]\\
		& = \frac{n}{2c_n} +\frac{1}{2c_n} +\frac{n}{6c_n^2} +o\Bigl(\frac{n}{c_n^2}\Bigr), 
	\end{align*}
	together with \eqref{eq:frac_sigma_1} 	which implies that 
	\begin{equation}\label{eq:determinist_calL}
		\frac{nF_1(c_n)+\mu_1}{\sigma_1} = \frac{n+1}{2} + o\Bigl(\frac{n^2}{p}\Bigr).
	\end{equation}
	Therefore, from \eqref{eq:tr_S2_S3_CLT}, \eqref{eq:random_calL} and \eqref{eq:determinist_calL}, we conclude that, under $H_0$, as $p\wedge n\to\infty$, $n^2/p=O(1)$,  
	\[
	\calL^* = \frac{\calL}{\sigma_1} - \frac{n F_1(c_n) + \mu_1}{\sigma_1} = \frac{1}{2}\Bigl(\tr(\widetilde{\bS}_n^2) - n-1\Bigr) +o(1) \convd \calN(0, 1),
	\]
	which is the same as the limiting distribution \eqref{eq:Bai_identity_test_stat_normalized} when $p\wedge n\to\infty$, $p/n=c_n\to c \in(1, \infty)$.	Finally, we summarize  the discussion above in  the following proposition.
	\begin{proposition}
		\begin{enumerate}
			\item[(1)] (\citet{Bai2009Correction}) 	Assume that $\bY=(\by_1,\ldots,\by_n)$ is a $p\times n$ data matrix with $n$ i.i.d. $p$-dimensional random vectors $\{\by_i=\bSigma_p^{1/2}\bx_i\}_{1\leqslant i \leqslant n}$ with covariance matrix $\bSigma_p=\Var(\by_i)$ and $\bx_i$ has $p$ i.i.d. components $\{X_{ij},~1\leq j\leq p\}$ satisfying $\Expe X_{ij}=0$, $\Expe X_{ij}^2=1$, $\Expe X_{ij}^4=\nu_4=3$.    $\calL_0$ is defined as \eqref{eq:Bai_identity_test_L0}. Then under $H_0$, when $p \wedge n\to\infty$, $p/n\to c \in (0,1)$, we have
			\[
			\frac{\calL_0 - n F_0(c_n) - \mu_0}{\sigma_0} \convd \calN(0,1),	
			\]
			where $c_n=p/n$ and 
			\[
			F_0(c_n) = 1 - \frac{c_n-1}{c_n} \log(1-c_n),
			~ \mu_0 = -\frac{\log(1-c_n)}{2},
			~ \sigma_0^2 = -2\log(1-c_n)-2c_n.
			\]
			\item[(2)] Under the same assumptions as in (1) and the normalized quasi LRT statistic $\calL^*$ is defined in \eqref{eq:Bai_identity_test_stat_normalized}. Then under $H_0$, when $p\wedge n\to\infty$,  $p/n\to c \in (1,\infty)$, we have
			\[
			\calL^* \convd \calN(0,1).
			\]
			\item[(3)] Under the same assumptions as in (1) and the normalized quasi LRT statistic $\calL^*$ is defined in \eqref{eq:Bai_identity_test_stat_normalized}. Then under $H_0$, when  $p\wedge n\to\infty$ and $n^2/p=O(1)$, we have
			\[
			\calL^* \convd \calN(0,1).
			\]
		\end{enumerate}
	\end{proposition}
	
	Note that the results (2) and (3) in this proposition are newly derived.

\subsection{Separable Covariance Structure for Matrix-valued Noise}\label{sec:separable_structure_matrix_noise}

	In this section, we develop a test for the structure of the covariance matrix of a matrix-valued white noise. \citet{chen2021autoregressive} proposed a matrix autoregressive model with the form
	\[
	\bX_t = \bA\bX_{t-1}\bB'+\bE_t,	~t=1,\cdots, T,
	\]
	where $\bX_t$ is a $p_1\times p_2$ random matrix observed at time $t$, $\bA$ and $\bB$ are $p_1\times p_1$ and $p_2\times p_2$ deterministic autoregressive coefficient matrices,  $\bE_t=(e_{t,ij})$ is a $p_1\times p_2$ matrix-valued white noise. It's assumed that the error white noise matrix $\bE_t$ has a specific covariance structure
	\[
	\Cov\bigl(\mathsf{vec}(\bE_t)\bigr)=\bSigma_1\otimes\bSigma_2,
	\]
	where $\mathsf{vec}(\cdot)$ denotes the vectorization, $\bSigma_1$ and $\bSigma_2$ are $p_1\times p_1$ and $p_2\times p_2$ non-negative definite matrices. In other words, the noise $\bE_t$ has a separable covariance matrix.

	Now for any observed matrix-valued time sequence, we aim to test whether it has a separable covariance matrix. Specifically, suppose that $\{\bE_t\}_{1\leqslant t\leqslant T}$ is an observed i.i.d. sequence of $p_1 \times p_2 $ matrices and  $p_1, p_2$,$T$ are of comparable magnitude, we aim to test
	\begin{equation}\label{eq:separable_test}
		H_0: \Cov\bigl(\mathsf{vec}(\bE_t)\bigr)=\bSigma_1\otimes\bSigma_2,\qquad \text{vs.} \qquad H_1: \Cov\bigl(\mathsf{vec}(\bE_t)\bigr) \ne \bSigma_1\otimes\bSigma_2,
	\end{equation}
	where $\bSigma_1$ and $\bSigma_2$ are two prespecified $p_1\times p_1$ and $p_2\times p_2$ non-negative definite matrices. Testing $H_0: \Cov\bigl(\mathsf{vec}(\bE_t)\bigr)=\bSigma_1\otimes\bSigma_2$ is equivalent to testing 
	\[
		H_0': \Cov\Bigl(\bigl(\bSigma_1\otimes\bSigma_2\bigr)^{-\nicefrac{1}{2}}\mathsf{vec}(\bE_t)\Bigr)=\bI_{p_1p_2}.
	\]
	To this end, we define a test statistic 
	\begin{equation}\label{eq:W_star_def_1}
		W^* = \frac{1}{p_1p_2} \tr \Bigl[ \bigl(\bB_T-\bI_{p_1p_2}\bigr)^2 \Bigr] - \frac{p_1p_2}{T} \Bigl[\frac{1}{p_1p_2}\tr(\bB_T)\Bigr]^2+\frac{p_1p_2}{T}
	\end{equation}
	where
	\begin{equation}\label{eq:W_star_def_2}
		\bB_T = \frac{1}{T} \bY_T\bY_T', \qquad \bY_T = \bigl(\bSigma_1\otimes\bSigma_2\bigr)^{-\nicefrac{1}{2}}\bigl(\mathsf{vec}(\bE_1),\ldots,\mathsf{vec}(\bE_T)\bigr):=(Y_{ij})_{p_1p_2\times T}.
	\end{equation}

	Note that $W^*$ measures the distance between sample covariance matrix of $\mathsf{vec}(\bE_t)$ and $\bSigma_1\otimes\bSigma_2$. Naturally we reject $H_0$ when $W^*$ is too large and the critical value is determined by the limiting null distribution of $W^*$.

	Since $p_1,p_2,T$ are about the same order, we examine the asymptotic behavior of $W^*$ under the  high dimensional regime
	\begin{equation}\label{eq:three_dimension_infinity}
		T\to\infty, \qquad \frac{p_1}{T}=\frac{p_1(T)}{T}\to d_1 \in (0, \infty),\qquad \frac{p_2}{T}=\frac{p_2(T)}{T}\to d_2 \in (0, \infty).
	\end{equation}
	The asymptotic null distribution of the test statistic $W^*$ is given in the following Theorem. It is a direct implementation of Theorem \ref{thm:W_limit_dist_H0}.
	\begin{theorem}\label{prop:separable_test_H0_limit_dist}
		Assume that
		\begin{enumerate}
			\item[(1)] $\{\bE_t=(e_{t,ij})_{p_1\times p_2}\}_{1\leqslant t\leqslant T}$  is a sequence of i.i.d. sample matrices satisfying  $\mathsf{vec}(\bE_t) = (\bSigma_1\otimes \bSigma_2)^{\nicefrac{1}{2}} \mathsf{vec}(\bZ_t)$, where $\bZ_t=(Z_{t,ij})_{p_1\times p_2}$ is a  $p_1\times p_2$ matrix with i.i.d. real entries $Z_{t,ij}$ satisfying $\Expe Z_{t,ij} = 0$, $\Expe Z_{t,ij}^2 = 1$, $\Expe Z_{t,ij}^4 = \nu_4$ and $\Expe |Z_{t,ij}|^{6+\varepsilon_0}<\infty$ for some small positive $\varepsilon_0$;
			\item[(2)] $p_1,p_2,T$ tend to infinity as in \eqref{eq:three_dimension_infinity}.
		\end{enumerate}
		Then under the null hypothesis $H_0: \Cov\bigl(\mathsf{vec}(\bE_t)\bigr) = \bSigma_1\otimes\bSigma_2$, $W^*$ is defined as in \eqref{eq:W_star_def_1}, we have
		\[
			TW^*-p_1p_2-(\nu_4-2) \convd \calN(0,4).
		\]
	\end{theorem}
	According to the asymptotic normality of $W^*$ presented in Theorem~\ref{prop:separable_test_H0_limit_dist}, we reject $H_0$ at nominal level $\alpha$ if 
	\[
		\frac{1}{2}\Bigl\{TW^*-p_1p_2-(\nu_4-2)\Bigr\}\geqslant z_{\alpha}.
	\]

	Moreover, the asymptotic power of the proposed test for \eqref{eq:separable_test} can be derived as follows.
	\begin{proposition}\label{prop:separable_test_power_theo} Suppose that assumptions (1) and (2) in Theorem \ref{prop:separable_test_H0_limit_dist} hold, and
		\begin{enumerate}
			\item[(3)] $\widetilde{\bSigma}_1$ and $\widetilde{\bSigma}_2$ are two $p_1\times p_1$ and $p_2\times p_2$ non-negative definite matrices with bounded spectral norm, 
			$\widetilde{\bSigma} := (\widetilde{\bSigma}_1\otimes \widetilde{\bSigma}_2)^{\nicefrac{1}{2}} (\bSigma_1\otimes \bSigma_2)^{-1} (\widetilde{\bSigma}_1\otimes \widetilde{\bSigma}_2)^{\nicefrac{1}{2}}$ and the following limits exist:
			\begin{gather*}
				\gamma=\lim\limits_{T\to\infty}\frac{1}{p_1 p_2}\tr(\widetilde{\bSigma}),			
				\qquad \theta =\lim\limits_{T\to\infty}\frac{1}{p_1 p_2}\tr(\widetilde{\bSigma}^2),
				\qquad \omega = \lim\limits_{T\to\infty} \frac{1}{p_1 p_2}\sum_{i=1}^{p_1p_2}\bigl(\widetilde{\bSigma}\bigr)_{ii}^2.
			\end{gather*}
		\end{enumerate}
	Then when  $p_1,p_2,T$ tend to infinity as in \eqref{eq:three_dimension_infinity},	the testing power of  $W^*$  for \eqref{eq:separable_test} 
		\[
			\beta(H_1) \rightarrow 1-\Phi\biggl( \frac{1}{2\theta}\Bigl[ 2z_{\alpha}-\omega(\nu_4-3)-\theta+n(2\gamma-1-\theta)+(\nu_4-2) \Bigr] \biggr).	
		\]
		If $\gamma=\theta=1$, then $\beta(H_1) \rightarrow 1-\Phi\bigl(z_{\alpha}-\frac{\omega-1}{2}(\nu_4-3)\bigr)$; otherwise, $\beta(H_1)\to 1$.
	\end{proposition}

\section{Simulation results}\label{sec:simulation}

In this section, we implement some simulation studies to examine
\begin{enumerate}
	\item[(1)]  finite-sample properties of some LSS for $\bA_n$ by comparing their empirical means and variances with theoretical limiting values;
	\item[(2)] finite-sample performance of the separable covariance structure test in Section~\ref{sec:separable_structure_matrix_noise}.
\end{enumerate}

\subsection{LSS of $\bA_n$}\label{sec:simu-CLT}

Firstly we compare the empirical mean and variance of normalized $\left\{G_n(f_i)=\tr(\bA_n^i),~i=1, 2, 3\right\}$ with their theoretical limits in Corollary~\ref{coro:CLT-x-x2-x3}. Define
\begin{align*}
	\overline{G}_n(f_1)  &:= \frac{G_n(f_1)}{\sqrt{\Var(Y(f_1))}} =\frac{\tr(\bA_n)}{\sqrt{\frac{\omega}{\theta}(\nu_4-3)+2}},\\[0.5em]
	\overline{G}_n(f_2)  &:= \frac{G_n(f_2)}{\sqrt{\Var(Y(f_2))}} =\frac{1}{2}\biggl\{\tr(\bA_n^2) - n - \Bigl[\frac{\tilde{b}_p}{b_p}(\nu_4-3)+1\Bigr]\biggr\},\\[0.5em]
	\overline{G}_n(f_3)  &:= \frac{G_n(f_3)}{\sqrt{\Var(Y(f_3))}} =\frac{\tr(\bA_n^3) - \tfrac{c_p}{b_p\sqrt{b_p}}\sqrt{\tfrac{n}{p}}\Bigl[ n+1+\tfrac{\tilde{b}_p}{b_p}(\nu_4-3) \Bigr] }{\sqrt{\tfrac{9\omega}{\theta}(\nu_4-3)+24}}.
\end{align*}
According to Corollary~\ref{coro:CLT-x-x2-x3},  $\{\overline{G}_n(f_i)\}\xrightarrow{d} \calN(0,1)$, $i=1,2,3$. Hence we directly compare the empirical distribution of $\{\overline{G}_n(f_i)\}$ with $\calN(0,1)$ under different scenarios. Specifically, we consider two data distributions of $\{X_{ij}\}$ and three types of covariance matrix $\bSigma_p$, i.e.
\begin{itemize}
	\item[(1)] \textbf{Gaussian data:} $\{X_{ij},~1\leqslant i\leqslant p,~1\leqslant j\leqslant n\}$ i.i.d. $\calN(0, 1)$, with $\Expe X_{ij}^4=\nu_4=3$. 
	\item[(2)] \textbf{Non-Gaussian data:} $\{X_{ij},~1\leqslant i\leqslant p,~1\leqslant j\leqslant n\}$ i.i.d.  $\mathsf{Gamma}(4,2)-2$, with $\Expe X_{ij}=0$, $\Expe X^2_{ij}=1$, $\Expe X^4_{ij}=4.5$. 
\end{itemize}
As for $\bSigma_p$, 
\begin{itemize}
	\item[(A)] $\bSigma_A=\bI_p$;
	\item [(B)] $\bSigma_B$ is diagonal, $1/4$ of its diagonal elements are $0.5$, and $3/4$ are $1$.
	\item [(C)] $\bSigma_C$ is diagonal,  one half of its diagonal elements are $0.5$, and one half are 1.
\end{itemize}

Empirical mean and variance of $\{\overline{G}_n(f_i)\}$ are calculated for various combinations of $(p,n)$ under different model settings. For each pair of $(p,n)$, 5000 independent replications are used to obtain the empirical mean and variance.
Table \ref{tab:Qn_f_Simulation_n2} reports the empirical values of $\{\overline{G}_n(f_i)\}$ when $p=n^2$.
Table \ref{tab:Qn_f_Simulation_n2p5} reports the case of $p=n^{2.5}$. 
As shown in Tables~\ref{tab:Qn_f_Simulation_n2} and \ref{tab:Qn_f_Simulation_n2p5}, the empirical mean and variance of $\{\overline{G}_n(f_i)\}$ perfectly match their theoretical limits $0$ and $1$ under all scenarios, including  all three types of $\bSigma_p$, and for both Gaussian and non-Gaussian data.
\begin{table}[!h]
	\centering
	\caption{Empirical mean and variance of $\overline{G}_n(f_i),\; i=1, 2, 3$ from $5000$ replications. Theoretical mean and variance are $0$ and $1$, respectively. Dimension $p=n^{2}$.}
	\label{tab:Qn_f_Simulation_n2}
	\begin{tabular}{@{}lccccccccc@{}}
	\toprule
	& & \multicolumn{2}{c}{$\bSigma_p=\bSigma_A$} & & \multicolumn{2}{c}{$\bSigma_p=\bSigma_B$} & & \multicolumn{2}{c}{$\bSigma_p=\bSigma_C$} \\
	\cmidrule(r){3-4} \cmidrule(lr){6-7} \cmidrule(l){9-10}
	$n$  & & mean & var &  & mean & var & & mean & var \\ \toprule
	50  && 0.0050   & 1.0092 && 0.0038  & 1.0074 && -0.0157 & 1.0292 \\
	100 && -0.0103 & 0.9962 && 0.0148  & 1.0073 && -0.0048 & 1.0252 \\
	150 && -0.0075 & 1.0293 && -0.0054 & 1.0372 && -0.0113 & 0.9915 \\
	200 && -0.0052 & 0.9989 && 0.0206  & 1.0140  && -0.0008 & 1.0012  \\ 
	\multicolumn{6}{l}{$\overline{G}_n(f_1)$\qquad \textbf{Gaussian}} & 	\\ \midrule 
	50  && 0.0048  & 1.0265 && -0.0079 & 1.0065 && 0.0179  & 1.0119 \\
	100 && -0.0034 & 1.0041 && 0.0011  & 0.9983 && 0.0066  & 1.0305 \\
	150 && 0.0009  & 0.9841 && 0.0064  & 1.0159 && -0.0199 & 1.0273 \\
	200 && -0.0091 & 1.0093 && 0.0070   & 0.9929 && 0.0087  & 0.9751   \\ 
	\multicolumn{6}{l}{$\overline{G}_n(f_1)$\qquad \textbf{Non-Gaussian}} &	\\ \toprule
	50  && -0.0068 & 1.0848 && -0.0012 & 1.0922 && -0.0194 & 1.0871 \\
	100 && -0.0052 & 1.0678 && -0.0078 & 1.0266 && -0.0139 & 1.0289 \\
	150 && 0.0163  & 1.0209 && -0.0262 & 1.0291 && -0.0057 & 1.0250  \\
	200 && 0.0196  & 1.0223 && -0.0047 & 0.9972 && -0.0008 & 0.9930 \\ 
	\multicolumn{6}{l}{$\overline{G}_n(f_2)$\qquad \textbf{Gaussian}} &	\\ \midrule 
	50  && 0.0049  & 1.1533 && -0.0184 & 1.1588 && -0.0195 & 1.2185 \\
	100 && -0.0163 & 1.0927 && 0.0071  & 1.0896 && -0.0167 & 1.0924 \\
	150 && 0.0017  & 1.0513 && 0.0232  & 1.0574 && -0.0106 & 1.0655 \\
	200 && -0.0020  & 1.0568 && 0.0173  & 1.0361 && -0.0131 & 1.0568 \\ 
	\multicolumn{6}{l}{$\overline{G}_n(f_2)$\qquad \textbf{Non-Gaussian}}&	\\  \toprule
	50  && 0.0734 & 1.1134 && 0.0579 & 1.1145 && 0.0480  & 1.1727 \\
	100 && 0.0307 & 1.0642 && 0.0392 & 1.0720  && 0.0537 & 1.0805 \\
	150 && 0.0230 & 1.0919 && 0.0421 & 1.0489 && 0.0361 & 1.0502 \\
	200 && 0.0198 & 1.0131 && 0.0412 & 1.0372 && 0.0329 & 1.0457 \\ 
	\multicolumn{6}{l}{$\overline{G}_n(f_3)$\qquad \textbf{Gaussian}}&	\\ \toprule 
	50  && 0.1500   & 1.1976 && 0.1284 & 1.1964 && 0.1688 & 1.2197 \\
	100 && 0.0895 & 1.1090  && 0.0922 & 1.0841 && 0.0885 & 1.0988 \\
	150 && 0.0736 & 1.0491 && 0.0701 & 1.0494 && 0.0760  & 1.0851 \\
	200 && 0.0698 & 1.0447 && 0.0690  & 1.0834 && 0.0693 & 1.0300  \\ 
	\multicolumn{6}{l}{$\overline{G}_n(f_3)$\qquad  \textbf{Non-Gaussian}}&	\\ 
	\bottomrule
	\end{tabular}
\end{table}
\begin{table}[!h]
	\centering
	\caption{Empirical mean and variance of $\overline{G}_n(f_i),\; i=1, 2, 3$ from $5000$ replications. Theoretical mean and variance are $0$ and $1$, respectively. Dimension $p=n^{2.5}$.}
	\label{tab:Qn_f_Simulation_n2p5}
	\begin{tabular}{@{}lccccccccc@{}}
	\toprule
	& & \multicolumn{2}{c}{$\bSigma_p=\bSigma_A$} & & \multicolumn{2}{c}{$\bSigma_p=\bSigma_B$} & & \multicolumn{2}{c}{$\bSigma_p=\bSigma_C$} \\
	\cmidrule(r){3-4} \cmidrule(lr){6-7} \cmidrule(l){9-10}
	$n$  & & mean & var &  & mean & var & & mean & var \\ \toprule
	50  && -0.0095 & 0.9870  && -0.0023 & 1.0067 && 0.0092  & 1.0233 \\
	100 && -0.0067 & 1.0274 && 0.0009  & 0.9991 && 0.0115  & 1.0150 \\
	150 && -0.0056 & 1.0164 && 0.0109  & 0.9772 && -0.0086 & 0.973  \\
	200 && 0.0139  & 0.9949 && 0.012   & 0.9907 && -0.0179 & 1.0002 \\ 
	\multicolumn{6}{l}{$\overline{G}_n(f_1)$\qquad \textbf{Gaussian}} & 	\\ \midrule 
	50  && 0.0087 & 1.0332 && -0.0011 & 0.9972 && -0.0056 & 0.992  \\
	100 && 0.0016 & 0.9859 && -0.0148 & 0.9899 && -0.0054 & 1.0226 \\
	150 && 0.0093 & 1.0325 && 0.0088  & 1.0284 && 0.0380   & 0.9894 \\
	200 && 0.0109 & 0.9947 && -0.0199 & 1.0085 && 0.0038  & 0.9948  \\ 
	\multicolumn{6}{l}{$\overline{G}_n(f_1)$\qquad \textbf{Non-Gaussian}} &	\\ \toprule
	50  && 0.0044 & 1.0243 && -0.0045 & 1.0124 && -0.0152 & 1.0265 \\
	100 && 0.0191 & 0.9982 && 0.0022  & 1.0169 && -0.0173 & 1.0314 \\
	150 && 0.0010  & 1.0353 && 0.0086  & 1.0120  && 0.0065  & 1.0105 \\
	200 && 0.0039 & 1.0111 && -0.0178 & 1.0089 && 0.0167  & 1.0124 \\ 
	\multicolumn{6}{l}{$\overline{G}_n(f_2)$\qquad \textbf{Gaussian}} &	\\ \midrule 
	50  && 0.0049  & 1.0585 && -0.003  & 1.0967 && -0.0015 & 1.1071 \\
	100 && -0.017  & 1.04   && 0.007   & 1.0426 && 0.0085  & 1.0805 \\
	150 && 0.0113  & 1.0449 && -0.0019 & 1.0396 && 0.0033  & 1.0244 \\
	200 && -0.0178 & 1.0492 && -0.01   & 1.041  && -0.0049 & 1.0336 \\ 
	\multicolumn{6}{l}{$\overline{G}_n(f_2)$\qquad \textbf{Non-Gaussian}}&	\\  \toprule
	50  && 0.0045  & 1.0491 && 0.0298 & 1.0607 && 0.0406 & 1.0826 \\
	100 && -0.0051 & 1.0387 && 0.0021 & 1.0235 && 0.0255 & 1.0371 \\
	150 && 0.0023  & 0.9959 && 0.0115 & 1.0224 && 0.0097 & 0.9771 \\
	200 && 0.0323  & 1.0186 && 0.0045 & 1.003  && 0.0084 & 1.0037 \\ 
	\multicolumn{6}{l}{$\overline{G}_n(f_3)$\qquad \textbf{Gaussian}}&	\\ \toprule 
	50  && 0.0551 & 1.1447 && 0.0539  & 1.1059 && 0.0604 & 1.1238 \\
	100 && 0.0342 & 1.0608 && 0.0273  & 1.0318 && 0.0322 & 1.0817 \\
	150 && 0.0281 & 1.0671 && 0.029   & 1.0528 && 0.0642 & 1.0259 \\
	200 && 0.0347 & 1.0355 && -0.0017 & 1.0086 && 0.0266 & 1.0289 \\ 
	\multicolumn{6}{l}{$\overline{G}_n(f_3)$\qquad  \textbf{Non-Gaussian}}&	\\ 
	\bottomrule
	\end{tabular}
\end{table}

\subsection{Test for the Separable Covariance Structure}

Empirical size and power of the separable structure test in Section~\ref{sec:separable_structure_matrix_noise} are examined to testify the asymptotic testing power of $W^*$ given in Proposition~\ref{prop:separable_test_power_theo}. We compare the empirical power of $W^*$ with its limits under various model settings. Specifically,   the vectorization of data matrix $\bE_t$ is $\mathsf{vec}(\bE_t) = (\bSigma_1\otimes \bSigma_2)^{\nicefrac{1}{2}} \mathsf{vec}(\bZ_t)$.  We consider two data distributions of $\bZ_t=\{Z_{t,ij}\}$.
\begin{itemize}
	\item[(1)] \textbf{Gaussian matrix white noise:} $\{Z_{t, ij},~1\leqslant i\leqslant p,~1\leqslant j\leqslant n\}$  i.i.d. $\calN(0, 1)$, with $\nu_4=\Expe Z_{t, ij}^4=3$.
	\item[(2)] \textbf{Non-Gaussian matrix white noise:} $\{Z_{t, ij},~1\leqslant i\leqslant p,~1\leqslant j\leqslant n\}$ i.i.d.  $\mathsf{Gamma}(4,2)-2$, with $\Expe Z_{t, ij}=0$, $\Expe Z^2_{t, ij}=1$, $\nu_4=\Expe Z^4_{t, ij}=4.5$.
\end{itemize}
 As for covariance matrix $\bSigma_1\otimes\bSigma_2$, we set $\bSigma_1$ as a $p_1\times p_1$ tri-diagonal matrix, and $\bSigma_2$ as a $p_2\times p_2$ symmetric Toeplitz matrix. More specifically, 
\[
	\bSigma_1 = \begin{pmatrix}
		2&1&&&\\
		1&2&1&\\
		&1&\ddots&\ddots&\\
		&&\ddots&\ddots&1\\
		&&&1&2\\
	\end{pmatrix}_{p_1\times p_1},
\]
and $\bSigma_2 = \bigl(\rho^{|i-j|}\bigr)_{p_2\times p_2}$ with $|\rho|<1$. We set $\rho=0.45$, $p_1=p_2=T$ and $p_1 = 40, 60, 80, 100, 120$. The nominal level of the test is $\alpha = 0.05$. 
To obtain the empirical power, we keep $\bSigma_1$ unchanged and replace $\rho$ in $\bSigma_2$ with $\rho(1+\lambda)$ satisfying $|\rho(1+\lambda)|<1$. We vary $\lambda = 0, 0.2, 0.3, 0.4, 0.5$ to obtain different levels of testing power. For each pair of $(p_1, p_2, T)$, $5000$ independent replications are used to obtain the empirical size and power. Empirical values and theoretical limits are compared in Table~\ref{tab:separable_test_power}. 
As shown in Table~\ref{tab:separable_test_power}, the empirical power tends to $1$ when either $p_1, p_2, T$  or $\lambda$ increases. Most importantly, the empirical power value is consistent with its theoretical limit under all scenarios.

\begin{table}[!h]
	\centering
	\caption{Empirical (Emp) and Theoretical (Theo) Size ($\lambda=0$) and Power of the Separable Structure Test with $5000$ replications.}
	\label{tab:separable_test_power}
	\begin{tabular}{@{}ccccccccccccc@{}}
		\toprule
	\multicolumn{3}{c}{}  & \multicolumn{2}{c}{$\lambda=0$}     & \multicolumn{2}{c}{$\lambda=0.2$} & \multicolumn{2}{c}{$\lambda=0.3$} &  \multicolumn{2}{c}{$\lambda=0.4$} & \multicolumn{2}{c}{$\lambda=0.5$}  \\ 
	 \cmidrule(r){4-5} \cmidrule(lr){6-7} \cmidrule(lr){8-9} \cmidrule(lr){10-11} \cmidrule(l){12-13}
	$p_1$        & $p_2$        & $T$       & Emp           & Theo         & Emp         & Theo        & Emp         & Theo  &       Emp         & Theo        & Emp         & Theo           
	        \\ \toprule
	40  & 40  & 40  & 0.0490 & 0.05 & 0.0950 & 0.0880 & 0.2856 & 0.3087 & 0.8230 & 0.8354 & 0.9992 & 0.9992 \\
	60  & 60  & 60  & 0.0554 & 0.05 & 0.1650 & 0.1625 & 0.6484 & 0.6606 & 0.9974 & 0.9969 & 1      & 1      \\
	80  & 80  & 80  & 0.0520 & 0.05 & 0.2600 & 0.2699 & 0.8994 & 0.9084 & 1      & 1      & 1      & 1      \\
	100 & 100 & 100 & 0.0526 & 0.05 & 0.3916 & 0.4049 & 0.9864 & 0.9878 & 1      & 1      & 1      & 1      \\
	120 & 120 & 120 & 0.0542 & 0.05 & 0.5356 & 0.5524 & 0.9986 & 0.9992 & 1      & 1      & 1      & 1          \\
	\multicolumn{13}{l}{\textbf{Gaussian}}         \\ \toprule
	40  & 40  & 40  & 0.0568 & 0.05 & 0.0716 & 0.0662 & 0.2214 & 0.2353 & 0.7008 & 0.7568 & 0.9942 & 0.9977 \\
	60  & 60  & 60  & 0.0610 & 0.05 & 0.1298 & 0.1277 & 0.5462 & 0.5752 & 0.9878 & 0.9930 & 1      & 1      \\
	80  & 80  & 80  & 0.0580 & 0.05 & 0.2202 & 0.2216 & 0.8356 & 0.8655 & 1      & 1      & 1      & 1      \\
	100 & 100 & 100 & 0.0530 & 0.05 & 0.3312 & 0.3464 & 0.9694 & 0.9785 & 1      & 1      & 1      & 1      \\
	120 & 120 & 120 & 0.0562 & 0.05 & 0.4886 & 0.4910 & 0.9974 & 0.9984 & 1      & 1      & 1      & 1    \\
	\multicolumn{13}{l}{\textbf{Non-Gaussian}} \\ \bottomrule
	\end{tabular}
\end{table}

\section{Proof of Theorem \ref{thm:CLT}}\label{sec:proof_CLT_1}

In Section~\ref{sec:truncation} we first present the preliminary step of data truncation. The general strategy of the main proof of Theorem~\ref{thm:CLT} is explained in Section~\ref{sec:strategy_of_proof}. Three major steps of the general strategy are presented in Section~\ref{sec:conv_Mn1_proof}, \ref{sec:tight_Mn1} and \ref{sec:conv_Mn2} respectively. 

\subsection{Truncation, Centralization and Rescaling}\label{sec:truncation}

We first truncate the elements of $\bX$ without changing the weak limit of $G_n(f)$. We choose a positive sequence $\{\delta_n\}$ such that
\begin{equation}\label{eq:condi-trun}
	\delta_n^{-4} \Expe |X_{11}|^4 \indicator_{\{|X_{11}|\geqslant\delta_n \sqrt[4]{np}\}} \to 0,\qquad \delta_n\downarrow 0,\quad \delta_n\sqrt[4]{np}\uparrow \infty,
\end{equation}
as $n\to\infty$. Define
\begin{align*}
	\widehat{X}_{ij}& = X_{ij}\indicator_{\{|X_{11}|\leqslant\delta_n \sqrt[4]{np}\}},\qquad \sigma^2 = \Expe |\widehat{X}_{ij}-\Expe \widehat{X}_{ij}|^2, \qquad \widehat{\bX} = (\widehat{X}_{ij})_{p\times n },\\
	\widetilde{X}_{ij} & = (\widehat{X}_{ij}-\Expe \widehat{X}_{ij})/\sigma, \qquad \widetilde{\bX} = (\widetilde{X}_{ij})_{p\times n},\\
	\widehat{\bA}_n& = \lb \widehat{\bX}'\bSigma_p \widehat{\bX} -pa_p  \bI_n\rb/\sqrt{npb_p},\qquad \widetilde{\bA}_n=\lb \widetilde{\bX}'\bSigma_p \widetilde{\bX} -pa_p  \bI_n\rb/\sqrt{npb_p}.
\end{align*}
Define $\widehat{G}_n(f)$ and $\widetilde{G}(f)$ similarly by means of \eqref{eq:LSS} with the matrix $\bA_n$ replaced by $\widehat{\bA}_n$ and $\widetilde{\bA}_n$, respectively.
First, observe that
\[
	\Prob\bigl(G_n(f)\ne\widehat{G}_n(f)\bigr)\leqslant \Prob(\bA_n\ne\widehat{\bA}_n)=o(1).
\]
Indeed,
\[
	\Prob(\bA_n\ne\widehat{\bA}_n)
	\leqslant np\Prob(|X_{11}|\leq\delta_n \sqrt[4]{np})\leqslant K\delta_n^{-4} \Expe |X_{11}|^4\indicator_{\{|X_{11}|\geqslant\delta_n \sqrt[4]{np}\}}=o(1).
\]
Now we consider the difference between $\widehat{G}_n(f)$ and $\widetilde{G}_n(f)$. For any analytic function $f$ on $\mathscr{U}$, we have
\begin{align*}
\Expe \Bigl|\widehat{G}_n(f) -  \widetilde{G}_n(f)\Bigr| & = \sum_{k=1}^n \Bigl|f\bigl(\lambda_j^{\widehat{\bA}_n}\bigr)-f\bigl(\lambda_j^{\widetilde{\bA}_n}\bigr)\Bigr|\leqslant \frac{K_f}{\sqrt{npb_p}}\sum_{k=1}^n \Bigl|\lambda_j^{\widehat{\bX}'\bSigma_p \widehat{\bX}}-\lambda_j^{\widetilde{\bX}'\bSigma_p \widetilde{\bX}}\Bigr|\\[0.5em]
& \leqslant \frac{K_f}{\sqrt{npb_p}} \Expe \Bigl| \tr(\widehat{\bX}-\widetilde{\bX})'\bSigma_p (\widehat{\bX}-\widetilde{\bX})\cdot 2 \bigl(\tr(\widehat{\bB}_n)+\tr(\widetilde{\bB}_n)\bigr) \Bigr|^{\nicefrac{1}{2}}\\[0.5em]
&\leqslant \frac{2K_f}{\sqrt{npb_p}}\Bigl|\Expe \tr(\widehat{\bX}-\widetilde{\bX})'\bSigma_p (\widehat{\bX}-\widetilde{\bX})\Bigr|^{\nicefrac{1}{2}} \cdot \Bigl|\Expe \tr(\widehat{\bB}_n)+\Expe \tr(\widetilde{\bB}_n)\Bigr|^{\nicefrac{1}{2}},
\end{align*}
where $K_f$ is a bound on $|f'(x)|$. 

It follows from \eqref{eq:condi-trun} that
\begin{align*}
|\sigma^2-1|&\leqslant 2 \Expe X_{11}^2 \indicator_{\{|X_{11}|\geqslant\delta_n \sqrt[4]{np}\}}\\
&\leqslant \frac{2}{\delta_n^2 \sqrt{np}} \Expe |X_{11}|^4 \indicator_{\{|X_{11}|\geqslant\delta_n \sqrt[4]{np}\}}=o\bigl((np)^{-\nicefrac{1}{2}}\bigr),
\end{align*}
and
\begin{align*}
\bigl|\Expe \widehat{X}_{11}\bigr| &= \bigl| \Expe X_{11} \indicator_{\{|X_{11}|\geqslant\delta_n \sqrt[4]{np}\}} \bigr|\leqslant \Expe |X_{11}| \indicator_{\{|X_{11}|\geqslant\delta_n \sqrt[4]{np}\}}\\
&\leqslant \frac{1}{\delta_n^3 (np)^{3/4}} \Expe |X_{11}|^4 \indicator_{\{|X_{11}|\geqslant\delta_n \sqrt[4]{np}\}}  =o\bigl((np)^{-3/4}\bigr).
\end{align*}
These give us
\begin{align*}
\frac{1}{\sqrt{np}}\Bigl[\tr(\widehat{\bX}-\widetilde{\bX})'\bSigma_p (\widehat{\bX}-\widetilde{\bX})\Bigr]^{\nicefrac{1}{2}} 
& \leqslant \sum_{i, j} \sigma_{ii}\Expe |\widehat{X}_{ij}-\widetilde{X}_{ij}|^2 = \sum_{i, j} \sigma_{ii} \Expe \biggl|\frac{\sigma-1}{\sigma}\widehat{X}_{ij}+\frac{\Expe \widehat{X}_{ij}}{\sigma}\biggr|^2\\
& \leqslant Kpn\biggl(\frac{(1-\sigma)^2}{\sigma^2}\Expe|\widehat{X}_{11}|^2 + \frac{1}{\sigma^2}  \Expe |\widehat{X}_{11}|^2\biggr) = o(1),
\end{align*}
and 
\[
	\Expe \tr\bigl(\widehat{\bX}'\widehat{\bX}\bigr) \leqslant \sum_{i,j} \Expe |\widehat{X}_{ij}|^2 \leqslant Knp,\qquad 
	\Expe \tr\bigl(\widetilde{\bX}'\widetilde{\bX}\bigr) \leqslant \sum_{i,j} \Expe |\widetilde{X}_{ij}|^2 \leqslant Knp.
\]
From the above estimates, we obtain
\[
G_n(f)=\widetilde{G}_n(f)+o_p(1).
\]
Thus, we only need to find the limit distribution of $\{\widetilde{G}(f_j), j=1,\ldots,k\}$. Hence, in what follows, we assume that the underlying variables are truncated at $\delta_n \sqrt[4]{np}$, centralized, and renormalized. For convenience, we shall suppress the superscript on the variables, and assume that, for any $1\leqslant i\leqslant p$ and $1\leqslant j \leqslant n$,  
\begin{align*}
	|X_{ij}|&\leqslant \delta_n \sqrt[4]{np},\qquad \Expe X_{ij}=0, \qquad \Expe X_{ij}^2 = 1,\\[0.5em]
	\Expe X_{ij}^{a} &= \nu_{a} + o(1),\quad a=4,5,\qquad\qquad  \Expe |X_{ij}|^{6+\varepsilon_0}< \infty,	
\end{align*}
where $\delta_n$ satisfies the condition \eqref{eq:condi-trun}.

\subsection{Strategy of the proof}\label{sec:strategy_of_proof}

The general strategy of the proof follows the method established in \citet{bai2004clt} and \citet{BaiYao2005}.

Let $\scrC$ be the closed contour formed by the boundary of the rectangle with $(\pm u_1, \pm i v_1)$ where $u_1>2, 0<v_1\leqslant 1$. Assume that $u_1$ and $v_1$ are fixed and sufficiently small such that $\scrC\subset \scrU$. By Cauchy theorem, with probability one, we have
\[
	G_n(f) = -\frac{1}{2\pi i} \oint_{\mathscr{C}} f(z) n\Bigl[m_n(z)-m(z)-\mathcal{X}_n(m(z))\Bigr] \mathrm{d} z,
\]
where $m_n(z)$ and $m(z)$ are the Stieltjes transforms of $F^{\bA_n}$ and $F$, respectively.
The representation reduces our problem to finding the limiting process of
\[
	M_n(z) = n\Bigl[m_n(z)-m(z)-\mathcal{X}_n(m(z))\Bigr],\quad z\in\scrC.
\]
For $z\in \scrC$, we decompose $M_n(z)$ into a random part $M_n^{(1)}(z)$, and a determinist part $M_n^{(2)}(z)$, where
\[
	M_n^{(1)}(z) = n\bigl[m_n(z)-\Expe m_n(z)\bigr], \qquad M_n^{(2)}(z) = n\bigl[\Expe m_n(z) -m(z)-\mathcal{X}_n(m(z))\bigr].
\]
Throughout the paper, we set 
$
	\mathbb{C}_1 = \{z: z=u+iv, u\in [-u_1, u_1],~|v|\geqslant v_1\}.
$
The limiting process of $M_n(z)$ on $\mathbb{C}_1$ is stated in the following proposition.

\begin{proposition}\label{prop:Mn_CLT}
	Under the assumption $p\wedge n\to \infty, n^2/p=O(1)$ and after truncation of the data, the empirical process $\{M_n(z), z\in\mathbb{C}_1\}$ converges weakly to a centred Gaussian process $\{M(z), z\in\mathbb{C}_1\}$ 
	with the covariance function
	\begin{equation}\label{eq:Mn_cov}
		\Lambda(z_1, z_2) = m'(z_1)m'(z_2)\Bigl[\frac{\omega}{\theta}(\nu_4-3)+2\bigl(1-m(z_1)m(z_2)\bigr)^{-2}\Bigr].
	\end{equation}
\end{proposition}

Write the contour $\scrC$ as $\scrC=\scrC_{\ell}\cup \scrC_{r}\cup \scrC_{u}\cup \scrC_0$, where
\begin{align*}
	\scrC_{\ell}&= \{ z=-u_1+iv, \xi_n/n < |v| < v_1 \},\\
	\scrC_{r}&= \{ z=u_1+iv, \xi_n/n < |v| < v_1 \},\\
	\scrC_{0}&= \{ z=\pm u_1+iv,  |v| \leqslant \xi_n/n \},\\
	\scrC_{u}&= \{ z=u \pm iv_1,  |u| \leqslant u_1 \}
\end{align*}
and $\xi_n$ is a slowly varying sequence of positive constants and $v_1$ is a positive constant which is independent of $n$. Note that $\scrC_{\ell}\cup\scrC_{0}\cup \scrC_{\ell}=\scrC\setminus \mathbb{C}_1$.
To prove Theorem~\ref{thm:CLT}, we need to show that for $j=\ell, r, 0$ and some event $U_n$ with $\mathsf{P}(U_n)\to 1$,
\begin{equation}\label{eq:Mn_integral_zero}
	\lim_{v_1\downarrow 0} \limsup_{n\to\infty} \int_{\scrC_j} \Expe \Bigl|M_n(z)\indicator_{U_n}\Bigr|^2\mathrm{d}z=0
\end{equation}
and
\begin{equation}\label{eq:M_integral_zero}
	\lim_{v_1\downarrow 0} \int_{\scrC_j} \Expe \bigl|M(z)\bigr|^2\mathrm{d}z=0.
\end{equation}
The verification of \eqref{eq:Mn_integral_zero} and \eqref{eq:M_integral_zero} follows similar procedures developed in Sections 2.3, 3.1 and 4.3 of \citet{BaiYao2005} and the details will be omitted here. 

The calculation of the limiting covariance function of $Y(f)$ (see \eqref{eq:cov_1} \& \eqref{eq:cov_2}) is quite similar to that given in Section 5 of \citet{BaiYao2005}, it is then omitted.

Next, we can prove Proposition~\ref{prop:Mn_CLT} by the following three steps:
\begin{itemize}
	\item Finite-dimensional convergence of the random part $M_n^{(1)}(z)$ in distribution on $\mathbb{C}_1$;
	\item Tightness of the random part $M_n^{(1)}(z)$.
	\item Convergence of the non-random part $M_n^{(2)}(z)$ to the mean function on $\mathbb{C}_1$.
\end{itemize}
Details of the three steps are presented in the coming sections~\ref{sec:conv_Mn1_proof}, \ref{sec:tight_Mn1} and \ref{sec:conv_Mn2}, respectively.

\subsection{Finite dimensional convergence of $M_n^{(1)}(z)$ in distribution}\label{sec:conv_Mn1_proof}

We first decompose the random part $M_n^{(1)}(z)$ as a sum of martingale difference sequences, which is given in \eqref{eq:Mn_MDS}. Then, we apply the martingale CLT (Lemma~\ref{lem:CLT_MDS}) to obtain the asymptotic distribution of $M_n^{(1)}(z)$. Note that we prove the finite dimensional convergence of $M_n^{(1)}(z)$ under the assumption $p/n\to\infty$, which is weaker than $n^2/p=O(1)$.

First, we introduce some notations.
Define
\begin{align*}
	\bX_k&=(\bx_1,\ldots,\bx_{k-1},\bx_{k+1},\ldots,\bx_n),\qquad \bA_k=\tfrac{1}{\sqrt{npb_p}}\Bigl(  \bX_k'\bSigma_p \bX_k-pa_p  \bI_{n-1}\Bigr),\\[0.5em]
	\bD&=(\bA-z\bI_n)^{-1},\qquad \bD_k = (\bA_k-z\bI_{n-1})^{-1},\qquad \bM_k^{(s)}  =\bSigma_p\bX_k\bD_k^{s}\bX_k'\bSigma_p, ~ s=1, 2,\\[0.5em]
	a_{kk}^{\mathsf{diag}}&=A_{kk}-z=\tfrac{1}{\sqrt{npb_p}}\bigl(\bx_k'\bSigma_p\bx_k-pa_p\bigr)-z,\qquad \bq_k'=\tfrac{1}{\sqrt{npb_p}}\bigl(\bx_k'\bSigma_p\bX_k\bigr)\\[0.5em]
	\beta_k & =\frac{1}{-a_{kk}^{\mathsf{diag}}+\bq_k'\bD_k\bq_k},\qquad \beta^{\mathsf{tr}}_k = \frac{1}{z+(npb_p)^{-1}\tr \bM_k^{(1)}},\\[0.5em]
	\gamma_{ks} & = -\frac{1}{npb_p}\tr\bM_k^{(s)}+\bq_k'\bD_k^{s}\bq_k, ~ s=1, 2,\qquad \eta_k = \tfrac{1}{\sqrt{npb_p}}\bigl(\bx_k'\bSigma_p\bx_k-pa_p\bigr)-\gamma_{k1},\\[0.5em]
	\ell_k&=-\beta_k\beta^{\mathsf{tr}}_k\eta_k\bigl(1+\bq_k'\bD_k^{2}\bq_k\bigr).
\end{align*}

Note that $a_{kk}^{\mathsf{diag}}$ is the $k$-th diagonal element of $\bD^{-1}$ and $\bq_k'$ is the vector from the $k$-th row of $\bD^{-1}$ by deleting the $k$-th element.
By applying Theorem A.5 in \citet{bai2010spectral}, we obtain the equality
\begin{equation}\label{eq:trace_inv_diff}
	\tr \bD-\tr \bD_k=-\frac{1+\bq_k'\bD_k^{2}\bq_k}{-a_{kk}^{\mathsf{diag}}+\bq_k'\bD_k\bq_k}=-\beta_k(1+\bq_k'\bD_k^{2}\bq_k).
\end{equation}
Straightforward calculation gives:
\begin{equation}\label{eq:betak_decom}
	\beta_k-\beta^{\mathsf{tr}}_k=\beta_k\beta^{\mathsf{tr}}_k\eta_k
\end{equation}
and
\begin{equation}\label{eq:Expe_beta_eta_kappa}
	(\Expe_k-\Expe_{k-1}) \beta^{\mathsf{tr}}_k(1+\bq_k'\bD_k^{2}\bq_k)=\Expe_k \bigl(\beta_k^{\tr}\gamma_{k2}\bigr),\qquad \Expe_{k-1}\bigl(\beta_k^{\tr}\gamma_{k2}\bigr)= 0,
\end{equation}
where $\Expe_k(\cdot)$ is the expectation with respect to the $\sigma$-field generated by the first $k$ columns of $\bX$. 

By the definition of $\bD$ and $\bD_k$, we obtain two basic identities:
\begin{align}
	\bD\bX'\bSigma_p\bX&=p a_p \bD +\sqrt{npb_p}(\bI_n+z\bD),\label{eq:DXSX}\\[0.5em]
	\bD_k\bX_k'\bSigma_p\bX_k&=p a_p \bD_k +\sqrt{npb_p}(\bI_{n-1}+z\bD_k).\label{eq:DXSX_k}
\end{align}
If $\bSigma_p=\bI_p$, it is straightforward to derive that the limit of $\tr\bigl(\bM_{k}^{(1)}(z)\bigr)/(npb_p)$ is $m(z)$ by using \eqref{eq:DXSX_k}. However, when $\bSigma_p\neq \bI_p$, we need more detailed estimate.

\begin{lemma}\label{lem:Mk_limit}
	Under the assumption that $p\wedge n\to\infty$ and $p/n\to\infty$, we have, for $z\in\mathbb{C}_1$,
	\begin{equation}
		\Expe \biggl| \frac{1}{npb_p} \tr\bigl(\bM_{k}^{(1)}(z)\bigr) - m(z) \biggr|^2
		\leqslant 	\frac{K n}{p} + \frac{K}{n^2}.\label{eq:Diff-trMk-m}
	\end{equation}
\end{lemma}

\begin{proof}[Proof]
	Using Lemma \ref{lem:quad-trace}, we have
	\begin{equation}\label{eq:quad-trace-upper-bound}
		\Expe (\bx_i'\bSigma_p^2\bx_i - pb_p)^2 \leqslant K  \nu_4\tr(\bSigma_p^4)\leqslant K \cdot  p\|\bSigma_p^4\| \leqslant K p .
	\end{equation}
	Note that $\tr(\bA^*\bB)$ is the inner product of $\mathsf{vec}(\bA)$ and $\mathsf{vec}(\bB)$ for any $n\times m$ matrices $\bA$ and $\bB$.
	It follows from the Cauchy-Schwarz inequality that
	\begin{equation}\label{eq:CS-trace}
		\bigl|\tr(\bA^*\bB)\bigr|^2 \leqslant \tr(\bA^*\bA)\cdot \tr(\bB^*\bB).
	\end{equation}
	By using \eqref{eq:CS-trace}, we have
	\begin{align*}
		\Expe \biggl| \frac{1}{npb_p} &\tr \bM_{k}^{(1)}(z) - \frac{1}{n}\tr  \bD_k^{1}(z)\biggr|^2 \\
		&=\frac{1}{(npb_p)^2} \Expe \; \Bigl|\tr\Bigl(\bD_k(z)\bigl(\bX_k'\bSigma_p^2\bX_k-pb_p\bI_{n-1}\bigr)\Bigr)\Bigr|^2\\
		&\leqslant \frac{1}{(npb_p)^2} \Expe \; \biggl[\tr\bigl(\bD_k(\bar{z})\bD_k(z)\bigr) \cdot \tr\bigl(\bX_k'\bSigma_p^2\bX_k-pb_p\bI_{n-1}\bigr)^2\biggr]\\
		&\leqslant \frac{1}{(npb_p)^2} \Expe \; \biggl[n\bigl\|\bD_k(\bar{z})\bD_k(z)\bigr\| \cdot \tr\bigl(\bX_k'\bSigma_p^2\bX_k-pb_p\bI_{n-1}\bigr)^2\biggr]\\
		&\leqslant \frac{1}{n(pb_p v_1)^2} \Expe \; \biggl[\tr\bigl(\bX_k'\bSigma_p^2\bX_k-pb_p\bI_{n-1}\bigr)^2\biggr].
	\end{align*}
	Indeed, by using \eqref{eq:quad-trace-upper-bound} and the fact $\Expe (\bx_i'\bSigma_p^2\bx_j)^2 = \tr(\bSigma_p^4)$, we have
	\begin{align*}
		\Expe \biggl[\tr\bigl(\bX_k'\bSigma_p^2\bX_k-pb_p\bI_{n-1}\bigr)^2\biggr] & = \sum_{i\neq k} \Expe (\bx_i'\bSigma_p^2\bx_i - pb_p)^2 + \sum_{i\neq j, i\neq k, j\neq k} \Expe (\bx_i'\bSigma_p^2\bx_j)^2 \\
		&\leqslant (n-1) \cdot p K + (n-1)(n-2) \cdot p K.
	\end{align*}
	Thus we have
	\begin{equation}\label{eq:Diff-trMk-trDk}
		\Expe \biggl| \frac{1}{npb_p} \tr \bM_{k}^{(1)}(z) - \frac{1}{n}\tr \bD_k(z) \biggr|^2
		\leqslant 	\frac{K n}{p}.
	\end{equation}
	Moreover, by \eqref{eq:trace_inv_diff} and \eqref{eq:beta_qDq_upper_bdd}, we have
	\begin{equation}
		\biggl|\frac{1}{n}\tr\bD (z)-\frac{1}{n}\tr\bD_k(z)\biggr| \overset{\eqref{eq:trace_inv_diff}}{=} \frac{1}{n}\Bigl|\beta_k(1+\bq_k'\bD_k^{2}\bq_k)\Bigr|\overset{\eqref{eq:beta_qDq_upper_bdd}}{\leqslant} \frac{1}{n v_1}, \label{eq:D_Dk_ESD_diff}
	\end{equation}
	which, together with \eqref{eq:Diff-trMk-trDk} and the fact that $m_n(z)\convas m$, implies \eqref{eq:Diff-trMk-m}.
\end{proof}


Applying \eqref{eq:trace_inv_diff} $\sim$ \eqref{eq:Expe_beta_eta_kappa}, we have the following decomposition:
\begin{align}
	M_n^{(1)}(z)&=\tr\bD-\Expe\tr\bD=\sum_{k=1}^{n} (\Expe_k-\Expe_{k-1})\bigl(\tr \bD-\tr \bD_k\bigr)\nonumber\\
	&=-\sum_{k=1}^{n} (\Expe_k-\Expe_{k-1})\beta_k\Bigl(1+\bq_k'\bD_k^{2}\bq_k\Bigr)\label{eq:Mn1_decom_0}\\
	&\overset{\eqref{eq:betak_decom}}{=}~\sum_{k=1}^{n} (\Expe_k-\Expe_{k-1})(-\beta_k\beta^{\mathsf{tr}}_k\eta_k)\Bigl(1+\bq_k'\bD_k^{2}\bq_k\Bigr)  - \sum_{k=1}^{n} (\Expe_k-\Expe_{k-1})\beta^{\mathsf{tr}}_k\Bigl(1+\bq_k'\bD_k^{2}\bq_k\Bigr) \nonumber\\
	&=\sum_{k=1}^{n} (\Expe_k-\Expe_{k-1}) \ell_k - \Expe_k \bigl(\beta_k^{\tr}\gamma_{k2}\bigr).\label{eq:Mn1_decom_1}
\end{align}

\noindent By using \eqref{eq:betak_decom}, we can split $\ell_k$ as
\begin{align}
	\ell_k&=-\bigl[(\beta^{\mathsf{tr}}_k)^2\eta_k+\beta_k(\beta^{\mathsf{tr}}_{k})^2\eta_k^2\bigr]\Bigl(1+\bq_k'\bD_k^{2}\bq_k\Bigr)\nonumber\\[0.5em]
	&=-(\beta^{\mathsf{tr}}_k)^2\eta_k \Bigl(1+\frac{1}{npb_p}\tr\bM_k^{(2)}\Bigr)-\beta^{\mathsf{tr}}_k\eta_k \gamma_{k2}- \beta_k(\beta^{\mathsf{tr}}_{k})^2\eta_k^2\Bigl(1+\bq_k'\bD_k^{2}\bq_k\Bigr)\nonumber\\
	&=:\ell_{k1}+\ell_{k2}+\ell_{k3}.\label{eq:ellk_decom}
\end{align}
By Lemma~\ref{lem:betak_upper_bounds} and Lemma~\ref{lem:gamma_moment_upper_bound}, it is not difficult to verify that
\begin{equation}
	\Expe \biggl| \sum_{k=1}^n (\Expe_k-\Expe_{k-1}) \ell_{k2}\biggr|^2 = o(1),\qquad \Expe \biggl| \sum_{k=1}^n (\Expe_k-\Expe_{k-1}) \ell_{k3}\biggr|^2 = o(1).
\end{equation}
These estimates, together with \eqref{eq:Mn1_decom_1} and \eqref{eq:ellk_decom}, imply that
\begin{align}
	M_n^{(1)}(z) &= \sum_{k=1}^n \Expe_k \biggl[-(\beta^{\mathsf{tr}}_k)^2\eta_k \Bigl(1+\frac{1}{npb_p}\tr\bM_k^{(2)}\Bigr)-\beta_k^{\tr}\gamma_{k2}\biggr]+o_{L_2}(1)\nonumber\\
	&=:\sum_{k=1}^n Y_k(z)+o_{L_2}(1),\label{eq:Mn_MDS}
\end{align}
where $Y_k(z)$ is a sequence of martingale difference. Thus, to prove finite-dimensional convergence of $M_n^{(1)}(z), z\in\mathbb{C}_1$, we need only to consider the limit of the following term:
\[
	\sum_{j=1}^r a_j M_n^{(1)}(z_j) = \sum_{j=1}^{r} a_j\sum_{k=1}^n Y_k(z_j) +o(1)	= \sum_{k=1}^n\Biggl(\sum_{j=1}^{r} a_j Y_k(z_j)\Biggr) + o(1),
\]
where $\{a_j\}$ are complex numbers and $r$ is any positive integer.

By Lemma~\ref{lem:betak_upper_bounds} and Lemma~\ref{lem:gamma_moment_upper_bound}, we have
\begin{equation}\label{eq:MDS_4th_moment}
	\Expe|Y_j(z)|^4 \leqslant K\frac{\delta_n^4}{n} + K\biggl(\frac{1}{n^2}+\frac{n}{p^2}\biggr),
\end{equation}
which implies that, for each $\varepsilon >0$,
\[
	\sum_{k=1}^n \Expe\Biggl( \biggl|\sum_{j=1}^{r} a_j Y_k(z_j)\biggr|^2 \indicator_{\{\sum_{j=1}^{r} a_j Y_k(z_j)\geqslant \varepsilon\}} \Biggr) \leqslant \frac{1}{\varepsilon^2} \sum_{k=1}^n \Expe \biggl|\sum_{j=1}^{r} a_j Y_k(z_j)\biggr|^4 = o(1),
\]
which implies that the second condition \eqref{eq:CLT_MDS_condition2} of the martingale CLT (see Lemma~\ref{lem:CLT_MDS}) is satisfied. Thus, to apply the martingale CLT, it is sufficient to verify that, for $z_1, z_2 \in\mathbb{C}^{+}$, the sum
\begin{equation}\label{eq:Lambda_1}
	{\Lambda}_n(z_1,z_2):=\sum_{k=1}^n \Expe_{k-1} \bigl(Y_k(z_1)Y_k(z_2)\bigr)
\end{equation}
converges in probability to a constant (and to determine this constant).

Note that
\[
	-\bigl(\beta^{\mathsf{tr}}_k\bigr)^2\eta_k\Bigl(1+\frac{1}{npb_p}\tr\bM_k^{(2)}\Bigr)  - \beta^{\mathsf{tr}}_k\gamma_{k2}=\frac{\partial}{\partial z} \Bigl[\beta^{\mathsf{tr}}_k(z)\eta_k(z)\Bigr],
\]
thus, we have
\begin{equation}\label{eq:Lambda_2}
	{\Lambda}_n(z_1,z_2)=\frac{\partial^2}{\partial z_2\partial z_1} \sum_{k=1}^{n} \Expe_{k-1}\Bigl[\Expe_k \bigl[\beta^{\mathsf{tr}}_k(z_1)\eta_k(z_1)\bigr] \cdot \Expe_k \bigl[(\beta^{\mathsf{tr}}_k(z_2)\eta_k(z_2)\bigr]\Bigr].
\end{equation}
It is enough to consider the limit of
\begin{equation}\label{eq:Lambda_tilde_1}
	\sum_{k=1}^{n} \Expe_{k-1}\Bigl[\Expe_k \bigl[\beta^{\mathsf{tr}}_k(z_1)\eta_k(z_1)\bigr] \cdot \Expe_k \bigl[(\beta^{\mathsf{tr}}_k(z_2)\eta_k(z_2)\bigr]\Bigr].
\end{equation}
By \eqref{eq:SCLaw_equation}, \eqref{lem:Mk_limit} and the dominated convergence theorem, we conclude that
\begin{equation}\label{eq:betak_tr_limit}
	\Expe \Bigl| \beta_k^{\tr}(z)+m(z) \Bigr|^2 = o(1).
\end{equation}
Combining \eqref{eq:betak_tr_limit} into \eqref{eq:Lambda_tilde_1} yields that
\begin{align}
	\;&\sum_{k=1}^{n} \Expe_{k-1}\Bigl[\Expe_k \bigl[\beta^{\mathsf{tr}}_k(z_1)\eta_k(z_1)\bigr] \cdot \Expe_k \bigl[(\beta^{\mathsf{tr}}_k(z_2)\eta_k(z_2)\bigr]\Bigr]\nonumber\\
	=\;& m(z_1)m(z_2)\sum_{k=1}^{n} \Expe_{k-1}\Bigl(\Expe_k\eta_k(z_1)\cdot \Expe_k \eta_k(z_2)\Bigr)+o_p(1)\nonumber\\
	=:\;& m(z_1)m(z_2)\widetilde{\Lambda}_n(z_1,z_2)+o_p(1).\label{eq:Lambda_tilde_2}
\end{align}
In view of \eqref{eq:Lambda_1} $\sim$ \eqref{eq:Lambda_tilde_2}, it suffices to derive the limit of $\widetilde{\Lambda}_n(z_1,z_2)$, which further gives the limit of \eqref{eq:Lambda_1}.




Since $\Expe_k \bigl[\eta_k(z)\bigr] =(1/\sqrt{npb_p})\bigl(\bx_k'\bSigma_p\bx_k-pa_p\bigr)-\Expe_k\bigl[\gamma_{k1}(z)\bigr]$, we have
\begin{equation}\label{eq:Expe_eta1_eta2}
	\Expe_{k-1}\Bigl[\Expe_k \eta_k(z_1)\cdot \Expe_k \eta_k(z_2)\Bigr] = \frac{1}{n}\biggl[\frac{\tilde{b}_p}{b_p}(\nu_4-3)+2\biggr] + A_{1}^{(k)} +A_{2}^{(k)} +A_{3}^{(k)},
\end{equation}
where
\begin{align*}
	A_1^{(k)} & = \Expe_{k-1}\Bigl[\Expe_k\gamma_{k1}(z_1)\cdot\Expe_k\gamma_{k1}(z_2)\Bigr],\qquad A_2^{(k)} = -\Expe_{k-1}\biggl[\tfrac{1}{\sqrt{npb_p}}\bigl(\bx_k'\bSigma_p\bx_k-pa_p\bigr) \cdot\Expe_k\gamma_{k1}(z_1)\biggr],\\[0.5em]
	A_3^{(k)} & = -\Expe_{k-1}\biggl[\tfrac{1}{\sqrt{npb_p}}\bigl(\bx_k'\bSigma_p\bx_k-pa_p\bigr)\cdot \Expe_k\gamma_{k1}(z_2)\biggr] .
\end{align*}

\noindent First, we show that $A_2^{(k)}$ and $A_3^{(k)}$ are negligible. Denote $\bM_k^{(1)}(z) = \bigl(a_{ij}^{(1)}(z)\bigr)_{p\times p}$, using the independence between $\bx_k$ and $\bM_k^{(1)}$, we have
\begin{align}
	A_2^{(k)} & = - \frac{1}{npb_p\sqrt{npb_p}}\Expe_{k-1}\biggl[\biggl(\sum_{i,j}\sigma_{ij} X_{ik}X_{jk}-pa_p\biggr)\biggl( \sum_{i\neq j} X_{ik}X_{jk}\Expe_{k}a_{ij}^{(1)}(z_1) + \sum_{i=1}^p (X_{ik}^2-1)\Expe_{k}a_{ii}^{(1)}(z_1) \biggr)  \biggr]\nonumber\\
	& = - \frac{1}{npb_p\sqrt{npb_p}}\Expe_{k-1}\biggl[\biggl( \sum_{i\neq j} \sigma_{ij}X_{ik}^2X_{jk}^2\Expe_{k}a_{ij}^{(1)}(z_1) +  \sum_{i=1}^p \sigma_{ii}X_{ik}^2(X_{ik}^2-1)\Expe_{k}a_{ii}^{(1)}(z_1) \biggr)  \biggr]\nonumber\\
	& = - \frac{1}{npb_p\sqrt{npb_p}}\Bigl(\sum_{i\neq j}\sigma_{ij} \Expe_{k}a_{ij}^{(1)}(z_1) + (\nu_4-1)\sum_{i=1}^p \sigma_{ii}\Expe_{k}a_{ii}^{(1)}(z_1)\Bigr)\nonumber \\
	& = - \frac{1}{\sqrt{npb_p}}\Expe_k\biggl[\frac{1}{npb_p} \tr\Bigl(\bSigma_p\bM_k^{(1)}\Bigr)- \frac{\nu_4-2}{npb_p}\sum_{i=1}^p \sigma_{ii}a_{ii}^{(1)}(z_1)\biggr]. \label{eq:square_bracket}
\end{align}
As for the first term in the bracket of \eqref{eq:square_bracket}, we can estimate it by using the similar argument as in the proof of Lemma~\ref{lem:Mk_limit}.
Replacing $pb_p$ and $\bM_k^{(1)}$ in the proof of Lemma~\ref{lem:Mk_limit} with $\tr(\bSigma_p^3)$ and $\bSigma_p\bM_k^{(1)}$, we can prove that 
\[
	\Expe \biggl| \frac{1}{n\tr(\bSigma_p^3)} \tr\Bigl(\bSigma_p\bM_{k}^{(1)}\Bigr) - \frac{1}{n}\tr \bD_k \biggr|^2
\leqslant 	\frac{K n}{p}.
\]
Moreover, by the fact $\tfrac{b_p^2}{a_p}\leqslant \tr(\bSigma_p^3)\leqslant Kp$, the first inequality of which follows from \eqref{eq:CS-trace}, we conclude that
\[
	\frac{1}{npb_p} \tr\Bigl(\bSigma_p\bM_{k}^{(1)}\Bigr) = \frac{\tr(\bSigma_p^3)}{pb_p}\cdot \frac{1}{n\tr(\bSigma_p^3)} \tr\Bigl(\bSigma_p\bM_{k}^{(1)}\Bigr) = O_p(1).
\]
As for the second term in the bracket of \eqref{eq:square_bracket}, we have
\[
	\frac{1}{npb_p}\sum_{i=1}^p \sigma_{ii}a_{ii}^{(1)} \leqslant \frac{\|\bSigma_p\|}{npb_p}\sum_{i=1}^p a_{ii}^{(1)} = \frac{\|\bSigma_p\|}{npb_p}\tr \bM_k^{(1)}=O_p(1).
\]
Thus, we conclude that the term in the square bracket of \eqref{eq:square_bracket} is bounded in probability. Thus $\Bigl|\sum_{k=1}^n A_2^{(k)}\Bigr|\to 0$.
Similarly, we can show that $\Bigl|\sum_{k=1}^n A_3^{(k)}\Bigr|\to 0$.

Now we consider $A_1^{(k)}$.
We consider the second terms on the RHS of \eqref{eq:Expe_eta1_eta2}  with the notation $\Expe_k\bM_k^{(1)}(z) = \bigl(a_{ij}^{(1)}(z)\bigr)_{n\times n}$,
\begin{align*}
A_1^{(k)}&=  \frac{1}{(npb_p)^2}\Expe_{k-1} \biggl[\biggl( \sum_{i\neq j} X_{ik}X_{jk}\Expe_{k}a_{ij}^{(1)}(z_1) + \sum_{i=1}^p (X_{ik}^2-1)\Expe_{k}a_{ii}^{(1)}(z_1) \biggr) \\
&\qquad\qquad\qquad\qquad\qquad \times \biggl( \sum_{i\neq j} X_{ik}X_{jk}\Expe_{k}a_{ij}^{(1)}(z_1) + \sum_{i=1}^p (X_{ik}^2-1)\Expe_{k}a_{ii}^{(1)}(z_1) \biggr)\biggr]\\
=&\; \frac{1}{(npb_p)^2}\Expe_{k-1} \biggl[ 2\sum_{i\neq j}X_{ik}^2 X_{jk}^2\Expe_{k}a_{ij}^{(1)}(z_1)\Expe_{k}a_{ij}^{(1)}(z_2) + \sum_{i=1}^p (X_{ik}^2-1)^2\Expe_{k}a_{ii}^{(1)}(z_1)\Expe_{k}a_{ii}^{(1)}(z_2) \biggr]\\
=&\; \frac{1}{(npb_p)^2} \biggl[ 2\sum_{i\neq j}\Expe_{k}a_{ij}^{(1)}(z_1)\Expe_{k}a_{ij}^{(1)}(z_2) + (\nu_4-1)\sum_{i=1}^p \Expe_{k}a_{ii}^{(1)}(z_1)\Expe_{k}a_{ii}^{(1)}(z_2) \biggr]\\
=&\; \frac{1}{(npb_p)^2} \biggl[ 2\sum_{i, j}\Expe_{k}a_{ij}^{(1)}(z_1)\Expe_{k}a_{ij}^{(1)}(z_2) + (\nu_4-3)\sum_{i=1}^p \Expe_{k}a_{ii}^{(1)}(z_1)\Expe_{k}a_{ii}^{(1)}(z_2) \biggr]\\
=&\;\frac{2}{(npb_p)^2} \biggl[\tr \Bigl(\Expe_k\bM_k^{(1)}(z_1)\cdot \Expe_k\bM_k^{(1)}(z_2)\Bigr)\biggr]+o_{L_1}(1),
\end{align*}
where the last step follows from
\begin{align*}
	\Expe\biggl|\sum_{i=1}^p \Expe_{k}a_{ii}^{(1)}(z_1)\cdot\Expe_{k}a_{ii}^{(1)}(z_2)\biggr|^2
	&\leqslant p\cdot \sum_{i=1}^p\Expe \biggl|\Expe_{k}a_{ii}^{(1)}(z_1)\cdot\Expe_{k}a_{ii}^{(1)}(z_2)\biggr|^2\\
	&\leqslant p\cdot \sum_{i=1}^p\biggl( \Expe \Bigl|\Expe_{k}a_{ii}^{(1)}(z_1)\Bigr|^4\biggr)^{\nicefrac{1}{2}}\cdot \biggl(\Expe\Bigl|\Expe_{k}a_{ii}^{(1)}(z_2)\Bigr|^4\biggr)^{\nicefrac{1}{2}}\\
	&\leqslant p\cdot \sum_{i=1}^p\biggl( \Expe \Bigl|a_{ii}^{(1)}(z_1)\Bigr|^4\biggr)^{\nicefrac{1}{2}}\cdot \biggl(\Expe\Bigl|a_{ii}^{(1)}(z_2)\Bigr|^4\biggr)^{\nicefrac{1}{2}}\\
	&\overset{\eqref{eq:ajj_upper_bound}}{\leqslant} K(n^4p^2+n^2p^3).
\end{align*}

By above estimates, we obtain
\begin{align}
	\widetilde{\Lambda}_n(z_1,z_2) &= \frac{2}{(npb_p)^2}\sum_{k=1}^n\tr \Bigl(\Expe_k\bM_k^{(1)}(z_1)\cdot \Expe_k\bM_k^{(1)}(z_2)\Bigr)+\Bigl[\frac{\tilde{b}_p}{b_p}(\nu_4-3)+2 \Bigr]+o_p(1),\nonumber\\
	&=\frac{2}{n}\sum_{k=1}^{n}\bbZ_k + \frac{\tilde{b}_p}{b_p}(\nu_4-3)+2 + o_p(1),\label{eq:Gamma_tilde}
\end{align}
where \begin{equation}\label{eq:Zk_defi}
\bbZ_k=\frac{1}{n(pb_p)^2} \tr \Bigl(\Expe_k\bM_k^{(1)}(z_1)\cdot \Expe_k\bM_k^{(1)}(z_2)\Bigr).
\end{equation}
In Lemma \ref{lem:Zk_asym}, we derive the asymptotic expression of $\bbZ_k$. This asymptotic expression ensures that
\begin{equation}\label{eq:Zk_limit}
	\frac{1}{n}\sum_{k=1}^n\bbZ_k
	\to \int_0^1 \frac{t m(z_1)m(z_2)}{1-tm(z_1)m(z_2)}\mathrm{d} t = -1-\frac{\log\bigl(1-m(z_1)m(z_2)\bigr)}{m(z_1)m(z_2)}.
\end{equation}
By \eqref{eq:Lambda_2}, \eqref{eq:Lambda_tilde_2}, \eqref{eq:Gamma_tilde} and \eqref{eq:Zk_limit}, we have
\[
	\widetilde{\Lambda}_n(z_1,z_2) \convp \frac{\omega}{\theta}(\nu_4-3)-\frac{2\log\bigl(1-m(z_1)m(z_2)\bigr)}{m(z_1)m(z_2)}.
\]
Therefore, 
\begin{align*}
	\Lambda_n(z_1,z_2) & \convp \frac{\partial^2}{\partial z_1\partial z_2} \biggl\{\frac{\omega}{\theta}(\nu_4-3)m(z_1)m(z_2) - 2\log\Bigl( 1-m(z_1)m(z_2)\Bigr)\biggr\} \nonumber\\[0.5em]
	&=m'(z_1)m'(z_2) \biggl[ \frac{\omega}{\theta}(\nu_4-3)+2\bigl( 1-m(z_1)m(z_2) \bigr)^{-2}  \biggr].
\end{align*}


\subsection{Tightness of $M_n^{(1)}(z)$}\label{sec:tight_Mn1}

This subsection is to verify the tightness of $M_n^{(1)}(z)$ for $z\in\mathbb{C}_1$ by using Lemma~\ref{lem:tightness}. Applying the Cauchy-Schwarz inequality, Lemma~\ref{lem:betak_upper_bounds} and Lemma~\ref{lem:gamma_moment_upper_bound}, we have
\[
	\Expe\biggl|\sum_{k=1}^n\sum_{j=1}^{r} a_j Y_k(z_j)\biggr|^2 = O(1),
\]
which shows that the condition (i) of Lemma~\ref{lem:tightness} holds. Condition (ii) of Lemma~\ref{lem:tightness} will be verified by showing 
\begin{equation}\label{eq:Mn1_Mn2_diff_bounded}
	\frac{\Expe\bigl|M_n^{(1)}(z_1)-M_n^{(1)}(z_2)\bigr|^2}{|z_1-z_2|^2} \leqslant K,\qquad z_1,z_2\in\mathbb{C}_1.
\end{equation}
The proof of \eqref{eq:Mn1_Mn2_diff_bounded} exactly follow \citet{chen2015clt}, it is then omitted.

\subsection{Convergence of $M_n^{(2)}(z)$}\label{sec:conv_Mn2}

In this section, we obtain the asymptotic expansion of $n\bigl(\Expe\, m_n(z) - m(z)\bigr)$ for $z\in\mathbb{C}_1$ (see definition $\mathbb{C}_1$ of in Section~\ref{sec:strategy_of_proof}) and the result is stated in Lemma~\ref{lem:Mn2_mean}.
This lemma, together with the finite dimensional convergence (see Section~\ref{sec:conv_Mn1_proof}) and the tightness of $M_n^{(1)}(z)$ (see Section~\ref{sec:tight_Mn1}), implies Proposition~\ref{prop:Mn_CLT}. To prove Lemma~\ref{lem:Mn2_mean}, we will follow the strategy in \citet{khorunzhy1996asymptotic} and \citet{bao2015asymptotic}. The main tool is the generalized Stein's equation (see Lemma~\ref{lem:general-stein}).

\begin{lemma}\label{lem:Mn2_mean}
	With the same notations as in the previous sections, for $z\in\mathbb{C}_1$,
	\begin{enumerate}
		\item[(1)] if $p\wedge n\to\infty$ and $n^2/p=O(1)$, we have
		\begin{equation}\label{eq:n2p_mean_expan}
			M_n^{(2)}=n\Bigl[\Expe m_n(z) -m(z)-\calX_n\bigl(m(z)\bigr)\Bigr] = o(1),
		\end{equation}
		where $\calX_n(m)$ is defined by \eqref{eq:calX};
		\item[(2)] if $p\wedge n\to\infty$ and $n^3/p = O(1)$, we have
		\begin{equation}\label{eq:n3p_mean_expan}
			n\Biggl[\Expe m_n(z) -m(z) +\sqrt{\frac{n}{p}}\frac{c_p}{b_p\sqrt{b_p}}\frac{m^4}{1-m^2} \Biggr]= \frac{m^3}{1-m^2}\biggl(\frac{m^2}{1-m^2}+\frac{\tilde{b}_p}{b_p}(\nu_4-3)+1\biggr) +o(1).
		\end{equation}
	\end{enumerate}
\end{lemma}

\begin{proof}
	Let $\bY=(npb_p)^{-1/4}\bX$, then
		\[
			\bA = \bY'\bSigma_p\bY-\sqrt{\frac{p}{n}}\frac{a_p}{\sqrt{b_p}}\bI_n.
		\]
		To simplify notations, we let
		\[
			\bE:= \bSigma_p\bY\bD\bY'\bSigma_p = (E_{ij})_{p\times p}, \qquad\bF:= \bSigma_p\bY\bD = (F_{ij})_{p\times n}.
		\]
		By the basic identity 
		\[
			\bD  = -\frac{1}{z}\bI_n + \frac{1}{z} \bD\bA
			= -\frac{1}{z}\bI_n + \frac{1}{z}  \biggl(\bD\bY'\bSigma_p\bY - \sqrt{\frac{p}{n}}\frac{a_p}{\sqrt{b_p}}\bD\biggr),
		\]
		we have
		\begin{align}
			\Expe m_n(z)  & = -\frac{1}{z} + \frac{1}{z}\cdot \frac{1}{n} \Expe \tr\bigl(\bD\bA\bigr)\nonumber\\
			& = -\frac{1}{z} - \frac{1}{z} \sqrt{\frac{p}{n}}\frac{a_p}{\sqrt{b_p}} \Expe \biggl(\frac{1}{n}\tr \bD\biggr) +  \frac{1}{zn}\Expe \;\tr\Bigl(\bY'\bSigma_p\bY\bD\Bigr)\nonumber \\
			& = -\frac{1}{z} - \frac{1}{z} \sqrt{\frac{p}{n}}\frac{a_p}{\sqrt{b_p}} \Expe m_n(z) + \frac{1}{zn}\sum_{j,k}\Expe \bigl( Y_{jk}F_{jk} \bigr).\label{eq:mn-decom}
		\end{align}
		The basic idea of the following derivation is regarding $F_{jk}:=(\bSigma_p\bY\bD)_{jk}$ as an analytic function of $Y_{jk}$, and then apply the generalized Stein's equation (Lemma~\ref{lem:general-stein} below) to expand $\Expe \bigl( Y_{jk}F_{jk} \bigr)$.
		
		\begin{lemma}[Generalized Stein's Equation, \citet{khorunzhy1996asymptotic}]\label{lem:general-stein}
			For any real-valued random variable $\xi$ with $\Expe |\xi|^{p+2}<\infty$ and complex-valued function $g(t)$ with continuous and bounded $p+1$ derivatives, we have
			\[
				\Expe \bigl[\xi g(\xi)\bigr] = \sum_{a=0}^{p} \frac{\kappa_{a+1}}{a!} \Expe \Bigl(g^{(a)}(\xi)\Bigr)+\varepsilon,	
			\]
			where $\kappa_{a}$ is the $a$-th cumulant of $\xi$, and 
			\[
				|\varepsilon|\leqslant C \sup_t \bigl| g^{(p+1)}(t) \bigr| \; \Expe\bigl( |\xi|^{p+2}\bigr),	
			\]
			where the positive constant $C$ depends on $p$.
		\end{lemma}
		Applying Lemma~\ref{lem:general-stein} to the last term in \eqref{eq:mn-decom}, we obtain the following expansion:
		\begin{equation}\label{eq:Emn_expan_0}
			\Expe m_n(z)
			= -\frac{1}{z} - \frac{1}{z} \sqrt{\frac{p}{n}}\frac{a_p}{\sqrt{b_p}} \Expe m_n(z) + \frac{1}{zn}\sum_{a=0}^4\frac{1}{(npb_p)^{(a+1)/4}}\sum_{j,k} \frac{\kappa_{a+1}}{a!}\Expe\biggl(\frac{\partial^a F_{jk}}{\partial Y_{jk}^a}\biggr)+\varepsilon_n,
		\end{equation}
		where $\kappa_{a}$ is the $a$-th cumulant of $Y_{jk}$, $\tfrac{\partial^a F_{jk}}{\partial Y_{jk}^a}$ denotes the $a$-th order derivative of $F_{jk}$ w.r.t. $Y_{jk}$, and
		\begin{equation}\label{eq:stein_remainder_up_bdd}
			\varepsilon_n \leqslant  \frac{K}{n} \frac{1}{(npb_p)^{6/4}}\sum_{j,k}\sup_{j,k} \Expe_{jk} 	\biggl|\frac{\partial^5 F_{jk}}{\partial Y_{jk}^5}\biggr|.
		\end{equation}
		The explicit formula of the derivatives of $F_{jk}$ are provided in Lemma~\ref{lem:Fjk_derivative}. These derivatives can be derived by using the chain rule and Lemma~\ref{lem:derivative} repeatedly, and the details will be omitted here.
		\begin{lemma}[Derivatives of $F_{jk}$]\label{lem:Fjk_derivative}
			\begin{align*}
				\frac{\partial  F_{jk}}{\partial Y_{jk}} \;&= \sigma_{j j} D_{kk} - E_{jj}D_{kk}-F_{jk}^2;\\[0.5em]
				\frac{\partial^2  F_{jk}}{\partial Y_{jk}^2} \;&= -6\sigma_{jj}F_{jk}D_{kk}+6E_{jj}F_{jk}D_{kk}+2F_{jk}^3;\\[0.5em]
				\frac{\partial^3  F_{jk}}{\partial Y_{jk}^3} \;&= -6\sigma_{jj}^2D_{kk}^2+36\sigma_{jj}F_{jk}^2D_{kk}+12\sigma_{jj}	E_{jj}D_{kk}^2 - 36E_{jj}F_{jk}^2D_{kk}-6E_{jj}^2D_{kk}^2-6F_{jk}^4;\\[0.5em]
				\frac{\partial^4  F_{jk}}{\partial Y_{jk}^4} \;&= 120\sigma_{jj}^2F_{jk}D_{kk}^2 -240\sigma_{jj}F_{jk}^3D_{kk}-240\sigma_{jj}E_{jj}F_{jk}D_{kk}^2+240E_{jj}F_{jk}^3D_{kk}\\[0.5em]
				&\qquad\qquad+120E_{jj}^2F_{jk}D_{kk}^2+24F_{jk}^5;\\
				\frac{\partial^5  F_{jk}}{\partial Y_{jk}^5} \;&= -120 F_{jk}^6 - 1800 E_{jj} F_{jk}^4 D_{kk} - 1800 E_{jj}^2 F_{jk}^2 D_{kk}^2 - 120 E_{jj}^3 D_{kk}^3 + 
				1800 \sigma_{jj}F_{jk}^4 D_{kk}  \\[0.5em]
				&\qquad\qquad + 3600 \sigma_{jj} E_{jj} F_{jk}^2 D_{kk}^2  + 360 \sigma_{jj} E_{jj}^2 D_{kk}^3  - 
				1800 \sigma_{jj}^2 F_{jk}^2 D_{kk}^2\\[0.5em]
				&\qquad \qquad  - 360 \sigma_{jj}^2 E_{jj} D_{kk}^3  + 120 \sigma_{jj}^3 D_{kk}^3.
			\end{align*}
		\end{lemma}

		From \eqref{eq:DXSX} and Lemma~\ref{lem:G11_tilde_upper_bbd}, it is not difficult to obtain the following estimates:
		\begin{align}
			D_{kk}^{a_1}F_{jk}^{a_2}E_{jj}^{a_3}& \leqslant Kn^{a_3/2}\biggl(\sum_{\alpha}\bigl[(\bSigma \bY)_{j\alpha}\bigr]^2\biggr)^{(a_2+2a_3)/2},\quad (a_1, a_2, a_3\geqslant 0)\label{eq:stein_esti_EFG}\\[0.5em]
			\Expe\biggl|\Bigl(\bSigma_p^{-\nicefrac{1}{2}}\bE\bSigma_p^{-\nicefrac{1}{2}}\Bigr)_{jj}&-\frac{\Expe m_n}{a_p\sqrt{p/(nb_p)}+z+\Expe m_n}\biggr|=O\biggl(\biggl(\frac{n}{p}\biggr)^2\biggr) + O\biggl(\frac{1}{p}\biggr),\label{eq:stein_esti_E}\\[0.5em]
			\biggl|\sum_{j,k}F_{jk}\biggr| & = O\Bigl((np)^{3/4}\Bigr),\label{eq:stein_esti_F_1}\\[0.5em]
			\biggl|\sum_{j,k}F_{jk}^{a_2}\biggr| & = O\Bigl(p^{a_2/4}n^{1-a_2/4}\Bigr),\quad (a_2\geqslant 2).\label{eq:stein_esti_F_2}
		\end{align}
		By \eqref{eq:stein_remainder_up_bdd} and \eqref{eq:stein_esti_EFG}, we have the following estimate for $\varepsilon_n$ defined in \eqref{eq:stein_remainder_up_bdd}:
		\begin{equation}\label{eq:stein_remainder_esti}
			|\varepsilon_n| = o\biggl(\frac{1}{n}\biggr).
		\end{equation}
		By the fact \begin{equation}
			\Expe \biggl|D_{\ell\ell}+\frac{1}{z+\Expe\, m_n}\biggr|^2 = O\biggl(\frac{1}{n}\biggr) + O\biggl(\frac{n}{p}\biggr),\qquad 1\leqslant \ell \leqslant n,\label{eq:Gkk_esti}
		\end{equation} 
		which is verified in Lemma~\ref{lem:G11_upper_bbd}, and the estimates above, we can extract the leading order terms in \eqref{eq:Emn_expan_0} to obtain
		\begin{align}
			\Expe m_n(z)
			& = -\frac{1}{z} - \frac{1}{z} \sqrt{\frac{p}{n}}\frac{a_p}{\sqrt{b_p}} \Expe m_n(z) +  \frac{1}{zn}\frac{1}{\sqrt{npb_p}}\sum_{j,k}\Expe\Bigl(\sigma_{jj}D_{kk}-E_{jj}D_{kk}-F_{jk}^2\Bigr)\nonumber\\
			&\qquad\qquad  - \frac{1}{zn}\frac{\nu_4-3}{npb_p}\sum_{j,k}\Expe\Bigl(\sigma_{jj}^2 D_{kk}^2\Bigr)+o\biggl(\frac{1}{n}\biggr)\nonumber\\
			& = -\frac{1}{z} - \frac{1}{zn}\frac{1}{\sqrt{npb_p}} \Expe\Bigl[\tr(\bE)\tr(\bD)\Bigr]  - \frac{1}{zn}\frac{1}{\sqrt{npb_p}} \Expe\Bigl[\tr(\bF\bF')\Bigr]\nonumber\\
			&\qquad\qquad - \frac{(\nu_4-3)\tilde{b}_p}{zn^2b_p}\Expe\Bigl( \sum_{k} D_{kk}^2\Bigr)+o\biggl(\frac{1}{n}\biggr).\label{eq:Emn_expan_1}
		\end{align}
		Using the same argument as in the proof of Lemma~\ref{lem:Mk_limit}, if $n^2/p=O(1)$, we can show that
		\begin{equation}
			\Expe \biggl| \frac{1}{\sqrt{npb_p}} \tr \bE - m(z) \biggr|^2
			=O\biggl(\frac{1}{n}\biggr).\label{eq:Diff-trE-m}
		\end{equation}
		This, together with $c_r$-inequality, implies that
		\begin{equation}\label{eq:Diff-trE-Expe}
			\Expe\, \biggl| \frac{1}{\sqrt{npb_p}} \tr \bE - \Expe\,\frac{1}{\sqrt{npb_p}} \tr \bE  \biggr|^2 = o(1).
		\end{equation}
		Together with the fact that $\Var(m_n)=O(n^{-2})$ (see Lemma~\ref{lem:var_mn_estimate} for more details), we obtain
		\begin{equation}
			\Cov\biggl(\frac{1}{\sqrt{npb_p}} \tr \bE, \frac{1}{n}\tr \bD \biggr) \leqslant \sqrt{\eqref{eq:Diff-trE-Expe}}\cdot \sqrt{\Var(m_n)} = o\biggl(\frac{1}{n}\biggr).\label{eq:trE_trG_Cov}
		\end{equation}
		Note that 
		\begin{equation}
		\tr(\bF\bF') = \tr\bigl(\bSigma_p\bY\bD^{2}\bY'\bSigma_p\bigr) = \frac{\partial }{\partial z} \tr\bigl(\bSigma_p\bY\bD\bY'\bSigma_p\bigr)	= \frac{\partial }{\partial z} \tr \bE .\label{eq:tr_FFt_dif}
		\end{equation}
		Applying \eqref{eq:Gkk_esti}, \eqref{eq:trE_trG_Cov}, and \eqref{eq:tr_FFt_dif} to \eqref{eq:Emn_expan_1}, we have
		\begin{align}
			\Expe m_n(z)
			& = -\frac{1}{z} - \frac{1}{z}\cdot\frac{1}{\sqrt{npb_p}} \Expe\bigl(\tr \bE \bigr)\cdot\frac{1}{n}\Expe\,\bigl(\tr \bD \bigr)  - \frac{1}{zn}\frac{1}{\sqrt{npb_p}} \Expe\,\biggl(\frac{\partial }{\partial z} \tr \bE \biggr)\nonumber\\
			&\qquad\qquad -\frac{\nu_4-3}{zn}\frac{\tilde{b}_p}{b_p}\bigl(\Expe m_n(z)\bigr)^2+o\biggl(\frac{1}{n}\biggr).\label{eq:Emn_expan_2}
		\end{align}
		The problem reduces to estimate $(1/\sqrt{npb_p})\Expe\bigl(\tr \bE \bigr)$. To this end, we apply Lemma~\ref{lem:general-stein} again to the term $(1/\sqrt{npb_p})\Expe\bigl(\tr \bE \bigr)$ to find its expansion.
		Denote
		\[
			\widetilde{\bE}:= \bSigma_p^2\bY\bD\bY'\bSigma_p^2, \qquad \widehat{\bE}:=\bSigma_p\bY\bD\bY'\bSigma_p^2,\qquad\widetilde{\bF}:= \bSigma_p^2\bY\bD,
		\]
		and write 
		\begin{equation}\label{eq:Expe_trE_expan}
			\frac{1}{\sqrt{npb_p}}\Expe\,\bigl(\tr \bE\bigr) =\frac{1}{\sqrt{npb_p}}\sum_{j,k} \Expe\bigl(Y_{jk}\widetilde{F}_{jk}\bigr).
		\end{equation}
		The first four derivatives of $\widetilde{F}_{jk}$ w.r.t. $Y_{jk}$ is presented in the following lemma.
		\begin{lemma}[Derivatives of $\widetilde{F}_{jk}$]\label{lem:Fjk_tilde_derivative}
			\begin{align*}
					\frac{\partial  \widetilde{F}_{jk}}{\partial Y_{jk}} \;& = \widetilde{\sigma}_{j j} D_{kk} - \widehat{E}_{jj}D_{k k}-F_{jk}\widetilde{F}_{jk};\\[0.5em]
					\frac{\partial^2  \widetilde{F}_{jk}}{\partial Y_{jk}^2} \;& = -2\sigma_{jj}\widetilde{F}_{jk}D_{kk}-4\widetilde{\sigma}_{jj}F_{jk}D_{kk} + 2F_{jk}^2\widetilde{F}_{jk}+4\widehat{E}_{jj}F_{jk}D_{kk}+2E_{jj}\widetilde{F}_{jk}D_{kk};\\[0.5em]
					\frac{\partial^3  \widetilde{F}_{jk}}{\partial Y_{jk}^3} \;& = -6\widetilde{\sigma}_{jj}\sigma_{jj}D_{kk}^2-6F_{jk}^3\widetilde{F}_{jk}-18\widehat{E}_{jj}F_{jk}^2D_{kk}-18E_{jj}F_{jk}\widetilde{F}_{jk}D_{kk}-6E_{jj}\widehat{E}_{jj}D_{kk}^2\\[0.5em]
					\;&\qquad +18\sigma_{jj}F_{jk}\widetilde{F}_{jk}D_{kk} + 6\sigma_{jj}\widehat{E}_{jj}D_{kk}^2+18\widetilde{\sigma}_{jj}F_{jk}^2D_{kk}+6\widetilde{\sigma}_{jj}E_{jj}D_{kk}^2;\\[0.5em]
					\frac{\partial^4  \widetilde{F}_{jk}}{\partial Y_{jk}^4}
					\;& = 24 F_{jk}^4 \widetilde{F}_{jk} + 96 \widehat{E}_{jj}F_{jk}^3D_{kk} +144 E_{jj} F_{jk}^2 \widetilde{F}_{jk} D_{kk} + 96 E_{jj}\widehat{E}_{jj} F_{jk} D_{kk}^2 + 24 E_{jj}^2 \widetilde{F}_{jk} D_{kk}^2 \\[0.5em]
					\; & \qquad - 144\sigma_{jj} F_{jk}^2 \widetilde{F}_{jk} D_{kk} - 96 \sigma_{jj}\widehat{E}_{jj}F_{jk}D_{kk}^2 - 48 \sigma_{jj}E_{jj} \widetilde{F}_{jk} D_{kk}^2 
					+24 \sigma_{jj}^2 \widetilde{F}_{jk} D_{kk}^2 \\[0.5em]
					\;&\qquad - 96 \widetilde{\sigma}_{jj}F_{jk}^3 D_{kk}- 96\widetilde{\sigma}_{jj}E_{jj} F_{jk} D_{kk}^2 + 96 \sigma_{jj} \widetilde{\sigma}_{jj} F_{jk} D_{kk}^2.
			\end{align*}
		\end{lemma}
		Applying generalized Stein's equation with the derivatives of $\widetilde{F}_{jk}$ (see Lemma~\ref{lem:Fjk_tilde_derivative}) to the last term in \eqref{eq:Expe_trE_expan}, and using the similar estimates above, gives us
		\begin{align}
			\;&\frac{1}{\sqrt{npb_p}}\Expe\,\bigl(\tr \bE\bigr) \nonumber\\
			=\;&\frac{1}{\sqrt{npb_p}}\sum_{a=0}^3\frac{1}{(npb_p)^{(a+1)/4}}\sum_{j,k} \frac{\kappa_{a+1}}{a!}\Expe\biggl(\frac{\partial^a \widetilde{F}_{jk}}{\partial Y_{jk}^a}\biggr) + \widetilde{\varepsilon}_n\nonumber\\
			=\;&\frac{1}{npb_p}\sum_{j,k}\Expe\Bigl( \widetilde{\sigma}_{j j} D_{kk} - \widehat{E}_{jj}D_{k k}-F_{jk}\widetilde{F}_{jk}\Bigr)+\frac{1}{\sqrt{npb_p}}\frac{\nu_4-3}{npb_p}\sum_{j,k}\Expe\Bigl(-\widetilde{\sigma}_{jj}\sigma_{jj}D_{kk}^2\Bigr)+o\biggl(\frac{1}{n}\biggr)\nonumber\\
			=\;&\Expe m_n- \frac{1}{npb_p} \Expe\Bigl[\tr(\widehat{\bE})\tr(\bD)\Bigr]  - \frac{1}{npb_p} \Expe\Bigl[\tr(\bF\widetilde{\bF}')\Bigr]\nonumber \\
			\;&\qquad\qquad\qquad\qquad- \frac{1}{\sqrt{npb_p}}\frac{\nu_4-3}{npb_p}\Expe\Bigl(\sum_{j}\widetilde{\sigma}_{jj}\sigma_{jj}\Bigr)\Bigl(\sum_{k}D_{kk}^2\Bigr)+o\biggl(\frac{1}{n}\biggr)\nonumber\\
			=\;&\Expe m_n- \sqrt\frac{n}{p}\cdot \frac{1}{\sqrt{np}b_p} \Expe\bigl(\tr \widehat{\bE}\bigr)\cdot\frac{1}{n}\Expe\bigl( \tr \bD\bigr) +o\biggl(\frac{1}{n}\biggr)\label{eq:Expe-trE-hat-trD}\\
			=\;& \Expe m_n -\sqrt{\frac{n}{p}}\biggl(\frac{c_p}{b_p\sqrt{b_p}}\Expe\, m_n + O\Bigl(\sqrt{\tfrac{n}{p}}\Bigr)\biggr)\cdot \Expe\, m_n+o\biggl(\frac{1}{n}\biggr)\label{eq:Expe-trace-E-hat-esti}\\
			=\;& \Expe m_n -\sqrt{\frac{n}{p}}\frac{c_p}{b_p\sqrt{b_p}}\bigl(\Expe\, m_n\bigr)^2+o\biggl(\frac{1}{n}\biggr)+o\biggl(\sqrt{\frac{n}{p}}\biggr).\label{eq:Expe-trace-E}
		\end{align}
		\noindent Plugging \eqref{eq:Expe-trace-E} into \eqref{eq:Emn_expan_2}, we have
			\begin{align*}
				\Expe m_n
				& = -\frac{1}{z} - \frac{1}{z}\biggl[\Expe m_n -\sqrt{\frac{n}{p}}\frac{c_p}{b_p\sqrt{b_p}}\bigl(\Expe m_n\bigr)^2 \biggr]\Expe m_n\\[0.5em]
				& \qquad\qquad-\frac{1}{zn}\cdot \frac{\partial}{\partial z} \biggl[ \Expe m_n-\sqrt{\frac{n}{p}}\frac{c_p}{b_p\sqrt{b_p}}\cdot \bigl(\Expe m_n\bigr)^2 \biggr]\\[0.5em]
				& \qquad\qquad-\frac{\nu_4-3}{zn}\frac{\tilde{b}_p}{b_p}\bigl(\Expe m_n\bigr)^2 + o\biggl(\frac{1}{n}\biggr)+o\biggl(\sqrt{\frac{n}{p}}\biggr),\\[0.5em]
				& = -\frac{1}{z} -\frac{1}{z} \bigl(\Expe\, m_n\bigr)^2 + \frac{1}{z}\sqrt{\frac{n}{p}}\frac{c_p}{b_p\sqrt{b_p}} m^3 \\[0.5em]
				&\qquad\qquad- \frac{1}{zn}\biggl[ \frac{m^2}{1-m^2}+\frac{(\nu_4-3)\tilde{b}_p}{b_p}m^2 \biggr]+ o\biggl(\frac{1}{n}\biggr)+o\biggl(\sqrt{\frac{n}{p}}\biggr).
			\end{align*}
			This implies \eqref{eq:n3p_mean_expan} under the assumption $n^3/p=O(1)$.

		Moreover, to obtain \eqref{eq:n2p_mean_expan} under the assumption $n^2/p = O(1)$, we need to figure out the remainder term $o(\sqrt{n/p})$ in \eqref{eq:Expe-trace-E} more carefully. Indeed, this remainder term comes from the estimate of $\Expe\bigl(\tr \widehat{\bE}\bigr)/(b_p\sqrt{np})$ in \eqref{eq:Expe-trE-hat-trD}. To get a more precise estimate, we use the similar argument above for calculating the asymptotic expansion of $\Expe\bigl(\tr \widehat{\bE}\bigr)/(b_p\sqrt{np})$.

		\begin{equation}\label{eq:Expe-trace-E-hat-expand}
			\frac{1}{\sqrt{np}b_p} \Expe\bigl(\tr \widehat{\bE}\bigr) = \Expe m_n - \sqrt{\frac{n}{p}}\frac{c_p}{b_p\sqrt{b_p}} (\Expe m_n)^2 + \frac{n}{p}\frac{d_p}{b_p^2} (\Expe m_n)^3 + o\biggl(\frac{1}{n}\biggr) + o\biggl(\frac{n}{p}\biggr).
		\end{equation}
		Plugging \eqref{eq:Expe-trace-E-hat-expand} into \eqref{eq:Emn_expan_2}, we have
		\begin{align*}
			\Expe m_n
			& = -\frac{1}{z} - \frac{1}{z}\biggl[\Expe m_n -\sqrt{\frac{n}{p}}\frac{c_p}{b_p\sqrt{b_p}}\bigl(\Expe m_n\bigr)^2 + \frac{n}{p}\frac{d_p}{b_p^2} (\Expe m_n)^3\biggr]\Expe m_n\\[0.5em]
			& \qquad\qquad-\frac{1}{zn}\cdot \frac{\partial}{\partial z} \biggl[ \Expe m_n-\sqrt{\frac{n}{p}}\frac{c_p}{b_p\sqrt{b_p}}\cdot \bigl(\Expe m_n\bigr)^2 \biggr]-\frac{\nu_4-3}{zn}\frac{\tilde{b}_p}{b_p}\bigl(\Expe m_n\bigr)^2 + o\biggl(\frac{1}{n}\biggr)+o\biggl(\frac{n}{p}\biggr),\\[0.5em]
			& = -\frac{1}{z} -\frac{1}{z} \bigl(\Expe\, m_n\bigr)^2 + \frac{1}{z}\sqrt{\frac{n}{p}}\frac{c_p}{b_p\sqrt{b_p}} \bigl(\Expe m_n\bigr)^3 - \frac{1}{z}\frac{n}{p}\frac{d_p}{b_p^2}m^4 \\[0.5em]
			&\qquad\qquad - \frac{1}{zn}\biggl[ \frac{m^2}{1-m^2}+\frac{(\nu_4-3)\tilde{b}_p}{b_p}m^2 \biggr]+ o\biggl(\frac{1}{n}\biggr)+o\biggl(\frac{n}{p}\biggr).
		\end{align*}
		Multiplying $-z$ on both sides, we have
		\begin{align}
			-z\Expe m_n
			&=1 + \bigl(\Expe\, m_n\bigr)^2 - \sqrt{\frac{n}{p}}\frac{c_p}{b_p\sqrt{b_p}} \bigl(\Expe m_n\bigr)^2\cdot\Expe m_n 
			+\frac{n}{p}\frac{d_p}{b_p^2}m^4 \nonumber \\[0.5em]
			&\qquad \qquad + \frac{1}{n}\biggl[ \frac{m^2}{1-m^2}+\frac{(\nu_4-3)\tilde{b}_p}{b_p}m^2 \biggr]+ o\biggl(\frac{1}{n}\biggr)+o\biggl(\frac{n}{p}\biggr).\label{eq:zEmn}
		\end{align}
		This implies that 
		\begin{equation}\label{eq:Emn2}
			\bigl(\Expe\, m_n\bigr)^2
			=-1 - z\Expe m_n + \sqrt{\frac{n}{p}}\frac{c_p}{b_p\sqrt{b_p}} \bigl(\Expe m_n\bigr)^3 + O\biggl(\frac{1}{n}\biggr)+O\biggl(\frac{n}{p}\biggr).
		\end{equation}
		Plugging \eqref{eq:Emn2} into \eqref{eq:zEmn} yields that
		\begin{align*}
			-z\Expe m_n
			& =1 + \bigl(\Expe\, m_n\bigr)^2 + \sqrt{\frac{n}{p}}\frac{c_p}{b_p\sqrt{b_p}} \bigl(1 + z\Expe m_n\bigr)\Expe m_n - \frac{n}{p}\frac{c_p^2}{b_p^3} m^4 \\[0.5em]
			& \qquad \qquad + \frac{n}{p}\frac{d_p}{b_p^2}m^4 + \frac{1}{n}\biggl[ \frac{m^2}{1-m^2}+\frac{(\nu_4-3)\tilde{b}_p}{b_p}m^2 \biggr]+ o\biggl(\frac{1}{n}\biggr)+o\biggl(\frac{n}{p}\biggr).
		\end{align*}
		This equation can be written as a quadratic equation of $\Expe\, m_n - m$:
		\begin{align*}
			0 &= \biggl(m-\sqrt{\frac{n}{p}}\frac{c_p}{b_p\sqrt{b_p}}\bigl(1+m^2\bigr)\biggr)\bigl(\Expe\, m_n - m\bigr)^2 \\[0.5em] 
			&\qquad \qquad+ \biggl[m^2-1-\biggl(\sqrt{\frac{n}{p}}\frac{c_p}{b_p\sqrt{b_p}}\biggr)m(1+2m^2)\biggr] \bigl(\Expe\, m_n - m  \bigr) \\[0.5em]
			&\qquad \qquad + \frac{m^3}{n}\biggl[ \frac{1}{1-m^2}+\frac{(\nu_4-3)\tilde{b}_p}{b_p} \biggr]  - \biggl(\sqrt{\frac{n}{p}}\frac{c_p}{b_p\sqrt{b_p}}\biggr) m^4 + \frac{n}{p}\biggl(-\frac{c_p^2}{b_p^3}  +\frac{d_p}{b_p^2}\biggr)m^5 \\[0.5em]
			&\qquad\qquad + o\biggl(\frac{1}{n}\biggr)+o\biggl(\frac{n}{p}\biggr)\\[0.5em]
			& =: \calA \bigl(\Expe\, m_n - m\bigr)^2 + \calB \bigl(\Expe\, m_n - m\bigr) + \calC + o\biggl(\frac{1}{n}\biggr)+o\biggl(\frac{n}{p}\biggr).
		\end{align*}
		Solving the equation yields two solutions:
		\[
			x_1 = \frac{-\calB+\sqrt{\calB^2-4\calA\calC}}{2\calA} + o\biggl(\frac{1}{n}\biggr)+o\biggl(\frac{n}{p}\biggr), \qquad	x_2 = \frac{-\calB-\sqrt{\calB^2-4\calA\calC}}{2\calA} + o\biggl(\frac{1}{n}\biggr)+o\biggl(\frac{n}{p}\biggr).
		\]
		When $n^2/p=O(1)$, from the definition of $\calA, \calB, \calC$, we can verify that $x_1 = o(1)$ while $x_2=(1-m^2)/m^2+o(1)$. 
		Since $\Expe\, m_n - m=o(1)$, we choose $x_1$ to be the expression of $\Expe\, m_n - m$, that is,  
		\[
			\Expe\, m_n - m = \frac{-\calB+\sqrt{\calB^2-4\calA\calC}}{2\calA} + o\biggl(\frac{1}{n}\biggr)+o\biggl(\frac{n}{p}\biggr).
		\]
		This implies \eqref{eq:n2p_mean_expan} under the assumption $n^2/p=O(1)$.
\end{proof}

\newpage
\appendix

\section{Some Technical Lemmas}


\begin{lemma}[\citet{bai2010spectral}, Lemma B.26]\label{lem:quad-trace}
	Let $\bA = (a_{ij})$ be an $n\times n$ nonrandom matrix and $\bx=(X_1, \ldots, X_n)'$ be a random vector of independent entries. Assume that $\Expe X_i = 0$, $\Expe |X_i|^2=1$ and $\Expe |X_i|^{\ell}\leqslant \nu_{\ell}$. Then, for any $k\geqslant 1$,
	\[
		\Expe |\bx^*\bA\bx - \tr \bA|^k \leqslant C_k \Bigl( \bigl(\nu_4 \tr(\bA\bA^*)\bigr)^{k/2} + \nu_{2k} \tr(\bA\bA^*)^{k/2} \Bigr),
	\]
	where $C_k$ is a constant depending on $k$ only.
\end{lemma}

\begin{lemma}[\citet{pan2011central}, Lemma 5]\label{lem:PanZhou-trace}
	Let $\bA$ be a $p\times p$ deterministic complex matrix with zero diagonal elements. Let $\bx=(X_1, \ldots, X_p)'$ be a random vector with i.i.d. real entries. Assume that $\Expe X_i = 0$, $\Expe |X_i|^2=1$. Then, for any $k\geqslant 2$,
	\begin{equation}
		\Expe |\bx'\bA\bx|^k \leqslant C_k \Bigl(\Expe |X_1|^k\Bigr)^2 \bigl( \tr \bA\bA^* \bigr)^{k/2},
	\end{equation}
	where $C_k$ is a constant depending on $k$ only.
\end{lemma}

\begin{lemma}[Burkholder's inequality, \citet{Burkholder1973Distribution}]\label{lem:Burkholder_ineq}
	Let $\{X_i\}$ be a complex martingale difference sequence withe respect to the increasing $\sigma$-field $\{\mathcal{F}_i\}$. Then for $k\geqslant 2$, the following inequality
	\[
		\Expe\biggl|\sum_{i} X_i\biggr|^k \leqslant C_k \Expe \biggl(\sum_{i}\Expe \Bigl(|X_i|^2\;\big|\;\mathcal{F}_{i-1}\Bigr)\biggr)^{k/2} + C_k \Expe\sum_i |X_i|^k
	\]
	holds, where $C_k$ is a constant depending on $k$ only.
\end{lemma}

\begin{lemma}[Martingale CLT, \cite{billingsley2008probability}]\label{lem:CLT_MDS}
	Suppose for each $n$, $Y_{n1}, Y_{n2}, \ldots, Y_{nr_n}$ is a real martingale difference sequence with respect to the $\sigma$-field $\{\mathcal{F}_{nj}\}$ having second moments. If as $n\to\infty$,
	\begin{equation}\label{eq:CLT_MDS_condition1}
		\sum_{j=1}^{r_n}\Expe(Y_{nj}^2\mid \mathcal{F}_{n,j-1}) \convp \sigma^2,
	\end{equation}
	where $\sigma^2$ is a positive constant, and for each $\varepsilon > 0$,
	\begin{equation}\label{eq:CLT_MDS_condition2}
		\sum_{j=1}^{r_n} \Expe \Bigl(Y_{nj}^2\indicator_{\{|Y_{nj}|\geqslant \varepsilon\}}\Bigr) \to 0,
	\end{equation}
	then
	\[
		\sum_{j=1}^{r_n}Y_{nj} \convd \mathcal{N}(0,\sigma^2).
	\]
\end{lemma}

\begin{lemma}[\citet{billingsley1968convergence}, Theorem 12.3]\label{lem:tightness}
	The sequence $\{X_n\}$ is tight if it satisfies these two conditions:
	\begin{itemize}
		\item[(i)] The sequence $\{X_n(0)\}$ is tight.
		\item[(ii)] There exist constants $\gamma\geqslant 0$ and $\alpha > 1$ and a non-decreasing, continuous function $F$ on $[0,1]$ such that 
		\[
			\Prob\Bigl( \bigl|X_n(t_2)-X_n(t_1)\bigr|\geqslant \lambda \Bigr) \leqslant \frac{1}{\lambda^{\gamma}}\bigl|F(t_2)-F(t_1)\bigr|^{\alpha}
		\] 
		holds for all $t_1, t_2$ and $n$ and all positive $\lambda$.
	\end{itemize}
\end{lemma}



\begin{lemma}\label{lem:Zk_asym} For $z_1, z_2 \in \mathbb{C}^{+}$, 
	\[
	\bbZ_k:=\frac{1}{n(pb_p)^2} \tr \Bigl(\Expe_k\bM_k^{(1)}(z_1)\cdot \Expe_k\bM_k^{(1)}(z_2)\Bigr)=\frac{\frac{k}{n}m(z_1)m(z_2)}{1-\frac{k}{n}m(z_1)m(z_2)}+o_{L_1}(1).
	\]
\end{lemma}
This lemma is used in Section~\ref{sec:conv_Mn1_proof} to derive the finite dimensional convergence of $M_n^{(1)}(z)$.

\begin{proof}
	Let $\{ \be_i, i=1.\ldots,k-1,k+1,\ldots,n \}$ be the $(n-1)$-dimensional unit vectors with the $i$-th (or $(i-1)$-th) element equal to $1$ and the remaining equal to $0$ according as $i<k$ (or $i>k$). Write $\bX_k=\bX_{ki}+\bx_i\be_i'$. Let
	\begin{align*}
		\bD_{ki,r}^{-1}&=\bD_k^{-1}-\be_i\bh_i'=\frac{1}{\sqrt{npb_p}}\Bigl( \bX_{ki}'\bSigma_p\bX_k-pa_p \bI_{(i)}\Bigr)-z\bI_{n-1},\\[0.5em]
		\bD_{ki}^{-1}&=\bD_k^{-1}-\be_i\bh_i'-\br_i\be_i'=\frac{1}{\sqrt{npb_p}}\Bigl(\bX_{ki}'\bSigma_p\bX_{ki}-pa_p\bI_{(i)} \Bigr)-z\bI_{n-1},\\[0.5em]
		\bh_i'&=\frac{1}{\sqrt{npb_p}}\bx_i'\bSigma_p\bX_{ki}+\frac{1}{\sqrt{npb_p}}\Bigl(\bx_i'\bSigma_p\bx_i-pa_p\Bigr)\be_i',\qquad \br_i=\frac{1}{\sqrt{npb_p}} \bX_{ki}'\bSigma_p\bx_i,\\[0.5em]
		\zeta_i&=\frac{1}{1+\vartheta_i},\qquad \vartheta_i=\bh_i'\bD_{ki,r}(z)\be_i,\qquad \bM_{ki}=\bSigma_p\bX_{ki}\bD_{ki}(z)\bX_{ki}'\bSigma_p.
	\end{align*}
	
	We have some crucial identities,
	\begin{equation}\label{eq:Xe_0}
		\bX_{ki}\be_i=\boldsymbol{0},\qquad \be_i'\bD_{ki,r}=\be_{i}'\bD_{ki}=-\frac{\be_{i}'}{z},
	\end{equation}
	where $\boldsymbol{0}$ is a $p$-dimensional vector with all the elements equal to $0$.
	By using \eqref{eq:Xe_0} and some frequently used formulas about the inverse of matrices, we have two useful identites,
	\begin{equation}\label{eq:diff_inv_1}
		\begin{aligned}
			\bD_{k}-\bD_{ki,r}&=-\bD_{ki,r}(\bD_k^{-1}-\bD_{ki,r}^{-1})\bD_k=-\bD_{ki,r}(\be_i\bh_i')\bD_{k}\\[0.5em]
			&=-\bD_{ki,r}(\be_i\bh_i')(\zeta_i\bD_{ki,r})=-\zeta_i\bD_{ki,r}(\be_i\bh_i')\bD_{ki,r}
		\end{aligned}
	\end{equation}
	and
	\begin{equation}\label{eq:diff_inv_2}
		\begin{aligned}
			\bD_{ki,r}-\bD_{ki}&=-\bD_{ki}(\bD_{ki,r}^{-1}-\bD_{ki}^{-1})\bD_{ki,r}=-\bD_{ki}(\br_i\be_i')\bD_{ki,r}\\[0.5em]
			&=-\bD_{ki}\Bigl(\frac{1}{\sqrt{npb_p}} \bX_{ki}'\bSigma_p\bx_i\be_i'\Bigr)\bD_{ki}=\frac{1}{z\sqrt{npb_p}}\bD_{ki}\bX_{ki}'\bSigma_p\bx_i\be_i'.
		\end{aligned}
	\end{equation}
	
	Using \eqref{eq:diff_inv_1} and \eqref{eq:diff_inv_2}, for $i<k$, we obtain the following decomposition $\Expe_k\bM_k^{(1)}(z)$,
	\begin{align}
		\Expe_k\bM_k^{(1)}(z) & = \Expe_k \Bigl(\bSigma_p(\bX_{ki}+\bx_i\be_i')\bD_k(\bX_{ki}+\bx_i\be_i')'\bSigma_p\Bigr)\nonumber\\
		&=\Expe_k\biggl( \bSigma_p\bX_{ki}\bD_k\bX_{ki}'\bSigma_p + \bSigma_p\bX_{ki}\bD_k\be_i\bx_i'\bSigma_p+\bSigma_p\bx_i\be_i'\bD_k\bX_{ki}'\bSigma_p +\bSigma_p\bx_i\be_i'\bD_k\be_i\bx_i'\bSigma_p \biggr)\nonumber\\
		&=\Expe_k \bM_{ki}-\Expe_{k}\biggl( \frac{\zeta_i(z)}{znpb_p} \bM_{ki}\bx_i\bx_i'\bM_{ki}\biggr)+\Expe_k \biggl(\frac{\zeta_i(z)}{z\sqrt{npb_p}}\bM_{ki}\biggr)\bx_i\bx_i'\bSigma_p\nonumber\\
		&\qquad\qquad +\bSigma_p\bx_i\bx_i'\Expe_k \biggl(\frac{\zeta_i(z)}{z\sqrt{npb_p}}\bM_{ki}\biggr)-\Expe_k\biggl(\frac{\zeta_i(z)}{z}\biggr)\bSigma_p\bx_i\bx_i'\bSigma_p\label{eq:E_Mk_decom}\\
		&:=\bB_1(z)+\bB_2(z)+\bB_3(z)+\bB_4(z)+\bB_5(z).\nonumber
	\end{align}
	
	Write
	\[
	\bD_k^{-1}=\sum_{i=1 (\neq k)}^n \be_i\bh_i'-z\bI_{n-1}.
	\]
	Multiplying $\bD_k$ on the right-hand side, we have
	\[
	z\bD_k=-\bI_{n-1}+\sum_{i=1(\neq k)}^n \be_i\bh_i'\bD_k.
	\]
	Multiplying $\bSigma_p\bX_k$ on the left-hand side, $\bX_k'\bSigma_p$ on the right-hand side, we get
	\[
	z\bM_k^{(1)}(z)=-\bSigma_p\bX_k\bX_k'\bSigma_p+\sum_{i=1(\neq k)}^n \bSigma_p\bX_k\be_i\bh_i'\bD_k\bX_k'\bSigma_p.
	\]
	Thus,
	\begin{align}
		z\Expe_k \bigl(\bM_k^{(1)}(z)\bigr)&=-\Expe_k\bigl(\bSigma_p\bX_k\bX_k'\bSigma_p\bigr)+\sum_{i=1(\neq k)}^n \Expe_k(\bSigma_p\bX_k\be_i\bh_i'\bD_k\bX_k'\bSigma_p)\nonumber\\
		&=-\bSigma_p\Expe_k\biggl(\sum_{i=1(\neq k)}^n\bx_i\bx_i'\biggr)\bSigma_p+\sum_{i=1(\neq k)}^n \Expe_k\Bigl(\zeta_i\bSigma_p\bx_i\bh_i'\bD_{ki,r}(\bX_{ki}'+\be_i\bx_i')\bSigma_p\Bigr)\nonumber\\
		&=-(n-k)\bSigma_p^2 - \sum_{i<k} \Bigl(\bSigma_p\bx_i\bx_i'\bSigma_p\Bigr)+
		\sum_{i=1(\neq k)}^n \Expe_{k}\biggl(\frac{\zeta_i}{\sqrt{npb_p}}\bSigma_p\bx_i\bx_i'\bSigma_p\bX_{ki}\bD_{ki,r}\bX_{ki}'\bSigma_p\biggr)\nonumber\\
		&\qquad\qquad+\sum_{i=1(\neq k)}^n \Expe_k \Bigl(\zeta_i\bSigma_p\bx_i\bh_i'\bD_{ki,r}\be_i\bx_i'\bSigma_p\Bigr)\nonumber\\
		&=-(n-k)\bSigma_p^2 - \sum_{i<k} \Bigl(\bSigma_p\bx_i\bx_i'\bSigma_p\Bigr)+\sum_{i=1(\neq k)}^n \Expe_{k}\biggl(\frac{\zeta_i}{\sqrt{npb_p}}\bSigma_p\bx_i\bx_i'\bM_{ki}\biggr)\nonumber\\
		&\qquad\qquad+\sum_{i=1(\neq k)}^n \Expe_k \Bigl(\zeta_i\vartheta_i\bSigma_p\bx_i\bx_i'\bSigma_p\Bigr).\label{eq:z_E_Mk_decom}
	\end{align}
	
	Applying \eqref{eq:E_Mk_decom} and \eqref{eq:z_E_Mk_decom} to $\Expe_k \bM_k^{(1)}(z_2)$ (for $i<k$) and $z_1\Expe_k \bM_k^{(1)}(z_1)$, we get the following decomposition:
	\begin{align}
		z_1\bbZ_k&=\frac{z_1}{n(pb_p)^2} \tr \Bigl(\Expe_k\bM_k^{(1)}(z_1)\cdot \Expe_k\bM_k^{(1)}(z_2)\Bigr)\nonumber\\
		&=\frac{1}{n(pb_p)^2} \tr \biggl\{\biggl[-(n-k)\bSigma_p^2 - \sum_{i<k} \Bigl(\bSigma_p\bx_i\bx_i'\bSigma_p\Bigr)+\sum_{i=1(\neq k)}^n \Expe_{k}\biggl(\frac{\zeta_i}{\sqrt{npb_p}}\bSigma_p\bx_i\bx_i'\bM_{ki}\biggr)\nonumber\\
		&\qquad\qquad\qquad\qquad+\sum_{i=1(\neq k)}^n \Expe_k \Bigl(\zeta_i\vartheta_i\bSigma_p\bx_i\bx_i'\bSigma_p\Bigr)\biggr]\times \Expe_k\bM_k^{(1)}(z_2) \biggr\}\nonumber\\
		&=C_1(z_1,z_2)+C_2(z_1,z_2)+C_3(z_1,z_2)+C_4(z_1,z_2),\label{eq:zZ_decom}
	\end{align}
	where
	\begingroup
	\allowdisplaybreaks
	\begin{align}
		C_1(z_1,z_2)&=-\frac{n-k}{n(pb_p)^2} \tr \Bigl(\bSigma_p^2\cdot \Expe_k\bM_k^{(1)}(z_2)\Bigr),\nonumber\\
		C_2(z_1,z_2)&=-\frac{1}{n(pb_p)^2}\sum_{i<k} \bx_i'\bSigma_p\biggl(\sum_{j=1}^5 \bB_j(z_2)\biggr)\bSigma_p\bx_i=\sum_{j=1}^5 C_{2j},\label{eq:C2}\\
		C_3(z_1,z_2)&=\frac{1}{n(pb_p)^2}\sum_{i<k} \Expe_k\biggl[ \frac{\zeta_i(z_1)}{\sqrt{npb_p}}\bx_i'\bM_{ki}(z_1)\biggl(\sum_{j=1}^5 \bB_j(z_2)\biggr)\bSigma_p\bx_i \biggr] \nonumber\\
		&\qquad +\frac{1}{n(pb_p)^2}\sum_{i>k} \Expe_k\biggl[ \frac{\zeta_i(z_1)}{\sqrt{npb_p}}\bx_i'\bM_{ki}(z_1)\Bigl(\Expe_k\bM_k^{(1)}(z_2)\Bigr)\bSigma_p\bx_i \biggr]=\sum_{j=1}^6 C_{3j},\label{eq:C3}\\
		C_4(z_1,z_2)&=\frac{1}{n(pb_p)^2}\sum_{i<k} \Expe_k\biggl[ \zeta_i(z_1)\vartheta_i(z_1)\bx_i'\bSigma_p\biggl(\sum_{j=1}^5 \bB_j(z_2)\biggr)\bSigma_p\bx_i \biggr]\nonumber\\
		&\qquad +\frac{1}{n(pb_p)^2}\sum_{i>k} \Expe_k\biggl[ \zeta_i(z_1)\vartheta_i(z_1)\bx_i'\bSigma_p\Bigl(\Expe_k\bM_k^{(1)}(z_2)\Bigr)\bSigma_p\bx_i \biggr]=\sum_{j=1}^6 C_{4j}.\label{eq:C4}
	\end{align}
	\endgroup

	Now we estimate all the terms in \eqref{eq:zZ_decom}. We will show that these terms are negligible as $n\to\infty$, expect $C_{25}, C_{33}, C_{45}$ defined in \eqref{eq:C2} $\sim$ \eqref{eq:C4}.
	Before proceeding, we provide two useful lemmas.

	For $C_1(z_1,z_2)$, we have
	\[
	\Expe|C_1(z_1,z_2)|=\frac{n-k}{n(pb_p)^2} \biggl|\tr \Bigl(\bSigma_p^2\cdot \Expe_k\bM_k^{(1)}(z_2)\Bigr)\biggr| = O\biggl(\frac{1}{p^2}\biggr) \cdot  O(np) = O\biggl(\frac{n}{p}\biggr),
	\]
	where the second equality follows from the fact $\Bigl|\tr \Bigl(\bSigma_p^2\cdot \Expe_k\bM_k^{(1)}(z_2)\Bigr)\Bigr|=O(np)$, which can be verified by using the similar argument in the proof of Lemma~\ref{lem:Mk_limit}.
	
	Applying Lemma~\ref{lem:vartheta_zeta} and inequality \eqref{eq:sMBs_bound_2} with $\bB=\bI_p$, we have
	\begin{align*}
		\Expe|C_{21}|& \leqslant \frac{1}{n(pb_p)^2}\sum_{i<k} \Expe\bigl| \bx_i'\bSigma_p\cdot \Expe_k\bM_{ki}(z_2)\cdot\bSigma_p\bx_i \bigr|\\
		&\leqslant \frac{1}{n(pb_p)^2}\sum_{i<k} \Bigl(\Expe\bigl| \bx_i'\bSigma_p\cdot \Expe_k\bM_{ki}(z_2)\cdot\bSigma_p\bx_i \bigr|^2\Bigr)^{\nicefrac{1}{2}}\leqslant \frac{Kn}{p}.
	\end{align*}
	Applying Lemma~\ref{lem:vartheta_zeta} and inequality \eqref{eq:sMBs_bound_2} with $\bB=\bSigma_p$, we have
	\begin{align*}
		\Expe |C_{22}| & \leqslant \frac{1}{n(pb_p)^2}\sum_{i<k} \Expe \biggl|\bx_i'\bSigma_p\cdot \Expe_{k}\biggl( \frac{\zeta_i(z_2)}{z_2npb_p} \bM_{ki}(z_2)\bx_i\bx_i'\bM_{ki}(z_2)\biggr)\cdot\bSigma_p\bx_i\biggr|\\
		& = \frac{K}{n^2(pb_p)^3}\sum_{i<k} \Expe\Bigl| \bx_i'\bM_{ki}(z_2)\bSigma_p\bx_i\Bigr|^2\leqslant \frac{Kn}{p}.
	\end{align*}
	Similarly, we obtain
	\begin{align*}
		\Expe|C_{23}| =\Expe|C_{24}|& \leqslant \frac{1}{n(pb_p)^2}\sum_{i<k} \Expe\Bigl|\bx_i'\bSigma_p\cdot \Expe_k \biggl(\frac{\zeta_i(z_2)}{z_2\sqrt{npb_p}}\bM_{ki}(z_2)\biggr)\bx_i\bx_i'\bSigma_p\cdot\bSigma_p\bx_i\Bigr|\\
		&\leqslant \frac{K}{np^2\sqrt{np}}\sum_{i<k} \Expe\Bigl|\bx_i'\bSigma_p  \bM_{ki}(z_2)\bx_i\cdot\bx_i'\bSigma_p^2\bx_i\Bigr|\\
		&\leqslant \frac{K}{np^2\sqrt{np}}\sum_{i<k} \biggl(\Expe\Bigl|\bx_i'\bSigma_p  \bM_{ki}(z_2)\bx_i\Bigr|^2\biggr)^{\nicefrac{1}{2}}\cdot\biggl(\Expe\Bigl|\bx_i'\bSigma_p^2\bx_i\Bigr|^2\biggr)^{\nicefrac{1}{2}}
		\leqslant K\sqrt{\frac{n}{p}}.
	\end{align*}
	Applying Lemma~\ref{lem:vartheta_zeta} and inequality \eqref{eq:sMBs_bound_1} with $\bB=\Expe_k\bM_{ki}(z_2) \bSigma_p$, we have
	\begin{align*}
		\Expe |C_{31}|&=\frac{1}{n(pb_p)^2}\sum_{i<k} \Expe\biggl|\Expe_k\biggl[ \frac{\zeta_i(z_1)}{\sqrt{npb_p}}\bx_i'\bM_{ki}(z_1)\cdot \Expe_k\bM_{ki}(z_2) \cdot\bSigma_p\bx_i \biggr]\biggr|\\
		&\leqslant \frac{K}{np^2\sqrt{np}}\sum_{i<k} \Expe\biggl| \bx_i'\bM_{ki}(z_1)\cdot \Expe_k\bM_{ki}(z_2) \cdot\bSigma_p\bx_i \biggr|
		\leqslant K\sqrt{\frac{n}{p}}.
	\end{align*}
	
	Define $\tilde{\zeta}_i$ and $\widetilde{\bM}_{ki}$, the analogues of $\zeta_i(z)$ and $\bM_{ki}(z)$ respectively, by $(\bx_1,\ldots,\bx_k, \tilde{\bx}_{k+1},\ldots,\tilde{\bx}_n)$, where $\tilde{\bx}_{k+1},\ldots,\tilde{\bx}_n$ are i.i.d. copies of $\bx_{k+1}, \ldots, \bx_n$ and independent of $\bx_1, \ldots, \bx_n$. Then, 
	\begin{align*}
		\Expe|C_{32}|&=\frac{1}{n(pb_p)^2}\sum_{i<k} \Expe\biggl|\Expe_k\biggl[ \frac{\zeta_i(z_1)}{\sqrt{npb_p}}\bx_i'\bM_{ki}(z_1)\cdot \Expe_{k}\biggl( \frac{\zeta_i(z_2)}{z_2npb_p} \bM_{ki}(z_2)\bx_i\bx_i'\bM_{ki}(z_2)\biggr)\cdot\bSigma_p\bx_i \biggr]\biggr|\\
		&=\frac{1}{n(pb_p)^2}\sum_{i<k} \Expe\biggl|\Expe_k\biggl[ \frac{\zeta_i(z_1)}{\sqrt{npb_p}}\bx_i'\bM_{ki}(z_1)\cdot \Expe_{k}\biggl( \frac{\tilde{\zeta}_i(z_2)}{z_2npb_p} \widetilde{\bM}_{ki}(z_2)\bx_i\bx_i'\widetilde{\bM}_{ki}(z_2)\biggr)\cdot\bSigma_p\bx_i \biggr]\biggr|\\
		&\leqslant \frac{K}{n^2p^3\sqrt{np}}\sum_{i<k} \Expe\biggl|\biggl[ \bx_i'\bM_{ki}(z_1) \widetilde{\bM}_{ki}(z_2)\bx_i\cdot\bx_i'\widetilde{\bM}_{ki}(z_2)\bSigma_p\bx_i \biggr]\biggr|\\
		&\leqslant \frac{K}{n^2p^3\sqrt{np}}\sum_{i<k} \biggl(\Expe\Bigl| \bx_i'\bM_{ki}(z_1) \widetilde{\bM}_{ki}(z_2)\bx_i\Bigr|^2\biggr)^{\nicefrac{1}{2}} \biggl(\Expe\Bigl|\bx_i'\widetilde{\bM}_{ki}(z_2)\bSigma_p\bx_i \Bigr|^2\biggr)^{\nicefrac{1}{2}}
		\overset{\eqref{eq:sMBs_bound_1}}{\leqslant} K\sqrt{\frac{n}{p}}.
	\end{align*}
	Similarly, we have
	\[
	\Expe |C_{3j}| \leqslant K\frac{n}{p},\qquad j=4,5,6.
	\]

	Applying Lemma~\ref{lem:vartheta_zeta} and inequality \eqref{eq:sMBs_bound_2} with $\bB=\bI_{n-1}$, we obtain
	\[
	\Expe |C_{4j}| \leqslant K\frac{n}{p},\qquad j=1,2,3,4,6.
	\]

	
	
	
	
	
	Moreover, by using Lemma~\ref{lem:Mk_limit}, Lemma~\ref{lem:vartheta_zeta} and Lemma~\ref{lem:quad_upper_bound}, we obtain the following limits:
	\begin{align*}
		C_{25}&=-\frac{1}{n(pb_p)^2}\sum_{i<k}\biggl\{ \bx_i'\bSigma_p\biggl[-\Expe_k\biggl(\frac{\zeta_i(z_2)}{z_2}\biggr)\bSigma_p\bx_i\bx_i'\bSigma_p\biggr]\bSigma_p\bx_i\biggr\}\\
		&=-\frac{1}{n(pb_p)^2}m(z_2)\sum_{i<k}\Bigl(\bx_i'\bSigma_p^2\bx_i\Bigr)^2\\
		&=-\frac{k}{n}m(z_2)+o_{L_1}(1),
	\end{align*}
	\begin{align*}
		C_{45}&=\frac{1}{n(pb_p)^2}\sum_{i<k} \Expe_k\biggl\{ \zeta_i(z_1)\vartheta_i(z_1)\bx_i'\bSigma_p \biggl[-\Expe_k\biggl(\frac{\zeta_i(z_2)}{z_2}\biggr)\bSigma_p\bx_i\bx_i'\bSigma_p\biggr]\bSigma_p\bx_i \biggr\}\\
		&=\frac{1}{n(pb_p)^2}\sum_{i<k}\Expe_k \biggl[-m^2(z_1)m(z_2)\Bigl(\bx_i'\bSigma_p^2\bx_i\Bigr)^2\biggr] +o_{L_1}(1)\\
		&=-\frac{k}{n}m^2(z_1)m(z_2)+o_{L_1}(1),
	\end{align*}
	and
	\begin{align*}
		C_{33}&=\frac{1}{n(pb_p)^2}\sum_{i<k} \Expe_k\biggl[ \frac{\zeta_i(z_1)}{\sqrt{npb_p}}\bx_i'\bM_{ki}(z_1)\biggl(\Expe_k \frac{\zeta_i(z_2)}{z_2\sqrt{npb_p}}\bM_{ki}(z_2)\biggr)\bx_i\bx_i'\bSigma_p^2 \bx_i\biggr]\\
		&=\frac{1}{n^2p^2b_p^2}z_1m(z_1)m(z_2)\biggl[ \sum_{i<k}\bx_i'\Expe_k \bM_{ki}(z_1)\Expe_{k}\bM_{ki}(z_2)\bx_i \biggr]+o_{L_4}(1)\\
		&=\frac{k}{n}m(z_1)m(z_2)z_1\bbZ_k+o_{L_1}(1).
	\end{align*}
	From above estimates, we have
	\begin{align*}
		z_1\bbZ_k&=-\frac{k}{n}m(z_2)-\frac{k}{n}m^2(z_1)m(z_2)+\frac{k}{n}m(z_1)m(z_2)z_1\bbZ_k+o_{L_1}(1)\\
		&=\frac{k}{n}z_1m(z_1)m(z_2)+\frac{k}{n}z_1m(z_1)m(z_2)\bbZ_k+o_{L_1}(1),
	\end{align*}
	which is equivalent to
	\[
	\bbZ_k=\frac{\frac{k}{n}m(z_1)m(z_2)}{1-\frac{k}{n}m(z_1)m(z_2)}+o_{L_1}(1).\qedhere
	\]
\end{proof}

\begin{lemma}\label{lem:vartheta_zeta} 
	For $\vartheta_i(z)$ and $\zeta_i(z)$ defined in Lemma \ref{lem:Zk_asym}, we have
	\begin{equation*}
		\Expe \biggl|\vartheta_i(z)-\frac{m(z)}{z}\biggr|^4 \to 0,\qquad \Expe \Bigl|\zeta_i(z)+zm(z)\Bigr|^4 \to 0,\qquad \text{as }\; n\to \infty. 
	\end{equation*}
\end{lemma}
\begin{proof}
	This lemma can be proved by using the similar arguments in Section 5.2.2 of \citet{chen2015clt}.
\end{proof}
\begin{lemma}\label{lem:quad_upper_bound}
	Let $\bB$ be any matrix independent of $\bx_i$.
	\begin{align}
		\Expe \bigl|\bx_i'\bM_{ki}\bB\bx_i\bigr|^2 &\leqslant Kp^2n^2\Expe \|\bB\|^2,\label{eq:sMBs_bound_1}\\
		\Expe \bigl|\bx_i'\bSigma_p\bM_{ki}\bB\bSigma_p\bx_i\bigr|^2 &\leqslant Kp^2n^2\Expe \|\bB\|^2.	\label{eq:sMBs_bound_2}
	\end{align}
\end{lemma}
\begin{proof}
	Note that $\bM_{ki}$ and $\bx_i$ are independent. By using Lemma~\ref{lem:quad-trace}, we have
	\begin{equation}
		\Expe \bigl|\bx_i'\bM_{ki}\bB\bx_i-\tr\bM_{ki}\bB\bigr|^2 \leqslant K \Bigl( \nu_4 \Expe\tr\bigl(\bM_{ki}\bB\overline{\bB}\overline{\bM}_{ki}\bigr) \Bigr)\leqslant K np^2\|\bB\|^2,\label{eq:sMBs_trMB_bound}
	\end{equation}
	where we use the fact that
	\begin{align}
		\bigl|\tr\bigl(\bM_{ki}\bB\overline{\bB}\overline{\bM}_{ki}\bigr)\bigr|&= \bigl| \tr\bigl(\bSigma_p\bX_{ki}\bD_{ki}\bX_{ki}'\bSigma_p \bB \overline{\bB}\bSigma_p\bX_{ki}\overline{\bD}_{ki}\bX_{ki}'\bSigma_p \bigr) \bigr| \nonumber\\
		& = \bigl| \tr\bigl(\bD_{ki}^{\nicefrac{1}{2}}\bX_{ki}'\bSigma_p \bB \overline{\bB}\bSigma_p\bX_{ki}\overline{\bD}_{ki}\bX_{ki}'\bSigma_p^2\bX_{ki}\bD_{ki}^{\nicefrac{1}{2}} \bigr) \bigr| \nonumber\\
		&\leqslant n \cdot\|\bD_{ki}^{\nicefrac{1}{2}}\bX_{ki}'\bSigma_p^{\nicefrac{1}{2}}\|\cdot \|\bSigma_p^{\nicefrac{1}{2}}\|\cdot\|\bB \overline{\bB}\| \cdot  \|\bSigma_p^{\nicefrac{1}{2}}\|\nonumber\\
		&\qquad\qquad\times \|\bSigma_p^{\nicefrac{1}{2}}\bX_{ki}\overline{\bD}_{ki}\bX_{ki}'\bSigma_p^{\nicefrac{1}{2}}\|\cdot \|\bSigma_p\|\cdot\|\bSigma_p^{\nicefrac{1}{2}}\bX_{ki}\bD_{ki}^{\nicefrac{1}{2}} \|\nonumber\\
		& = n\cdot\|\bSigma_p\|^2\cdot\|\bB\|^2\cdot \|\bSigma_p^{\nicefrac{1}{2}}\bX_{ki}\bD_{ki}\bX_{ki}'\bSigma_p^{\nicefrac{1}{2}}\|^2\nonumber\\
		& = n\cdot\|\bSigma_p\|^2\cdot\|\bB\|^2\cdot \|\bD_{ki}\bX_{ki}'\bSigma_p\bX_{ki}\|^2\nonumber\\
		& =  n\cdot\|\bSigma_p\|^2\cdot\|\bB\|^2\cdot \| \sqrt{npb_p}(\bI_{n-1}+z\bD_{ki}) + pa_p\bI_{(i)}\bD_{ki} \|^2 \nonumber\\
		& \leqslant Knp^2\|\bB\|^2.\label{eq:trace_2M}
	\end{align}
	By \eqref{eq:sMBs_trMB_bound} and the $c_r$-inequality, we have
	\[
	\Expe \bigl|\bx_i'\bM_{ki}\bB\bx_i\bigr|^2
	\leqslant K\Bigl( \Expe \bigl|\bx_i'\bM_{ki}\bB\bx_i-\tr\bM_{ki}\bB\bigr|^2  + \Expe \bigl|\tr \bM_{ki}\bB\bigr|^2\Bigr)\leqslant Kp^2n^2\Expe \|\bB\|^2,
	\]
	which completes the proof of \eqref{eq:sMBs_bound_1}. By using the same argument, we get \eqref{eq:sMBs_bound_2}.
\end{proof}

Lemma \ref{lem:vartheta_zeta} and \ref{lem:quad_upper_bound} are used in the proof of Lemma \ref{lem:Zk_asym}.

Recall that $\mathbb{C}_1 = \{z: z=u+iv, u\in [-u_1, u_1],~|v|\geqslant v_1\}$, where $u_1>2$ and $0<v_1\leqslant 1$.

\begin{lemma}\label{lem:betak_upper_bounds} For $z\in\mathbb{C}_1$, we have
	\begin{align}
		|\beta_k(z)|\leqslant 1/v_1 ,\qquad |\beta_k^{\mathsf{tr}}(z)|\leqslant 1/v_1,\nonumber\\
			\biggl|1+\frac{1}{npb_p}\tr \bM_k^{(s)}(z)\biggr|\leqslant 1+\frac{1}{v_1^s},\qquad s=1, 2,\nonumber\\
			\Bigl|\beta_k\Bigl(1+\bq_k'\bD_k^{2}(z)\bq_k\Bigr)\Bigr|\leqslant \frac{1}{v_1}. \label{eq:beta_qDq_upper_bdd}
	\end{align}
\end{lemma}
\begin{proof}
	The proof exactly follows \citet{chen2015clt}, so is omitted.
\end{proof}
\begin{lemma}\label{lem:gamma_moment_upper_bound}
	Under the assumption $p\wedge n\to\infty,~p/n\to\infty$ and truncation, for $z\in\mathbb{C}_1$, 
	\begin{align}
		\Expe &|\gamma_{ks}|^2\leqslant \frac{K}{n}, \\
		\Expe&|\gamma_{ks}|^4 \leqslant K\biggl(\frac{1}{n^2}+\frac{n}{p^2}\biggr),\\
		\Expe & |\eta_k|^2 \leqslant \frac{K}{n},\\
		\Expe&|\eta_k|^4 \leqslant K\frac{\delta_n^4}{n} + K\biggl(\frac{1}{n^2}+\frac{n}{p^2}\biggr).
	\end{align}
\end{lemma}

\begin{proof}
	By Lemma~\ref{lem:quad-trace} and taking $\bB=\bI_p$ in the inequality~\eqref{eq:trace_2M}, we have
	\[
		\Expe |\gamma_{k2}|^2 \leqslant \frac{K}{n^2p^2} \tr \bigl(\bM_k^{(s)}\overline{\bM}_k^{(s)}\bigr)\leqslant \frac{K}{n}.
	\]
	Similarly, we can prove that $\Expe |\eta_k|^2 \leqslant K/n$.

	Now, we prove the bounds for the $4$-th moments of $\gamma_{ks}$ and $\eta_{ks}$.
	Let $\bH$ be $\bM_k^{(s)}$ with all diagonal elements replaced by zeros, then we have
	\begin{equation}\label{eq:sHs_1}
		\Expe |\bx_k'\bH\bx_k|^4\leqslant K(\Expe X_{11}^4)^2 \Expe \bigl(\tr \bH\bH^*\bigr)^2 \leqslant K\Expe\bigl(\tr \bM_k^{(s)}\overline{\bM}_k^{(s)}\bigr)^2 \leqslant Kn^2p^4.
	\end{equation}
	The first inequality follows from Lemma~\ref{lem:PanZhou-trace}, and the last inequality follows from \eqref{eq:trace_2M}.
	Denote $\Expe_j(\cdot)$ be the conditional expectation with respect to $(X_{1k}, X_{2k}, \ldots, X_{jk})$, where $j=1,2,\ldots,p$. Since $\Expe_{j-1}(X_{jk}^2-1)a_{jj}^{(s)}=0$, then
	$(X_{jk}^2-1)a_{jj}^{(s)}$ can be expressed as a martingale difference
	\begin{equation}\label{eq:Xjk_MDS}
		(X_{jk}^2-1)a_{jj}^{(s)} = (\Expe_{j}-\Expe_{j-1})\Bigl[(X_{jk}^2-1)a_{jj}^{(s)}\Bigr].
	\end{equation}
	Applying the Burkholder’s inequality (Lemma~\ref{lem:Burkholder_ineq}) to \eqref{eq:Xjk_MDS} yields that
	\begin{align}
		\Expe \biggl|\sum_{j=1}^p (X_{jk}^2-1)a_{jj}^{(s)}\biggr|^4 & \leqslant K \Expe \biggl( \sum_{j=1}^p\Expe_{j-1}\Bigl|(X_{jk}^2-1)a_{jj}^{(s)}\Bigr|^2 \biggr)^2 + K\Expe\biggl(\sum_{j=1}^p\Bigl|(X_{jk}^2-1)a_{jj}^{(s)}\Bigr|^4 \biggr)^2\nonumber\\
		&\leqslant K\biggl(\sum_{j=1}^p \Expe |X_{11}|^4 \bigl|a_{jj}^{(s)}\bigr|^2\biggr) + K \sum_{j=1}^p \Expe |X_{11}|^8 \Expe\bigl|a_{jj}^{(s)}\bigr|^4\nonumber\\
		&\leqslant Kn^5p^2+Kn^3p^3,\label{eq:sHs_2}
	\end{align}
	where we use the fact that, with $\be_j$ be the $j$-th $p$-dimensional standard basis vector and $\by$ be an $(n-1)$-dimensional random vector with $\Expe y_i=0$ and $\Expe y_i^2=1$,
	\begin{align}
		\Expe\bigl|a_{jj}^{(s)}\bigr|^4 & = \Expe \Bigl| \be_j'\bSigma_p\bX_k\bD_k^{s}\bX_k'\bSigma_p\be_j \Bigr|^4\nonumber\\
		&\leqslant v_1^{-4s}\Expe \bigl\|\be_j'\bSigma_p\bX_k\bigr\|^8
		=v_1^{-4s}\widetilde{\sigma}_{jj}^4\Expe \bigl\|\by\bigr\|^8\leqslant Kn^4 + K n^2p,\label{eq:ajj_upper_bound}
	\end{align}
	where $\widetilde{\sigma}_{jj} = \sum_{\ell} \sigma_{j\ell}^2$ is the $j$-th diagonal elements of $\bSigma_p^2$. By Rayleigh-Ritz Theorem, we know that $\widetilde{\sigma}_{jj}\leqslant \lambda_{\max}(\bSigma_p^2)\leqslant K$. Combining \eqref{eq:sHs_1} and \eqref{eq:sHs_2} yields that
	\begin{align*}
		\Expe |\gamma_{ks}|^4 & \leqslant \frac{1}{(npb_p)^4}\Expe \biggl| \sum_{j=1}^p(X_{jk}^2-1)a_{jj}^{(s)} + \bx_k'\bH\bx_k \biggr|^4\\
		& \leqslant \frac{K}{n^4p^4} \Expe \biggl|\sum_{j=1}^p (X_{jk}^2-1)a_{jj}^{(s)}\biggr|^4 + \frac{K}{n^4p^4}\Expe |\bx_k'\bH\bx_k|^4\\
		&\leqslant K\biggl(\frac{1}{n^2}+\frac{n}{p^2}\biggr).
	\end{align*}
	Moreover, by Lemma~\ref{lem:quad-trace}, we have
	\[
		\Expe |\eta_k|^4 \leqslant \frac{K}{n^2p^2} \Expe \bigl|\bx_k'\bSigma_p\bx_k-pa_p\bigr|^4 + K \Expe \bigl|\gamma_{k1}\bigr|^4\leqslant \frac{K\delta_n^4}{n}+K\biggl(\frac{1}{n^2}+\frac{n}{p^2}\biggr).\qedhere
	\]
\end{proof}

Lemmas \eqref{lem:var_mn_estimate} , \eqref{lem:G11_upper_bbd} and \eqref{lem:G11_tilde_upper_bbd} below are used in Section~\ref{sec:conv_Mn2} to derive the convergence of the non-random part $M_n^{2}(z)$. We will prove them following the strategy in \citet{bao2015asymptotic}.
\begin{lemma}\label{lem:var_mn_estimate} 
	Under the assumption $p\wedge n\to\infty,~p/n\to\infty$, for $z\in\mathbb{C}_1$, we have
	\begin{equation}\label{eq:var_mn_estimate}
		\Var(m_n) = O\biggl(\frac{1}{n^2}\biggr).
	\end{equation}
\end{lemma}
\begin{proof}
	By the identity $m_n-\Expe m_n = -\sum_{k=1}^n \Bigl(\Expe_{k-1}\,m_n-\Expe_k\, m_n\Bigr)$, we have
	\[
		\Var(m_n) 
		= \sum_{k=1}^n\Expe\,\Bigl|\Expe_{k-1}\,m_n-\Expe_k\, m_n\Bigr|^2 + 2\sum_{1\leqslant s<t\leqslant 1} \Expe\, \Bigl(\Expe_{s-1}\,m_n-\Expe_s\, m_n\Bigr)\Bigl(\Expe_{t-1}\,m_n-\Expe_t\, m_n\Bigr).
	\]
	Since each term in the second sum on the RHS of the above identity is zero, we write
	\begin{align*}
		\Var(m_n) & = \sum_{k=1}^n\Expe\,\Bigl|\Expe_{k-1}\,m_n-\Expe_k\, m_n\Bigr|^2 \\
		& = \sum_{k=1}^n\Expe\,\Bigl|\Expe_{k-1}\,\Bigl(m_n-\Expe_{(k)}\, m_n\Bigr)\Bigr|^2 \\
		& \leqslant \sum_{k=1}^n\Expe\,\Bigl|m_n-\Expe_{(k)}\, m_n\Bigr|^2 ,
	\end{align*}
	where $\Expe_{(k)}(\cdot)$ denotes the expectation w.r.t. the $\sigma$-field generated by $\bx_k$. To prove \eqref{eq:var_mn_estimate}, it suffices to show 
	\begin{equation}\label{eq:mn_var_martigale_element}
		\Expe\,\Bigl|m_n-\Expe_{(k)}\, m_n\Bigr|^2 = O\biggl(\frac{1}{n^3}\biggr),\qquad 1\leqslant k\leqslant n.
	\end{equation}
	Now we deal with the case $k=1$, and the remaining cases are analogous and omitted. 

	Denote $\widetilde{\bY}=(\widetilde{Y}_{ij})_{p\times n} := \bSigma_p^{\nicefrac{1}{2}}\bY$ where $\bY=(npb_p)^{-1/4}\bX$, and let $\widetilde{\by}_k$ be the $k$-th column of $\widetilde{\bY}$. Let $\widetilde{\bY}_k$ be the $p\times(n-1)$ matrix extracted from $\widetilde{\bY}$ by removing $\widetilde{\by}_k$, then the matrix model \eqref{eq:A_def} can be written as
	\[
		\bA_n = \begin{pmatrix}
			\widetilde{\by}_1'\widetilde{\by}_1 -\sqrt{\frac{p}{n}}\frac{a_p}{\sqrt{b_p}} & (\widetilde{\bY}_1'\widetilde{\by}_1)'\\
			\widetilde{\bY}_1'\widetilde{\by}_1 & \widetilde{\bY}_1'\widetilde{\bY}_1 - \sqrt{\frac{p}{n}}\frac{a_p}{\sqrt{b_p}} \bI_{n-1}
		\end{pmatrix}.
	\]
	With notations 
	$
		\bA_k = \widetilde{\bY}_k'\widetilde{\bY}_k - \sqrt{\frac{p}{n}}\frac{a_p}{\sqrt{b_p}} \bI_{n-1}	
	$
	and $\bD_k= (\bA_k-z\bI_n)^{-1}$, we have
	\begin{align*}
		\;&\tr \bD- \tr \bD_1 \\
		=\;&  \frac{1+(\widetilde{\bY}_1'\widetilde{\by}_1)'\Bigl(\widetilde{\bY}_1'\widetilde{\bY}_1 - \sqrt{\frac{p}{n}}\frac{a_p}{\sqrt{b_p}} \bI_{n-1} - z\bI_{n-1}\Bigr)^{-2}(\widetilde{\bY}_1'\widetilde{\by}_1)}{\Bigl(\widetilde{\by}_1'\widetilde{\by}_1 -\sqrt{\frac{p}{n}}\frac{a_p}{\sqrt{b_p}}-z\Bigr) - (\widetilde{\bY}_1'\widetilde{\by}_1)'\Bigl(\widetilde{\bY}_1'\widetilde{\bY}_1 - \sqrt{\frac{p}{n}}\frac{a_p}{\sqrt{b_p}} \bI_{n-1} - z\bI_{n-1}\Bigr)^{-1}(\widetilde{\bY}_1'\widetilde{\by}_1)}\\[0.5em]
		=\;& \frac{1+\widetilde{\by}_1'\Bigl[\widetilde{\bY}_1\widetilde{\bY}_1' \Bigl(\widetilde{\bY}_1\widetilde{\bY}_1'- \sqrt{\frac{p}{n}}\frac{a_p}{\sqrt{b_p}} \bI_{n-1} - z\bI_{n-1}\Bigr)^{-2}\Bigr]\widetilde{\by}_1}{\Bigl(\widetilde{\by}_1'\widetilde{\by}_1 -\sqrt{\frac{p}{n}}\frac{a_p}{\sqrt{b_p}}-z\Bigr) - \widetilde{\by}_1'\Bigl[\widetilde{\bY}_1\widetilde{\bY}_1' \Bigl(\widetilde{\bY}_1\widetilde{\bY}_1'- \sqrt{\frac{p}{n}}\frac{a_p}{\sqrt{b_p}} \bI_{n-1} - z\bI_{n-1}\Bigr)^{-1}\Bigr]\widetilde{\by}_1}\\[0.5em]
		=:\;& \frac{1+U}{V},
	\end{align*}
	where the second ``='' comes from the identity
	\[
		\bB(\bA\bB-\alpha \bI)^{-n}\bA = \bB\bA(\bB\bA-\alpha\bI)^{-n}.
	\]
	Moreover, with notations $U$ and $V$, we can write $D_{11}=1/V$ and
	\begin{align*}
		\Expe\,\Bigl|m_n-\Expe_{(1)}\, m_n\Bigr|^2 & = \frac{1}{n^2} \Expe\,\biggl| (\tr \bD -\tr\bD_1) - \Expe_{(1)}(\tr \bD -\tr\bD_1)\biggr|^2\qquad \bigl(\because \Expe_{(1)}\tr\bD_1 = \tr\bD_1\bigr)\\
		&= \frac{1}{n^2}\Expe\,\biggl|\frac{1+U}{V} - \Expe_{(1)} \biggl(\frac{1+U}{V}\biggr)\biggr|^2\\
		& \leqslant \frac{2}{n^2} \biggl\{\Expe\,\biggl|\frac{1}{V} - \Expe_{(1)} \biggl(\frac{1}{V}\biggr)\biggr|^2 + \Expe\,\biggl|\frac{U}{V} - \Expe_{(1)} \biggl(\frac{U}{V}\biggr)\biggr|^2\biggr\}.
	\end{align*}
	By the same arguments as those on Page 196 of \citet{bao2015asymptotic}, it is sufficient to prove that
	\begin{equation}\label{eq:U_V_esti}
		\Expe_{(1)}\,|U-\Expe_{(1)}\, U|^2 = O\biggl(\frac{1}{n}\biggr),\qquad 	\Expe_{(1)}\,|V-\Expe_{(1)}\, V|^2 = O\biggl(\frac{1}{n}\biggr).
	\end{equation}
	For simplicity of presentation, we define
	\[
		\bH^{[\ell]}  = \Bigl(H_{jk}^{[\ell]}\Bigr)_{p\times p}:= \widetilde{\bY}_1\widetilde{\bY}_1' \Bigl(\widetilde{\bY}_1\widetilde{\bY}_1'- \sqrt{\frac{p}{n}}\frac{a_p}{\sqrt{b_p}} \bI_{n-1} - z\bI_{n-1}\Bigr)^{-\ell},\qquad \ell=1, 2.
	\]
	Then, we write
	\begin{align}
		U-\Expe_{(1)}\, U & = \sum_{i\neq j} H_{ij}^{[2]} \widetilde{Y}_{i1}\widetilde{Y}_{j1} + \sum_{i=1}^p H_{ii}^{[2]} \Bigl( \widetilde{Y}_{i1}^2 - \Expe\, \widetilde{Y}_{i1}^2\Bigr),\label{eq:U_exp}\\
		V-\Expe_{(1)}\, V & = \widetilde{\by}_1'\widetilde{\by}_1 - \sqrt{\frac{p}{n}}\frac{a_p}{\sqrt{b_p}} - \sum_{i\neq j} H_{ij}^{[1]} \widetilde{Y}_{i1}\widetilde{Y}_{j1} -\sum_{i=1}^p H_{ii}^{[1]} \Bigl( \widetilde{Y}_{i1}^2 - \Expe\, \widetilde{Y}_{i1}^2\Bigr).\label{eq:V_exp}
	\end{align}
	Now we proceed to prove \eqref{eq:U_V_esti}. From \eqref{eq:V_exp}, we have
	\begin{align}
		&\;\Expe_{(1)}|V-\Expe_{(1)}\, V|^2\nonumber\\
		\leqslant&\; K\biggl\{\Expe_{(1)} \biggl|\widetilde{\by}_1'\widetilde{\by}_1 - \sqrt{\frac{p}{n}}\frac{a_p}{\sqrt{b_p}}\biggr|^2 
			+ \Expe_{(1)}\,\biggl| \sum_{i\neq j} H_{ij}^{[1]} \widetilde{Y}_{i1}\widetilde{Y}_{j1}\biggr|^2 
			+ \Expe_{(1)}\,\biggl|\sum_{i=1}^p H_{ii}^{[1]} \Bigl( \widetilde{Y}_{i1}^2 - \Expe\, \widetilde{Y}_{i1}^2\Bigr)\biggr|^2\biggr\}\nonumber\\
		=&\; K\biggl\{\Expe\,\biggl|\widetilde{\by}_1'\widetilde{\by}_1 - \sqrt{\frac{p}{n}}\frac{a_p}{\sqrt{b_p}}\biggr|^2 
			+ \sum_{i\neq j} \Bigl|H_{ij}^{[1]}\Bigr|^2 \Expe\,\Bigl(\widetilde{Y}_{i1}^2\widetilde{Y}_{j1}^2\Bigr)
			+ \sum_{i=1}^p \Bigl|H_{ii}^{[1]}\Bigr|^2 \Expe\,\Bigl( \widetilde{Y}_{i1}^2 - \Expe\, \widetilde{Y}_{i1}^2\Bigr)^2\biggr\}.\label{eq:U_esti_decom}
	\end{align}
	After some straightforward calculations, we obtain some estimates:
	\begin{equation}\label{eq:ytilde_moment_esti}
		\Expe\,\widetilde{Y}_{i1}^2 = O\biggl(\frac{1}{\sqrt{np}}\biggr),
		\qquad 
		\Expe\,\widetilde{Y}_{i1}^4 =O\biggl(\frac{1}{np}\biggr),
		\qquad 
		\Expe\,\biggl(\widetilde{\by}_1'\widetilde{\by}_1 - \sqrt{\frac{p}{n}}\frac{a_p}{\sqrt{b_p}}\biggr)^2 = O\biggl(\frac{1}{n}\biggr).
	\end{equation}
	Combining \eqref{eq:ytilde_moment_esti} and \eqref{eq:U_esti_decom}, we obtain 
	\begin{equation}
		\Expe_{(1)}\,\bigl|V-\Expe_{(1)}\, V\bigr|^2 \leqslant \frac{K}{n}+\frac{K}{np} \tr\Bigl|\bH^{[1]}\Bigr|^2.
	\end{equation}
	Similarly, we can show that
	\begin{equation}
		\Expe_{(1)}\,\bigl|U-\Expe_{(1)}\, U\bigr|^2 \leqslant \frac{K}{np} \tr\Bigl|\bH^{[2]}\Bigr|^2.
	\end{equation}
	To get \eqref{eq:U_V_esti}, it suffices to show that
	\[
		\tr\bigl|\bH^{[\ell]}\bigr|^2 = O(p),\qquad \ell=1,2.	
	\]
	Let $\{\mu_i^{(k)},\, i=1,2,\ldots,n-1\}$ be eigenvalues of $\bA_k$, then the eigenvalues of $\bH^{[\ell]}\, (\ell=1, 2)$ are 
	\[
		\frac{\Bigl(\mu_{i}^{(1)}+a_p\sqrt{p/(nb_p)}\Bigr)^2}{\bigl|\mu_i^{(1)}-z\bigr|^{2\ell}}, \qquad i = 1,2,\ldots,n-1,	
	\]
	and a zero eigenvalue with algebraic multiplicity $(p-n+1)$. Using the fact $\mu_i^{(1)} \geqslant -a_p\sqrt{p/(nb_p)}$, we conclude that
	\[
		\tr\bigl|\bH^{[\ell]}\bigr|^2 = \sum_{i=1}^{n-1} \frac{\Bigl(\mu_{i}^{(1)}+a_p\sqrt{p/(nb_p)}\Bigr)^2}{\bigl|\mu_i^{(1)}-z\bigr|^{2\ell}}= O(p),\qquad \ell=1,2.	\qedhere
	\]
\end{proof}

\begin{lemma}\label{lem:G11_upper_bbd}
	Under the assumption $p\wedge n\to\infty,~p/n\to\infty$, for $z\in\mathbb{C}_1$ and $1\leqslant \ell \leqslant n$, 
	\begin{equation}
		\Expe \biggl|D_{\ell\ell}+\frac{1}{z+\Expe\, m_n}\biggr|^2 = O\biggl(\frac{1}{n}\biggr) + O\biggl(\frac{n}{p}\biggr).
	\end{equation}
\end{lemma}
\begin{proof}
	We only provide the estimation of $D_{11}$, since others are analogous. Note that
	\[
		D_{11}=V^{-1}=\Bigl(\widetilde{\by}_1'\widetilde{\by}_1-\sqrt{\frac{p}{n}}\frac{a_p}{\sqrt{b_p}}-z-\widetilde{\by}_1'\bH^{[1]}\widetilde{\by}_1\Bigr)^{-1}.	
	\]

	Let $\bv_i^{(1)}=\bigl(v_{i1}^{(1)},\ldots,v_{ip}^{(1)}\bigr),\, (i=1,2,\ldots,n-1)$ be the unit eigenvector of $\bA_1$ corresponding to the eigenvalue $\mu_i^{(1)}$, and let
	\[
		w_i^{(1)} = \frac{\sqrt{np}a_p}{\sqrt{b_p}} \bigl|\widetilde{\by}_1'\bv_i^{(1)}\bigr|^2.
	\]
	Applying spectral decomposition to $\bH^{[1]}$ yields 
	\begin{align}
		D_{11} & = \Biggl[\widetilde{\by}_1'\widetilde{\by}_1-\sqrt{\frac{p}{n}}\frac{a_p}{\sqrt{b_p}}-z-\sum_{i=1}^{n-1}\biggl(\frac{\mu_i^{(1)}+\sqrt{\tfrac{p}{n}}\tfrac{a_p}{\sqrt{b_p}}}{\mu_i^{(1)}-z}\biggr)\bigl|\widetilde{\by}_1'\bv_i^{(1)}\bigr|^2\Biggr]^{-1}	\nonumber\\[0.5em]
		& = \Biggl[\widetilde{\by}_1'\widetilde{\by}_1-\sqrt{\frac{p}{n}}\frac{a_p}{\sqrt{b_p}}-z-\frac{1}{\sqrt{np}}\frac{\sqrt{b_p}}{a_p}\sum_{i=1}^{n-1}\frac{\Bigl(\mu_i^{(1)}+\sqrt{\tfrac{p}{n}}\tfrac{a_p}{\sqrt{b_p}}\Bigr)w_i^{(1)}}{\mu_i^{(1)}-z}\Biggr]^{-1}\nonumber\\
		&=: \Bigl(-z-m_n(z)+h_1\Bigr)^{-1},\label{eq:G11_inv}
	\end{align}
	where 
	\[
		h_1 = \Biggl[m_n-\frac{1}{n}\sum_{i=1}^{n-1}\Biggl(\frac{\sqrt{\tfrac{n}{p}}\tfrac{\sqrt{b_p}}{a_p}\mu_i^{(1)}+1}{\mu_i^{(1)}-z}\Biggr)\Biggr]
		+ \Biggl[\widetilde{\by}_1'\widetilde{\by}_1-\sqrt{\frac{p}{n}}\frac{a_p}{\sqrt{b_p}}-\frac{1}{n}\sum_{i=1}^{n-1}\Biggl(\frac{\sqrt{\tfrac{n}{p}}\tfrac{\sqrt{b_p}}{a_p}\mu_i^{(1)}+1}{\mu_i^{(1)}-z}\Biggr)\bigl(w_i^{(1)}-1\bigr)\Biggr].	
	\]
	By \eqref{eq:G11_inv}, we obtain 
	\[
		\biggl|D_{11} + \frac{1}{z+\Expe\, m_n}\biggr|  = \biggl|\frac{\Expe\, m_n -m_n + h_1}{\bigl(-z-m_n+h_1\bigr)\bigl(z+\Expe\, m_n\bigr)} \biggr|\leqslant K\Bigl| (\Expe\, m_n -m_n) + h_1\Bigr|,
	\]
	which implies that
	\begin{align}
		\Expe\, \biggl|&D_{11} + \frac{1}{z+\Expe\, m_n}\biggr|^2  \nonumber\\
		& \leqslant K \Biggl\{ \Expe\, \bigl|\Expe\, m_n-m_n\bigr|^2 + \Expe\, \Biggl|m_n - \Biggl(\frac{1}{n}\sum_{i=1}^{n-1}\frac{1}{\mu_i^{(1)}-z}\Biggr)-\sqrt{\frac{n}{p}}\frac{\sqrt{b_p}}{a_p}\Biggl(\frac{1}{n}\sum_{i=1}^{n-1}\frac{\mu_i^{(1)}}{\mu_i^{(1)}-z}\Biggr)\Biggr|^2\nonumber \\[0.5em]
		&\qquad\qquad + \Expe\, \Biggl|\frac{1}{n}\sum_{i=1}^{n-1}\Biggl(\frac{\sqrt{\tfrac{n}{p}}\tfrac{\sqrt{b_p}}{a_p}\mu_i^{(1)}+1}{\mu_i^{(1)}-z}\Biggr)\bigl(w_i^{(1)}-1\bigr)\Biggr|^2 + \Expe\, \biggl|\widetilde{\by}_1'\widetilde{\by}_1-\sqrt{\frac{p}{n}}\frac{a_p}{\sqrt{b_p}}\biggr|^2\Biggr\}\nonumber\\
		& =: K (\mathrm{I+II+III+IV})\nonumber\\
		& = O\biggl(\frac{1}{n^2}\biggr)  + \biggl[O\biggl(\frac{1}{n^2}\biggr) + O\biggl(\frac{n}{p}\biggr)\biggr] + O\biggl(\frac{1}{n}\biggr) + O\biggl(\frac{1}{n}\biggr)\label{eq:G11_upper_bbd_0}\\
		& = O\biggl(\frac{1}{n}\biggr) + O\biggl(\frac{n}{p}\biggr).\label{eq:G11_upper_bbd}
	\end{align}
	Below we explain \eqref{eq:G11_upper_bbd_0} in more detail:
	\begin{enumerate}
		\item[(I)] Follows from Lemma~\ref{lem:var_mn_estimate}.
		\item[(II)] Use the fact 
		\[
			\sqrt{\frac{n}{p}}\frac{\sqrt{b_p}}{a_p}\Biggl|\frac{1}{n}\sum_{i=1}^{n-1}\frac{\mu_i^{(1)}}{\mu_i^{(1)}-z}\Biggr| = O\biggl(\sqrt{\frac{n}{p}}\biggr)
		\]
		and
		\[
			\Biggl|m_n - \frac{1}{n}\sum_{i=1}^{n-1}\frac{1}{\mu_i^{(1)}-z}\Biggr| = \biggl|\frac{1}{n}\tr \bD -\frac{1}{n}\tr \bD_k\biggr| \overset{\eqref{eq:D_Dk_ESD_diff}}{=} O\biggl(\frac{1}{n}\biggr).
		\] 
		\item[(III)] Use \eqref{eq:ytilde_moment_esti}.
		\item[(IV)] Analogous to the estimation of $\Expe\,\bigl|V-\Expe_{(1)} V\bigr|^2$.\qedhere
	\end{enumerate}
\end{proof}
The following lemma is used to prove \eqref{eq:stein_esti_E}.
Define 
\[
	\widetilde{\bD}=(\widetilde{D}_{ij})_{p\times p} = \biggl( \bSigma_p^{\nicefrac{1}{2}}\bY\bY'\bSigma_p^{\nicefrac{1}{2}}-\sqrt{\frac{p}{n}}\frac{a_p}{\sqrt{b_p}\bI_p}-z\bI_p \biggr)^{-1}.
\]

\begin{lemma}\label{lem:G11_tilde_upper_bbd}
	Under the assumption $p\wedge n\to\infty,~p/n\to\infty$, for $z\in\mathbb{C}_1$ and $1\leqslant \ell \leqslant p$,
	\begin{equation}\label{eq:G11_tilde_upper_bbd}
		\Expe\,\biggl| \widetilde{D}_{\ell\ell}+\frac{1}{a_p\sqrt{p/(nb_p)}+z+\Expe\, m_n} \biggr|^2  = O\biggl(\biggl(\frac{n}{p}\biggr)^3\biggr) + O\biggl(\frac{n}{p^2}\biggr).
	\end{equation}
\end{lemma}

\begin{proof}
	We only provide the estimation of $\widetilde{D}_{11}$, since the others are analogous. 
	
	Let $\widetilde{\br}_k'$ be $k$-th row of $\widetilde{\bY}$ and let $\bB_k$ be the $(p-1)\times n$ matrix extracted from $\widetilde{\bY}$ by deleting $\widetilde{\br}_k'$.

	With notations defined above, we can write
	\[
		\widetilde{\bA}=\begin{pmatrix}
			\widetilde{\br}_1'\widetilde{\br}_1-\sqrt{\tfrac{p}{n}}\tfrac{a_p}{\sqrt{b_p}} & \widetilde{\br}_1'\bB_1'\\
			\bB_1\widetilde{\br}_1 & \bB_1\bB_1'-\sqrt{\tfrac{p}{n}}\tfrac{a_p}{\sqrt{b_p}}\bI_{p-1}
		\end{pmatrix}.
	\]
	Denote \[
		\widetilde{\bA}_k = \bB_k'\bB_k - \sqrt{\frac{p}{n}}\frac{a_p}{\sqrt{b_p}}\bI_n,
	\]
	and \[
		W= \widetilde{\br}_1'\bB_1'\bB_1\biggl(\bB_1'\bB_1-\sqrt{\tfrac{p}{n}}\tfrac{a_p}{\sqrt{b_p}}\bI_{n}\biggr)^{-1}\widetilde{\br}_1.
	\]
	Let $\{\widetilde{\mu}_i^{(k)},\, i=1,2,\ldots,n-1\}$ be the eigenvalues of $\widetilde{\bA}_k$, and let $\widetilde{\bv}_i^{(1)}=\bigl(\widetilde{v}_{i1}^{(1)},\ldots,\widetilde{v}_{ip}^{(1)}\bigr),\, (i=1,2,\ldots,n)$ be the unit eigenvector of $\widetilde{\bA}_1$ corresponding to the eigenvalue $\widetilde{\mu}_i^{(1)}$, and set
	\[
		\widetilde{w}_i^{(1)} = \frac{\sqrt{np}a_p}{\sqrt{b_p}} \bigl|\widetilde{\br}_1'\widetilde{\bv}_i^{(1)}\bigr|^2,
	\]
	then we have
	\[
		W=\frac{1}{n}\sum_{i=1}^{n}\Biggl(\frac{\sqrt{\tfrac{n}{p}}\tfrac{\sqrt{b_p}}{a_p}\widetilde{\mu}_i^{(1)}+1}{\widetilde{\mu}_i^{(1)}-z}\Biggr)\widetilde{w}_i^{(1)},
	\]
	and
	\[
		\widetilde{D}_{11} = \biggl(\widetilde{\br}_1'\widetilde{\br}_1-\sqrt{\tfrac{p}{n}}\tfrac{a_p}{\sqrt{b_p}}-z-W\biggr)^{-1} 
		=: \biggl(-\sqrt{\tfrac{p}{n}}\tfrac{a_p}{\sqrt{b_p}}-z-m_n+\widetilde{h}_1\biggr)^{-1},
	\]
	where
	\begin{align}
		\widetilde{h}_1 & = \widetilde{\br}_1'\widetilde{\br}_1+m_n-W \nonumber\\
		& = \widetilde{\br}_1'\widetilde{\br}_1+m_n- \frac{1}{n}\sum_{i=1}^{n}\Biggl(\frac{\sqrt{\tfrac{n}{p}}\tfrac{\sqrt{b_p}}{a_p}\widetilde{\mu}_i^{(1)}+1}{\widetilde{\mu}_i^{(1)}-z}\Biggr) - \frac{1}{n}\sum_{i=1}^{n}\Biggl(\frac{\sqrt{\tfrac{n}{p}}\tfrac{\sqrt{b_p}}{a_p}\widetilde{\mu}_i^{(1)}+1}{\widetilde{\mu}_i^{(1)}-z}\Biggr)\bigl(\widetilde{w}_i^{(1)}-1\bigr).\label{eq:h_tilde}
	\end{align}
	We define the set of events 
	\[
		\Omega_0 = \biggl\{ \bigl|\Expe\, m_n-m_n+\widetilde{h}_1\bigr| \geqslant \frac{1}{2}\sqrt{\frac{p}{n}} \biggr\},	
	\]
	then the inequality 
	\[
		\biggl|\biggl(\sqrt{\frac{p}{n}}\frac{a_p}{\sqrt{b_p}}+z+\Expe\, m_n\biggr)\biggl(\sqrt{\frac{p}{n}}\frac{a_p}{\sqrt{b_p}}+z+m_n-\widetilde{h}_1\biggr)\biggr| \geqslant K\frac{p}{n}
	\]
	holds on $\Omega_0$. Thus we obtain
	\begin{align*}
		\Expe\,\biggl|& \widetilde{D}_{11}+\frac{1}{a_p\sqrt{p/(nb_p)}+z+\Expe\, m_n} \biggr|^2 \\
		&\leqslant \Expe\, \biggl| \frac{\Expe\, m_n-m_n+\widetilde{h}_1}{\bigl(a_p\sqrt{p/(nb_p)}+z+\Expe\, m_n\bigr)\bigl(a_p\sqrt{p/(nb_p)}+z+m_n-\widetilde{h}_1\bigr)} \biggr|^2\\
		&\leqslant K\biggl[\biggl(\frac{n}{p}\biggr)^2\cdot\Prob(\Omega_0^{\mathsf{c}})+\frac{n}{p}\cdot\Prob(\Omega_0)\biggr]\cdot \Expe\,\bigl|\Expe\, m_n-m_n+\widetilde{h}_1\bigr|^2,
	\end{align*}
	where we use the inequality
	\begin{equation}\label{eq:denominator_lower_bbd}
		\biggl|\biggl(\sqrt{\frac{p}{n}}\frac{a_p}{\sqrt{b_p}}+z+\Expe\, m_n\biggr)\biggl(\sqrt{\frac{p}{n}}\frac{a_p}{\sqrt{b_p}}+z+m_n-\widetilde{h}_1\biggr)\biggr| \geqslant K\sqrt{\frac{p}{n}},
	\end{equation}
	that holds on the full set $\Omega$. The inequality \eqref{eq:denominator_lower_bbd} follows from the facts
	\begin{align*}
		\;&\sqrt{\frac{p}{n}}\frac{a_p}{\sqrt{b_p}}+z+m_n-\widetilde{h}_1 \\
		= \;& \sqrt{\frac{p}{n}}\frac{a_p}{\sqrt{b_p}}+z-\widetilde{\br}_1'\widetilde{\br}_1+\widetilde{\br}_1'\bB_1'\bB_1\biggl(\bB_1'\bB_1-\sqrt{\tfrac{p}{n}}\tfrac{a_p}{\sqrt{b_p}}\bI_{n}\biggr)^{-1}\widetilde{\br}_1\\
		= \;& \sqrt{\frac{p}{n}}\frac{a_p}{\sqrt{b_p}}+z+\widetilde{\br}_1'\biggl[\bB_1'\bB_1\biggl(\bB_1'\bB_1-\sqrt{\tfrac{p}{n}}\tfrac{a_p}{\sqrt{b_p}}\bI_{n}\biggr)^{-1}-\bI_n\biggr]\widetilde{\br}_1\\
		= \; & \sqrt{\frac{p}{n}}\frac{a_p}{\sqrt{b_p}}+z + \frac{1}{\sqrt{np}}\frac{\sqrt{b_p}}{a_p}\sum_{i=1}^{n}\Biggl(\frac{\widetilde{\mu}_i^{(1)}+\sqrt{\tfrac{p}{n}}\tfrac{a_p}{\sqrt{b_p}}}{\widetilde{\mu}_i^{(1)}-z}-1\Biggr)\widetilde{w}_i^{(1)}\\
		= \; & \biggl(\sqrt{\frac{p}{n}}\frac{a_p}{\sqrt{b_p}}+z\biggr) \biggl[1+\frac{1}{\sqrt{np}}\frac{\sqrt{b_p}}{a_p}\sum_{i=1}^{n}\frac{\widetilde{w}_i^{(1)}}{\widetilde{\mu}_i^{(1)}-z}\biggr]\\
		=:\;&  \biggl(\sqrt{\frac{p}{n}}\frac{a_p}{\sqrt{b_p}}+z\biggr)(1+S),
	\end{align*}
	and \[
		|1+S|\geqslant K\sqrt{\frac{n}{p}}.	
	\]
	We now proceed to complete the proof of \eqref{eq:G11_tilde_upper_bbd}. Note that we have
	\[
		\Prob(\Omega_0)\leqslant \frac{4n}{p}\Expe\,\bigl|\Expe\, m_n-m_n+\widetilde{h}_1\bigr|^2, 
	\]
	thus it is sufficient to prove that
	\begin{equation}\label{eq:G11_tilde_upper_bbd_esti_1}
		\Expe\,\bigl|\Expe\, m_n-m_n+\widetilde{h}_1\bigr|^2 = 	O\biggl(\frac{1}{n}\biggr) + O\biggl(\frac{n}{p}\biggr).
	\end{equation}
	Applying \eqref{eq:h_tilde} gives us
	\begin{align}
		\Expe\,\bigr| &\Expe\, m_n-m_n+\widetilde{h}_1\bigr|^2\nonumber  \\
		& \leqslant K \Biggl\{ \Expe\, \bigl|\Expe\, m_n-m_n\bigr|^2 + \Expe\, \Biggl|m_n - \Biggl(\frac{1}{n}\sum_{i=1}^{n-1}\frac{1}{\widetilde{\mu}_i^{(1)}-z}\Biggr)-\sqrt{\frac{n}{p}}\frac{\sqrt{b_p}}{a_p}\Biggl(\frac{1}{n}\sum_{i=1}^{n-1}\frac{\widetilde{\mu}_i^{(1)}}{\widetilde{\mu}_i^{(1)}-z}\Biggr)\Biggr|^2\nonumber\\
		&\qquad\qquad + \Expe\, \Biggl|\frac{1}{n}\sum_{i=1}^{n-1}\Biggl(\frac{\sqrt{\tfrac{n}{p}}\tfrac{\sqrt{b_p}}{a_p}\widetilde{\mu}_i^{(1)}+1}{\widetilde{\mu}_i^{(1)}-z}\Biggr)\bigl(\widetilde{w}_i^{(1)}-1\bigr)\Biggr|^2 + \Expe\, \bigl(\widetilde{\br}_1'\widetilde{\br}_1\bigr)^2\Biggr\}.\label{eq:G11_tilde_upper_bbd_esti}
	\end{align}
	Combining the similar method used for \eqref{eq:G11_upper_bbd} with \eqref{eq:G11_tilde_upper_bbd_esti} and the fact
	\begin{align*}
		\Expe\,\bigl(\widetilde{\br}_1' \widetilde{\br}_1\bigr)^2  & = \Expe\,\biggl[\sum_{j=1}^n\Bigl(\sum_{i=1}^p\hat{\sigma}_{1i}Y_{ij}\Bigr)^2\biggr]^2\\
		&=\Expe\,\biggl[\sum_{j=1}^n\Bigl(\sum_{i=1}^p\hat{\sigma}_{1i}Y_{ij}\Bigr)^4 + \sum_{j_1\neq j_2}\Bigl(\sum_{i=1}^p\hat{\sigma}_{1i}Y_{ij_1}\Bigr)^2\Bigl(\sum_{i=1}^p\hat{\sigma}_{1i}Y_{ij_2}\Bigr)^2 \biggr] = O\biggl(\frac{n}{p}\biggr),
	\end{align*}
	we obtain \eqref{eq:G11_tilde_upper_bbd_esti_1}.
\end{proof}

The following two lemmas about the derivatives of some quantities with respect to $Y_{jk}$, which can be used to obtain the derivatives of $F_{jk}$ (Lemma~\ref{lem:Fjk_derivative}) and $\widetilde{F}_{jk}$ (Lemma~\ref{lem:Fjk_tilde_derivative}).

Recall that
\[
	\bE:= \bSigma_p\bY\bD\bY'\bSigma_p = (E_{ij})_{p\times p}, \qquad\bF:= \bSigma_p\bY\bD = (F_{ij})_{p\times n}.
\]

\begin{lemma}\label{lem:derivative}
	For any $\alpha, j\in\{1,2,\ldots,p\}$ and $\beta, k\in \{1,2,\ldots,n\}$, we have
	\[
		\begin{aligned}
				\frac{\partial D_{\alpha\beta}}{\partial Y_{jk}} &= -F_{j\alpha}D_{\beta k}-F_{j\beta}D_{\alpha k};\\[0.5em]
				\frac{\partial  F_{\alpha\beta}}{\partial Y_{jk}} &= \sigma_{\alpha j} D_{k\beta} - E_{j\alpha}D_{\beta k}-F_{j\beta}F_{\alpha k};\\[0.5em]
				\frac{\partial (E_{jj}D_{kk})}{\partial Y_{jk}} &= 2\sigma_{j j} F_{jk}D_{kk} - 4 E_{jj}F_{j k} D_{kk}.
		\end{aligned}
	\]
\end{lemma}

\begin{proof}
	\begin{table}[!htbp]
		\centering
		\caption{Derivatives of $(Y_{rs}Y_{\ell t})$ w.r.t. $Y_{jk}$}
		\label{tab:derivative}
		\begin{tabular}{@{}ccccc@{}}
			\toprule
			$\frac{\partial \bigl(Y_{rs}Y_{\ell t}\bigr)}{\partial Y_{jk}}$ & $r=\ell=j$ & $r\neq j, \ell \neq j$ & $r=j, \ell\neq j$ & $r\neq j, \ell = j$ \\ \midrule
			$s=t=k$                                             & $2Y_{jk}$  & $0$                    & $Y_{\ell k}$      & $Y_{rk}$               \\
			$s\neq k, t\neq k$                                  & $0$        & $0$                    & $0$               & $0$                    \\
			$s=k, t\neq k$                                      & $Y_{jt}$   & $0$                    & $Y_{\ell t}$      & $0$                    \\
			$s\neq k, t=k$                                      & $Y_{js}$   & $0$                    & $0$               & $Y_{rs}$               \\ \bottomrule
		\end{tabular}
	\end{table}
	\begin{enumerate}
		\item[(1)] By using chain rule and the results in the Table~\ref{tab:derivative}, we have
		\begin{align*}
			\frac{\partial D_{\alpha\beta}}{\partial Y_{jk}} & = \sum_{1\leqslant s\leqslant t\leqslant p}\frac{\partial D_{\alpha\beta}}{\partial A_{st}}\cdot \frac{\partial A_{st}}{\partial Y_{jk}}\qquad \biggl[ \frac{\partial A_{st}}{\partial Y_{jk}} := \frac{\partial (\bY'\bSigma_p\bY)_{st}}{\partial Y_{jk}} \biggr]\\
			& = \sum_{s =1}^p\frac{\partial D_{\alpha\beta}}{\partial A_{ss}}\cdot \frac{\partial A_{ss}}{\partial Y_{jk}} + \sum_{1\leqslant s < t  \leqslant p}\frac{\partial D_{\alpha\beta}}{\partial A_{st}}\cdot \frac{\partial A_{st}}{\partial Y_{jk}}\\
			& = \sum_{s=1}^p\Bigl( -D_{\alpha s}D_{t\beta} \Bigr)\cdot \sum_{r,\ell } \biggl(\sigma_{r\ell}\frac{\partial \bigl(Y_{rs}Y_{\ell s}\bigr)}{\partial Y_{jk}} \biggr)	+ \sum_{s<t}\Bigl( -D_{\alpha s}D_{t\beta}-D_{\alpha t}D_{s\beta} \Bigr)\cdot \sum_{r,\ell } \biggl(\sigma_{r\ell}\frac{\partial \bigl(Y_{rs}Y_{\ell t}\bigr)}{\partial Y_{jk}}\biggr)\\
			& = \Bigl( -D_{\alpha k}D_{k\beta} \Bigr)\cdot  \biggl(2\sigma_{jj}Y_{jk} + \sum_{\ell \neq j} \sigma_{j\ell}Y_{\ell k} +\sum_{r\neq j} \sigma_{rj}Y_{rk} \biggr)\\
			&\qquad + \sum_{k< t}\Bigl( -D_{\alpha k}D_{t\beta}-D_{\alpha t}D_{k\beta} \Bigr)\cdot \biggl(\sigma_{jj} Y_{jt} +\sum_{\ell\neq j}\sigma_{j\ell} Y_{\ell t}\biggr)\\
			&\qquad + \sum_{s <k}\Bigl( -D_{\alpha s}D_{k\beta}-D_{\alpha k}D_{s\beta} \Bigr)\cdot \biggl(\sigma_{jj} Y_{js} +\sum_{r \neq j}\sigma_{r j} Y_{rs}\biggr)\\
			& =  \sum_{s=1}^p  \Bigl( -D_{\alpha s}D_{k\beta}-D_{\alpha k}D_{s\beta} \Bigr) \Bigl(\sum_{r=1}^p \sigma_{rj} Y_{rs}\Bigr)\\
			& =  \sum_{s, r} \biggl[-\Bigl(\sigma_{jr}Y_{rs}D_{s\alpha}\Bigr)D_{\beta k}-\Bigl(\sigma_{jr}Y_{rs}D_{s\beta}\Bigr)D_{\alpha k}\biggr]\\
			& = -F_{j\alpha}D_{\beta k}-F_{j\beta}D_{\alpha k},
		\end{align*}
		where the third equality follows from the formula (II. 18) in \citet{khorunzhy1996asymptotic};
		\item[(2)]
		\begin{align*}
			\frac{\partial  F_{\alpha\beta}}{\partial Y_{jk}} &  = \frac{\partial }{\partial Y_{jk}} \sum_{s,t}\Bigl(\sigma_{\alpha s} Y_{st}D_{t\beta}\Bigr)
			= \sum_{s, t } \sigma_{\alpha s} \biggl( \frac{\partial Y_{st}}{\partial Y_{jk}}\cdot D_{t\beta} + Y_{st} \cdot \frac{\partial D_{t\beta}}{\partial Y_{jk}} \biggr) \\
			& = \sigma_{\alpha j} D_{k\beta } -\sum_{s, t} \sigma_{\alpha s}Y_{st}  \biggl(F_{jt}D_{\beta k}+F_{j\beta}D_{t k}\biggr)\\
			& = \sigma_{\alpha j} D_{k\beta} - E_{j\alpha}D_{\beta k}-F_{j\beta}F_{\alpha k};
		\end{align*}
		\item[(3)]
		\begin{align*}
			\frac{\partial E_{jj} }{\partial Y_{jk}}
			& = \frac{\partial }{\partial Y_{jk}} \sum_{r}  \bigl(\bSigma_p\bY\bD\bigr)_{jr}\bigl(\bY'\bSigma_p\bigr)_{rj}\\
			& = \sum_{r}  \frac{\partial F_{jr}}{\partial Y_{jk}}\cdot \bigl(\bY'\bSigma_p\bigr)_{rj} + \sum_{r}F_{jr}\cdot \frac{\partial \bigl(\bY'\bSigma_p\bigr)_{rj}}{\partial Y_{jk}}\\
			& = \sum_{r}\Bigl( \sigma_{jj}D_{kr} - E_{jj} D_{rk} - F_{jr}F_{jk} \Bigr) \cdot \bigl(\bY'\bSigma_p\bigr)_{rj} + \sigma_{jj}F_{jk}\\
			& = 2\sigma_{j j} F_{jk} -2 E_{jj}F_{j k},
		\end{align*}
		\begin{align*}
			\frac{\partial (E_{jj}D_{kk})}{\partial Y_{jk}} & = \frac{\partial E_{jj}}{\partial Y_{jk}} \cdot  D_{kk} + \frac{\partial D_{kk}}{\partial Y_{jk}} \cdot E_{jj} \\
			& = \Bigl(2\sigma_{j j} F_{jk} -2 E_{jj}F_{j k}\Bigr)\cdot D_{kk} - 2F_{jk}D_{kk}\cdot E_{jj}\\
			& = 2\sigma_{j j} F_{jk}D_{kk} - 4 E_{jj}F_{j k} D_{kk}.\qedhere
		\end{align*}
	\end{enumerate}
\end{proof}
\noindent Recall that $\bSigma_p^2 = (\widetilde{\sigma}_{ij})$, and 
\[
	\widetilde{\bE}:= \bSigma_p^2\bY\bD\bY'\bSigma_p^2, \qquad \widehat{\bE}=\bSigma_p\bY\bD\bY'\bSigma_p^2,\qquad\widetilde{\bF}:= \bSigma_p^2\bY\bD.
\]
\begin{lemma}\label{lem:derivative-Ftilde}
	For any $\alpha, j\in\{1,2,\ldots,p\}$ and $\beta, k\in \{1,2,\ldots,n\}$, we have
	\[
		\begin{aligned}
				\frac{\partial  \widetilde{F}_{\alpha\beta}}{\partial Y_{jk}} & = \widetilde{\sigma}_{\alpha j} D_{k\beta} - \widehat{E}_{j \alpha}D_{\beta k}-F_{j\beta}\widetilde{F}_{\alpha k};\\[0.5em]
				\frac{\partial (\widehat{E}_{jj}D_{kk})}{\partial Y_{jk}} & = \sigma_{jj}\widetilde{F}_{jk}D_{kk} + \widetilde{\sigma}_{jj}F_{jk}D_{kk} - E_{jj}\widetilde{F}_{jk}D_{kk}-3\widehat{E}_{jj}F_{jk}D_{kk}.
		\end{aligned}
	\]
\end{lemma}

\begin{proof}
 \begin{align*}
			\frac{\partial  \widetilde{F}_{\alpha\beta}}{\partial Y_{jk}} & = \frac{\partial }{\partial Y_{jk}} \sum_{s,t}\Bigl(\widetilde{\sigma}_{\alpha s} Y_{st}D_{t\beta}\Bigr)
			= \sum_{s, t } \widetilde{\sigma}_{\alpha s} \biggl( \frac{\partial Y_{st}}{\partial Y_{jk}}\cdot D_{t\beta} + Y_{st} \cdot \frac{\partial D_{t\beta}}{\partial Y_{jk}} \biggr) \\
			& = \widetilde{\sigma}_{\alpha j} D_{k\beta } -\sum_{s, t} \widetilde{\sigma}_{\alpha s}Y_{st}  \biggl(F_{jt}D_{\beta k}+F_{j\beta}D_{t k}\biggr)\\
			& = \widetilde{\sigma}_{\alpha j} D_{k\beta} - \widehat{E}_{j \alpha}D_{\beta k}-F_{j\beta}\widetilde{F}_{\alpha k};
		\end{align*}
	 \begin{align*}
			\frac{\partial E_{jr}}{\partial Y_{jk}} & = \frac{\partial }{\partial Y_{jk}}\sum_{\ell} F_{j\ell}\Bigl(\bY' \bSigma_p\Bigr)_{\ell r} = \sum_{\ell} \frac{\partial F_{j\ell}}{\partial Y_{jk}}\cdot\Bigl(\bY' \bSigma_p\Bigr)_{\ell r} + \sum_{\ell} F_{j\ell} \cdot \frac{\partial \bigl(\bY' \bSigma_p\bigr)_{\ell r}}{\partial Y_{jk}} \\
			& = \sum_{\ell} \Bigl( \sigma_{jj}D_{k\ell} -E_{jj}D_{\ell k} -F_{j\ell}F_{jk} \Bigr)\cdot\Bigl(\bY' \bSigma_p\Bigr)_{\ell r} + \sigma_{jr}F_{jk}\\
			& = \sigma_{jj}F_{rk} + \sigma_{jr}F_{jk} -E_{jj}F_{rk} - F_{jk}E_{jr},
		\end{align*}
		\begin{align*}
			\frac{\partial (\widehat{E}_{jj}D_{kk})}{\partial Y_{jk}} & = \biggl(\frac{\partial }{\partial Y_{jk}} \sum_r E_{jr}\sigma_{rj}\biggr)\cdot D_{kk} + \widehat{E}_{jj}\cdot \Bigl(-2F_{jk}D_{kk}\Bigr)\\
			& = D_{kk}\sum_r \sigma_{rj}\biggl(\sigma_{jj}F_{rk} + \sigma_{jr}F_{jk} -E_{jj}F_{rk} - F_{jk}E_{jr}\biggr)  -2\widehat{E}_{jj}F_{jk}D_{kk}\\
			& = \sigma_{jj}\widetilde{F}_{jk}D_{kk} + \widetilde{\sigma}_{jj}F_{jk}D_{kk} - E_{jj}\widetilde{F}_{jk}D_{kk}-3\widehat{E}_{jj}F_{jk}D_{kk}.\qedhere
		\end{align*}
\end{proof}

\section{Proofs in applications}

\noindent
\textbf{Proof of Theorem~\ref{thm:W_limit_dist_H1}.}

\begin{proof}
	For notational simplicity, we denote $a_p=\tr(\bSigma_p)/p$ and $b_p=\tr(\bSigma_p^2)/p$. Let 
	$
	\widetilde{\bA}_n = \frac{1}{\sqrt{npb_p}}(\bY'\bY-pa_p\bI_n)= \frac{1}{\sqrt{npb_p}}(\bX'\bSigma_p\bX-pa_p\bI_n).
	$
	By some elementary calculations, we obtain two identities: 
	\begin{gather*}
		\tr(\bS_n)  = \sqrt{\frac{pb_p}{n}}\tr(\widetilde{\bA}_n)+pa_p,
		\qquad
		\tr(\bS_n^2)  =  \frac{pb_p}{n}\tr(\widetilde{\bA}_n^2) + \frac{2pa_p}{n}\sqrt{\frac{pb_p}{n}}\tr(\widetilde{\bA}_n) + \frac{(pa_p)^2}{n}.
	\end{gather*}
	Then $W$ can be written as 
	\[
	W = \frac{b_p}{n}\tr(\widetilde{\bA}_n^2) - \frac{2}{p}\sqrt{\frac{pb_p}{n}}\tr(\widetilde{\bA}_n) -\frac{b_p}{n^2} \bigl[\tr(\widetilde{\bA}_n)\bigr]^2 +\frac{p}{n}-2a_p+1.
	\]
	\citet{li2016testing} derived the limiting joint distribution of $\bigl(\tr(\widetilde{\bA}_n^2)/n, \tr(\widetilde{\bA}_n)/n\bigr)$ (see their Lemma 3.1) as follows:
	\begin{equation}\label{eq:joint_dist_tr_S}
		n\Biggl(\begin{matrix}
			\frac{1}{n}\tr(\widetilde{\bA}_n^2) - 1 - \frac{1}{n}\bigl( \frac{\omega}{\theta}(\nu_4-3)+1 \bigr)\\[0.5em]
			\frac{1}{n}\tr(\widetilde{\bA}_n)
		\end{matrix}\Biggr)	
		\convd \calN \Biggl(\Biggl(\begin{matrix}
			0\\[0.5em]0
		\end{matrix}\Biggr), \Biggl(\begin{matrix}
			4 & 0\\[0.5em]
			0 & \frac{\omega}{\theta}(\nu_4-3)+2
		\end{matrix}\Biggr)\Biggr).
	\end{equation}
	Define the function 
	\[
	g(x,y)  = b_px - \frac{2n}{p}\sqrt{\frac{pb_p}{n}}y-b_py^2+\frac{p}{n}-2a_p+1,
	\]
	then $W = g\bigl(\tr(\widetilde{\bA}_n^2)/n, \tr(\widetilde{\bA}_n)/n \bigr)$, we have
	\begin{align*}
		& \frac{\partial g}{\partial x} \Bigl( 1 + \frac{1}{n}\Bigl( \frac{\omega}{\theta}(\nu_4-3)+1 \Bigr), 0 \Bigr) = b_p,\\
		& \frac{\partial g}{\partial y} \Bigl( 1 + \frac{1}{n}\Bigl( \frac{\omega}{\theta}(\nu_4-3)+1 \Bigr), 0 \Bigr) =-\frac{2n}{p}\sqrt{\frac{pb_p}{n}},\\
		& g\Bigl( 1 + \frac{1}{n}\Bigl( \frac{\omega}{\theta}(\nu_4-3)+1 \Bigr), 0 \Bigr) = b_p+\frac{b_p}{n}\Bigl( \frac{\omega}{\theta}(\nu_4-3)+1 \Bigr) +\frac{p}{n} - 2a_p+1.
	\end{align*}
	By \eqref{eq:joint_dist_tr_S}, we have
	\[
	n\biggl(W-g\Bigl( 1 + \frac{1}{T}\Bigl( \frac{\omega}{\theta}(\nu_4-3)+1 \Bigr), 0 \Bigr)\biggr) \convd\calN(0,\lim A),	
	\]
	where 
	\[
	A = \Biggl(
	\begin{matrix}
		\frac{\partial g}{\partial x} ( 1 + \frac{1}{n}( \frac{\omega}{\theta}(\nu_4-3)+1 ), 0 )\\
		\frac{\partial g}{\partial x} ( 1 + \frac{1}{n}( \frac{\omega}{\theta}(\nu_4-3)+1 ), 0 )
	\end{matrix}
	\Biggr)' 	
	\Biggl(\begin{matrix}
		4 & 0\\
		0 & \frac{\omega}{\theta}(\nu_4-3)+2
	\end{matrix}\Biggr)
	\Biggl(
	\begin{matrix}
		\frac{\partial g}{\partial x} ( 1 + \frac{1}{n}( \frac{\omega}{\theta}(\nu_4-3)+1 ), 0 )\\
		\frac{\partial g}{\partial x} ( 1 + \frac{1}{n}( \frac{\omega}{\theta}(\nu_4-3)+1 ), 0 )
	\end{matrix}
	\Biggr) \to 4\theta^2.
	\]
	Thus,
	\[
	n\Bigl(W-b_p-\frac{b_p}{n}\Bigl( \frac{\omega}{\theta}(\nu_4-3)+1 \Bigr) + 2a_p-1-\frac{p}{n} \Bigr) \convd \calN(0, 4\theta^2),
	\]
	that is,
	\[
	nW - p - \theta\Bigl( \frac{\omega}{\theta}(\nu_4-3) +1 \Bigr) +n(2\gamma-1-\theta) \convd \calN(0,4\theta^2).\qedhere
	\]
\end{proof}

\noindent \textbf{Proof of Proposition~\ref{prop:power_theo_W}.}

\begin{proof}
	For the test based on statistic $W$, by Theorem \ref{thm:W_limit_dist_H0} and \ref{thm:W_limit_dist_H1}, we have
	\begin{align*}
		\beta(H_1) & = \Prob\biggl( \frac{1}{2}\Bigl(nW-p-(\nu_4-2)\Bigr) \geqslant z_{\alpha} \;\Big|\; H_1 \biggr) \\ 
		& = \Prob\biggl( nW-p - \theta\Bigl( \frac{\omega}{\theta}(\nu_4-3) +1 \Bigr) +n(2\gamma-1-\theta) \\
		&\qquad\qquad   \geqslant 2z_{\alpha} - \theta\Bigl( \frac{\omega}{\theta}(\nu_4-3) +1 \Bigr) +n(2\gamma-1-\theta) + (\nu_4-2) \;\Big|\; H_1 \biggr)\\
		& = 1-\Phi\biggl( \frac{1}{2\theta} \Bigl\{ 2z_{\alpha} -  \omega(\nu_4-3) -\theta +n(2\gamma-1-\theta) + (\nu_4-2) \Bigr\} \biggr),
	\end{align*}
	since $2\gamma-1\leqslant \gamma^2\leqslant \theta$, Proposition \ref{prop:power_theo_W} follows.
\end{proof}

\bibliography{reference}

\begin{thebibliography}{22}
\providecommand{\natexlab}[1]{#1}
\providecommand{\url}[1]{\texttt{#1}}
\expandafter\ifx\csname urlstyle\endcsname\relax
  \providecommand{\doi}[1]{doi: #1}\else
  \providecommand{\doi}{doi: \begingroup \urlstyle{rm}\Url}\fi

\bibitem[Bai and Silverstein(2004)]{bai2004clt}
Zhidong Bai and Jack~W Silverstein.
\newblock C{LT} for linear spectral statistics of large-dimensional sample
  covariance matrices.
\newblock \emph{The Annals of Probability}, 32:\penalty0 553--605, 2004.

\bibitem[Bai and Silverstein(2010{\natexlab{a}})]{BSbook}
Zhidong Bai and Jack~W Silverstein.
\newblock \emph{Spectral analysis of large dimensional random matrices}.
\newblock Springer, 2nd edition, 2010{\natexlab{a}}.

\bibitem[Bai and Silverstein(2010{\natexlab{b}})]{bai2010spectral}
Zhidong Bai and Jack~W Silverstein.
\newblock \emph{Spectral analysis of large dimensional random matrices},
  volume~20.
\newblock Springer, 2010{\natexlab{b}}.

\bibitem[Bai and Yao(2005)]{BaiYao2005}
Zhidong Bai and Jianfeng Yao.
\newblock On the convergence of the spectral empirical process of {W}igner
  matrices.
\newblock \emph{Bernoulli}, 11\penalty0 (6):\penalty0 1059--1092, 2005.

\bibitem[Bai and Yin(1988)]{bai1988convergence}
Zhidong Bai and Yongquan Yin.
\newblock Convergence to the semicircle law.
\newblock \emph{The Annals of Probability}, pages 863--875, 1988.

\bibitem[Bai et~al.(2009)Bai, Jiang, Yao, and Zheng]{Bai2009Correction}
Zhidong Bai, Dandan Jiang, Jian-Feng Yao, and Shurong Zheng.
\newblock {Corrections to LRT on large-dimensional covariance matrix by RMT}.
\newblock \emph{The Annals of Statistics}, 37\penalty0 (6B):\penalty0 3822 --
  3840, 2009.

\bibitem[Bao(2015)]{bao2015asymptotic}
Zhigang Bao.
\newblock On asymptotic expansion and central limit theorem of linear
  eigenvalue statistics for sample covariance matrices when $ {N}/{M} \to 0$.
\newblock \emph{Theory of Probability \& Its Applications}, 59\penalty0
  (2):\penalty0 185--207, 2015.

\bibitem[Billingsley(1968)]{billingsley1968convergence}
Patrick Billingsley.
\newblock \emph{Convergence of probability measures}.
\newblock New York: Wiley, 1968.

\bibitem[Billingsley(2008)]{billingsley2008probability}
Patrick Billingsley.
\newblock \emph{Probability and measure}.
\newblock John Wiley \& Sons, 2008.

\bibitem[Burkholder(1973)]{Burkholder1973Distribution}
Donald~L. Burkholder.
\newblock {Distribution Function Inequalities for Martingales}.
\newblock \emph{The Annals of Probability}, 1\penalty0 (1):\penalty0 19--42,
  1973.

\bibitem[Chen and Pan(2012)]{chen2012convergence}
Binbin Chen and Guangming Pan.
\newblock Convergence of the largest eigenvalue of normalized sample covariance
  matrices when $ p $ and $ n $ both tend to infinity with their ratio
  converging to zero.
\newblock \emph{Bernoulli}, 18\penalty0 (4):\penalty0 1405--1420, 2012.

\bibitem[Chen and Pan(2015)]{chen2015clt}
Binbin Chen and Guangming Pan.
\newblock C{LT} for linear spectral statistics of normalized sample covariance
  matrices with the dimension much larger than the sample size.
\newblock \emph{Bernoulli}, 21\penalty0 (2):\penalty0 1089--1133, 2015.

\bibitem[Chen et~al.(2021)Chen, Xiao, and Yang]{chen2021autoregressive}
Rong Chen, Han Xiao, and Dan Yang.
\newblock Autoregressive models for matrix-valued time series.
\newblock \emph{Journal of Econometrics}, 222\penalty0 (1):\penalty0 539--560,
  2021.

\bibitem[Couillet and Debbah(2011)]{couillet2011random}
Romain Couillet and Merouane Debbah.
\newblock \emph{Random matrix methods for wireless communications}.
\newblock Cambridge University Press, 2011.

\bibitem[Khorunzhy et~al.(1996)Khorunzhy, Khoruzhenko, and
  Pastur]{khorunzhy1996asymptotic}
Alexei~M Khorunzhy, Boris~A Khoruzhenko, and Leonid~A Pastur.
\newblock Asymptotic properties of large random matrices with independent
  entries.
\newblock \emph{Journal of Mathematical Physics}, 37\penalty0 (10):\penalty0
  5033--5060, 1996.

\bibitem[Ledoit and Wolf(2002)]{LedoitWolf2002}
Olivier Ledoit and Michael Wolf.
\newblock {Some hypothesis tests for the covariance matrix when the dimension
  is large compared to the sample size}.
\newblock \emph{The Annals of Statistics}, 30\penalty0 (4):\penalty0 1081 --
  1102, 2002.

\bibitem[Li and Yao(2016)]{li2016testing}
Zeng Li and Jianfeng Yao.
\newblock Testing the sphericity of a covariance matrix when the dimension is
  much larger than the sample size.
\newblock \emph{Electronic Journal of Statistics}, 10\penalty0 (2):\penalty0
  2973--3010, 2016.

\bibitem[Nagao(1973)]{Nagao1973On}
Hisao Nagao.
\newblock {On Some Test Criteria for Covariance Matrix}.
\newblock \emph{The Annals of Statistics}, 1\penalty0 (4):\penalty0 700 -- 709,
  1973.

\bibitem[Pan and Zhou(2011)]{pan2011central}
Guangming Pan and Wang Zhou.
\newblock Central limit theorem for {H}otelling's ${T}^2$ statistic under large
  dimension.
\newblock \emph{The Annals of Applied Probability}, pages 1860--1910, 2011.

\bibitem[Wang and Paul(2014)]{wang2014limiting}
Lili Wang and Debashis Paul.
\newblock Limiting spectral distribution of renormalized separable sample
  covariance matrices when $p/n\to 0$.
\newblock \emph{Journal of Multivariate Analysis}, 126:\penalty0 25--52, 2014.

\bibitem[Wang and Yao(2013)]{wang2013sphericity}
Qinwen Wang and Jianfeng Yao.
\newblock On the sphericity test with large-dimensional observations.
\newblock \emph{Electronic Journal of Statistics}, 7:\penalty0 2164--2192,
  2013.

\bibitem[Yao et~al.(2015)Yao, Zheng, and Bai]{yao2015sample}
Jianfeng Yao, Shurong Zheng, and ZD~Bai.
\newblock \emph{Sample covariance matrices and high-dimensional data analysis},
  volume~2.
\newblock Cambridge University Press, 2015.

\end{thebibliography}

\end{document}